\documentclass[a4paper]{article}
\pdfoutput=1
\usepackage{amsthm,amssymb}
\usepackage{mathptmx}
\usepackage{framed}
\usepackage[authoryear]{natbib}
\usepackage{hyperref}%
\usepackage[all]{hypcap}
\hypersetup{
draft=false,
colorlinks=true,
linkcolor=black,
anchorcolor=black,
citecolor=black,
urlcolor=black,
filecolor=black,
breaklinks=true}
\usepackage{fullpage}
\usepackage{anyfontsize}
\DeclareMathAlphabet{\mathcal}{OMS}{cmsy}{m}{n}
\usepackage{graphicx}
\usepackage{xspace}
\usepackage[T1]{fontenc}
\usepackage{float}
\floatstyle{boxed}
\restylefloat{table}
\restylefloat{figure}
\usepackage{boldline}
\usepackage{amsmath}
\usepackage{amsfonts}
\allowdisplaybreaks
\usepackage{wrapfig}
\usepackage{proof}
\usepackage[braket]{qcircuit}
\usepackage{tikz}
\usetikzlibrary{cd}
\usepackage[only=llparenthesis,rrparenthesis,llbracket,rrbracket]{stmaryrd}
\SetSymbolFont{stmry}{bold}{U}{stmry}{m}{n}
\usepackage{multicol,multirow}

\newcommand\proba[1]{\mathtt{\{#1\}}}
\newcommand\ntype{\!\not\,:}

\newcommand\may[1][\alpha]{[{#1}.]}
\newcommand\LambdaS{Lambda-\ensuremath{\mathcal S}\xspace}
\newcommand\xrecap[4]{\noindent {\bf #1 \ref{#3} (#2).} \emph{#4}}
\newcommand\recap[3]{\noindent {\bf #1 \ref{#2}.} \emph{#3}}
\newcommand\citaCR{\cite[Thm.~7.25]{ArrighiDowekLMCS17}}
\newcommand\citaSR{\cite[Thm.~5.12]{DiazcaroDowekRinaldiBIO19}}
\newcommand\citaSN{\cite[Thm.~6.10]{DiazcaroDowekRinaldiBIO19}}
\newcommand\etaD{\hat\eta}
\newcommand\muD{\hat\mu}
\newcommand\paren[1]{(#1)}
\newcommand\home[2]{[#1,#2]}
\newcommand\lra{\longrightarrow}
\newcommand\sem[1]{\left\llbracket {#1}\right\rrbracket}
\newcommand\com[1]{\mathfrak{C}_{#1}}
\newcommand\B{\ensuremath{{\mathbb B}}}
\newcommand\types{\ensuremath{\mathcal T}}
\newcommand\qtypes{\ensuremath{\mathcal Q}}
\newcommand\btypes{\ensuremath{\mathcal B}}
\newcommand\vars{\ensuremath{\mathsf{Vars}}}
\newcommand\values{\ensuremath{\mathsf V}}
\newcommand\tbasis{\ensuremath{\mathsf B}}

\newcommand\D{\ensuremath{\mathsf D}}
\newcommand\s[1]{\ensuremath{\mathsf{#1}}}
\newcommand\head{\text{\sl head}}
\newcommand\tail{\text{\sl tail}}
\newcommand\pair[2]{({#1}+{#2})}
\newcommand\z[1][A]{\vec 0_{S{#1}}}
\newcommand\ite[3]{{#1}?{#2}\mathord{\cdot}{#3}}
\newcommand\C{\mathbb C}
\newcommand\One{1}
\newcommand\I{I}
\newcommand\Id{\mathsf{Id}}
\newcommand\xlra[1]{\xrightarrow{#1}}
\newcommand\tax{\ensuremath{\mathsf{Ax}}}
\newcommand\tif{\ensuremath{\mathsf{If}}}
\newcommand\rbetab{(\s{\beta_b})}
\newcommand\rbetan{(\s{\beta_n})}
\newcommand\riftrue{(\s{if_{1}})}
\newcommand\riffalse{(\s{if_{0}})}
\newcommand\rlinr{(\s{lin^+_r})}
\newcommand\rlinscalr{(\s{lin^\alpha_r})}
\newcommand\rlinzr{(\s{lin^0_r})}
\newcommand\rlinl{(\s{lin^+_\ell})}
\newcommand\rlinscall{(\s{lin^\alpha_\ell})}
\newcommand\rlinzl{(\s{lin^0_\ell})}
\newcommand\rneut{(\s{neutral})}
\newcommand\runit{(\s{unit})}
\newcommand\rzeros{(\s{zero_\alpha})}

\newcommand\rzero{(\s{zero})}
\newcommand\rprod{(\s{prod})}
\newcommand\rdists{(\s{\alpha dist})}
\newcommand\rdistcasum{(\s{dist^+_\Uparrow})}
\newcommand\rdistcascal{(\s{dist^\alpha_\Uparrow})}
\newcommand\rdistcazeror{(\s{dist^0_{\Uparrow_r}})}
\newcommand\rdistcazerol{(\s{dist^0_{\Uparrow_\ell}})}
\newcommand\rcaneutl{(\s{neut^\Uparrow_\ell})}
\newcommand\rcaneutr{(\s{neut^\Uparrow_r})}

\newcommand\rcaneutzl{(\s{neut^\Uparrow_{0\ell}})}
\newcommand\rcaneutzr{(\s{neut^\Uparrow_{0r}})}
\newcommand\rfact{(\s{fact})}
\newcommand\rfacto{(\s{fact^1})}
\newcommand\rfactt{(\s{fact^2})}
\newcommand\rproj{(\s{proj})}
\newcommand\rprojz{(\s{proj}\ensuremath{_{\vec 0}})}
\newcommand\rhead{(\s{head})}
\newcommand\rtail{(\s{tail})}
\newcommand\rdistzr{(\s{dist^0_r})}
\newcommand\rdistzl{(\s{dist^0_\ell})}
\newcommand\rdistscalr{(\s{dist^\alpha_r})}
\newcommand\rdistscall{(\s{dist^\alpha_\ell})}
\newcommand\rdistsumr{(\s{dist^+_r})}
\newcommand\rdistsuml{(\s{dist^+_\ell})}
\newcommand\rcomm{(\s{comm})}
\newcommand\rassoc{(\s{assoc})}

\newcommand\elimeq{\approx_e}
\newcommand\opeq{\approx_{op}}

\newtheorem{thm}{Theorem}[section]
\newtheorem{lem}[thm]{Lemma}
\newtheorem{prop}[thm]{Proposition}

\theoremstyle{definition}
\newtheorem{defin}[thm]{Definition}
\newtheorem{example}[thm]{Example}
\newtheorem{rmk}[thm]{Remark}
\newtheorem{rmks}[thm]{Remarks}

\begin{document}

\title{A concrete model for a typed linear algebraic lambda calculus\footnote{This paper is the long journal version of
  \citep{DiazcaroMalherbeLSFA18}. In the present paper, the main new result is
to revisit some rewrite rules in order to prove a theorem of adequacy.}}
\author{
  Alejandro D\'iaz-Caro\\[1ex]
  {\small CONICET-Universidad de Buenos Aires}\\
  {\small Instituto de Ciencias de la Computaci\'on (ICC)}\\
  {\small Buenos Aires, Argentina}\\[1ex]
  {\small Departamento de Ciencia y Tecnolog\'{\i}a}\\
  {\small Universidad Nacional de Quilmes}\\
  {\small Bernal, Buenos Aires, Argentina}
  \and
  Octavio Malherbe\\[1ex]
  {\small IMERL, Facultad de Ingenier\'{\i}a}\\
  {\small Universidad de la Rep\'ublica}\\
  {\small Montevideo, Uruguay}
}
\date{}
\maketitle

\begin{abstract}
  We give an adequate, concrete, categorical-based model for \LambdaS, which is
  a typed version of a linear-algebraic lambda calculus, extended with
  measurements. \LambdaS is an extension to first-order lambda calculus unifying
  two approaches of non-cloning in quantum lambda-calculi: to forbid duplication
  of variables, and to consider all lambda-terms as algebraic linear functions.
  The type system of \LambdaS have a superposition constructor $S$ such that a
  type $A$ is considered as the base of a vector space while $SA$ is its span. Our
  model considers $S$ as the composition of two functors in an adjunction relation
  between the category of sets and the category of vector spaces over $\C$. The
  right adjoint is a forgetful functor $U$, which is hidden in the language, and
  plays a central role in the computational reasoning.
\end{abstract}

%\paragraph*{keywords}
%    Quantum computing, algebraic lambda-calculus, categorical semantics

\section{Introduction}
The non-cloning property of quantum computing has been treated in different ways
in quantum programming languages. One way is to forbid duplication of variables
with linear types~\citep{GirardTCS87,AbramskyTCS93}, and hence, a program taking
a quantum argument will not duplicate it,
e.g.~\citep{AltenkirchGrattageLICS05,GreenLeFanulumsdaineRossSelingerValironPLDI13,PaganiSelingerValironPOPL14,ZorziMSCS16,SelingerValironMSCS06}.
Another way is to consider all lambda-terms as expressing linear functions, {in
what is known as linear-algebraic lambda-calculi},
e.g.~\citep{ArrighiDowekLMCS17,ArrighiDiazcaroValironIC17,DiazcaroPetitWoLLIC12,ArrighiDiazcaroLMCS12}.
The first approach forbids a term $\lambda x.(x\otimes x)$ (for some convenient
definition of $\otimes$), while the second approach distributes $(\lambda
x.(x\otimes x))(\ket 0+\ket 1)$ to $\lambda x.(x\otimes x)\ket 0+\lambda
x.(x\otimes x)\ket 1$, mimicking the way linear operations act on vectors
in a vector space. However, adding a measurement operator to a calculus
following the linear-algebraic approach needs to also add linear types: indeed,
if $\pi$ represents a measurement operator, $(\lambda x.\pi x)\pair{\ket 0}{\ket
  1}$ should not reduce to ${(\lambda x.\pi x)\ket 0}+{(\lambda x.\pi x)\ket 1}$
but to $\pi({\ket 0}+{\ket 1})$. Therefore, there must be functions taking
superpositions and functions distributing over them. However, the functions
taking a superposition have to be marked in some way, so to ensure that they
will not use their arguments more than once (i.e.~to ensure linearity in the
linear-logic sense).

{Lineal, the first linear-algebraic lambda-calculus, is an untyped calculus
introduced by \citet{ArrighiDowekLMCS17} to study the superposition of programs,
with quantum computing as a goal. However, Lineal is not a quantum calculus in
the sense that there is no construction allowing one to characterize which terms
can be directly compiled into a quantum machine.
Vectorial~\citep*{ArrighiDiazcaroValironIC17} has been the conclusion of a long
path to obtain a typed
Lineal~\citep{ArrighiDiazcaroLMCS12,DiazcaroPetitWoLLIC12,ArrighiDiazcaroValironIC17}.
In Vectorial, the type system gives information on whether the final term can be
considered or not as a quantum state (of norm $1$). Nevertheless, it fails to
establish whether typed programs can be considered quantum, in the sense of
implementing unitary transformations and measurements---in any case,
measurements are left out of the equation in these versions of typed Lineal.

The calculus \LambdaS\ is a start over, with a new type system not related to
Vectorial. {It is a} first-order fragment of Lineal, extended with
measurements. It has been introduced by~\citet*{DiazcaroDowekTPNC17} and
slightly modified later by~\citet*{DiazcaroDowekRinaldiBIO19}. Following this
line, }\citet{DiazcaroGuillermoMiquelValironLICS19} presented a calculus defined
through realizability techniques, {which validates this long line of
research on Lineal as a quantum calculus, by proving the terms which are
typable with certain types coincide with implementations of unitary operators.
In \citep{DiazcaroMalherbeACS20} we gave a categorical model of \LambdaS without measurements.
The object of the current paper is to set up the bases for a categorical model of
\LambdaS\ in full (with measurements), by defining a concrete model with a
categorical presentation, paving the way to an abstract construction in future
research.}

In linear logic, a type $A$ without decoration represents a type of a term
that cannot be duplicated, while ${!}A$ types duplicable terms. In \LambdaS
instead, $A$ are the types of the terms that cannot be superposed, while $SA$ are
the terms that can be superposed, and since superposition forbids duplication,
$A$ means that we can duplicate, while $SA$ means that we cannot duplicate. So
the $S$ is not the same as the bang ``$!$'', but somehow the opposite{, in
the sense that we mark the fragile terms (those that cannot be duplicated)}.
This can be explained by the fact that linear logic is focused on the
possibility of duplication, while here we focus on the possibility of
superposition, which implies the impossibility of duplication.

\citet*{DiazcaroDowekTPNC17} gave a first denotational semantics for \LambdaS\
(in environment style) where the atomic type $\B$ is interpreted as $\{\ket
0,\ket 1\}$ while $S\B $ is interpreted as $\mathsf{Span}(\{\ket 0,\ket
1\})=\C^2$, and, in general, a type $A$ is interpreted as a basis while $SA$ is
the vector space generated by such a basis. In this paper we go beyond and give
a categorical interpretation of \LambdaS where $S$ is a functor of an adjunction
between the category $\mathbf{Set}$ and the category $\mathbf{Vec}$. Explicitly,
when we evaluate $S$, we obtain formal finite linear combinations of elements of
a set with complex numbers as coefficients. The other functor of the adjunction,
$U$, allows us to forget the vectorial structure.

The main structural feature of our model is that it is expressive enough to
describe the bridge between the quantum and the classical universes explicitly
by controlling its interaction. This is achieved by providing a monoidal
adjunction. In the literature, intuitionistic linear (as in linear-logic) models
are obtained by a comonad determined by a monoidal adjunction $(S,m)\dashv
(U,n)$\footnote{{Where $m$ and $n$ are the mediating arrows given by the
monoidality of the adjunction.}}, i.e.~the bang $!$ is interpreted by the
comonad $SU$ (see~\citep{BentonCSL94}). In a different way, a crucial ingredient
of our model is to consider the monad $US$ for the interpretation of $S$
determined by a similar monoidal adjunction. This implies that, on the one hand,
we have a tight control of the Cartesian structure of the model
(i.e.~duplication, etc) and, on the other hand, the world of superpositions
lives inside the classical world, i.e.~determined externally by classical rules
until we decide to explore it. This is given by the following composition of
maps:
\[
  US\B \times US\B \xlra n U(S\B \otimes S\B )\xlra{Um} US(\B\times \B)
\]
that allows us to operate in a monoidal structure representing the quantum world
and then to return to the Cartesian product.

This is different from linear logic, where the classical world lives inside the
quantum world i.e.~$({!}\B)\otimes ({!}\B)$ is a product inside a monoidal
category.

Another source of inspiration for our model has been the work of
\citet*{SelingerQPL05} and \citet*{AbramskyCoeckeLICS04} where 
they formalized the concept of scalars
and inner product in a more abstract categorical setting,
i.e.~a category in which there is an abstract notion of a dagger functor. It is
envisaged that this approach will provide the basis for an abstract model in
future work.

The paper is structured as follows. In Section~\ref{sec:calculus} we recall the
definition of \LambdaS and give some examples, stating its main properties.
Section~\ref{sec:DenSem} is divided into three subsections: first we define the
categorical constructions needed to interpret the calculus, then we give the
interpretation, and finally we prove such an interpretation to be {adequate}. We
conclude in Section~\ref{sec:conclusion}. An appendix with the full proofs
follows the article.

\section{The calculus \texorpdfstring{\LambdaS}{Lambda-S}}\label{sec:calculus}

We give a slightly modified presentation of
\LambdaS~\citep{DiazcaroDowekRinaldiBIO19}. In particular, instead of giving a
probabilistic rewrite system, where $t\to_{{\mathtt{p_k}}}r_k$ means that $t$
reduces with probability ${\mathtt{p_k}}$ to $r_k$, we introduce the notation
$t\lra \proba{p_1}r_1\parallel\dots\parallel \proba{p_n}r_n$, {where
$\proba{p_1}r_1\parallel\dots\parallel \proba{p_n}r_n$} denotes a finite distribution. This way, the rewrite system is
deterministic and the probabilities are internalized.

The syntax of terms and types is given in Figure~\ref{fig:syntax}. We write
$\B^n$ for $\B\times\cdots\times\B$ $n$-times, with the convention that
$\B^1=\B$, and may write $\bigparallel_{i=1}^n \proba{p_i} t_i$, for $\proba{p_1}
t_1\parallel\cdots\parallel\proba{p_n}t_n$.
We use capital Latin letters
($A,B,C,\dots$) for general types and the capital Greek letters $\Psi$, $\Phi$, $\Xi$, and $\Upsilon$
for qubit types.  
$\qtypes$ is the set of qubit types, and $\types$ is the set of all the types
($\qtypes\subsetneq\types$). We write
$\btypes=\{\B^n\mid n\in\mathbb N\}\cup\{\Psi\Rightarrow A\mid \Psi\in\qtypes,
A\in\types\}$, that is, the set of non-superposed types.  In addition, $\vars$ is the
set of variables, $\tbasis$ is the set of basis terms, $\values$ the set of
values, $\Lambda$ the set of terms, and $\D$ the set of distributions
on terms. We have
$\vars\subsetneq\tbasis\subsetneq\values\subsetneq\Lambda\subsetneq\D$, where
the last inclusion is considering the constant function that associates
probability $1$ to any term.
As customary, we may write $x$ instead of $x^\Psi$ when the type is clear from the context. Notice that this language is in Church-style.

\begin{figure}[!h]
  \centering
  \[
    \begin{array}{rlr}
      \Psi & := \B^n\mid S\Psi\mid \Psi\times\Psi & \textrm{Qubit types (\qtypes)} \\
      A & := \Psi\mid \Psi\Rightarrow A\mid SA & \textrm{Types (\types)} \\                                                    
      \\
      b  & := x^\Psi\mid \lambda x{:}\Psi.t\mid \ket 0\mid \ket 1\mid \ite{}tt\mid b\times b & \textrm{Basis terms (\tbasis)} \\
      v  & := b\mid \pair vv\mid \z\mid \alpha.v\mid v\times v & \textrm{Values (\values)} \\
      t  & := v\mid tt\mid \pair tt\mid \pi_j t\mid \alpha.t\mid
           t\times t\mid \head~t\mid \tail~t\mid \Uparrow_r t\mid \Uparrow_\ell
           t & \textrm{Terms ($\Lambda$)}\\
      p & := \proba{p_1}t_1\parallel\cdots\parallel\proba{p_n}t_n & \textrm{Distributions ($\D$)}
    \end{array}
  \]
  where $\alpha\in\mathbb C$ and $\mathtt{p_i}\in[0,1]\subseteq\mathbb R$.
  \caption{Syntax of types and terms of \LambdaS.}
  \label{fig:syntax}
\end{figure}

The terms are considered modulo associativity and commutativity of the
syntactic symbol $+$.  On the other hand, the symbol $\parallel$ is used to
represent a true distribution over terms, not as a syntactic
symbol, and so it is not only associative and commutative, we also have that
$\proba pt\parallel\proba qt$ is the same as $\proba{p+q}t$, $\proba
pt\parallel\proba 0r =\proba pt$, and $\bigparallel_{i=1}^1 \proba 1t=\proba 1t$\footnote{As
  a remark, notice that $\parallel$ can be seen as the $+$ symbol of the
  algebraic lambda calculus~\citep{VauxMSCS09}, where the equality is confluent
  since scalars are positive, while the $+$ symbol in \LambdaS coincides with the $+$ from
Lineal~\citep{ArrighiDowekLMCS17}
(see~\citep{AssafDiazcaroPerdrixTassonValironLMCS14} for a more detailed
discussion on different presentations of algebraic lambda calculi).}. 

There is one atomic type $\B$, for basis qubits $\ket 0$ and $\ket 1$, and
three constructors: $\times$, for pairs, $\Rightarrow$, for first-order
functions, and $S$ for superpositions.

The syntax of terms contains:
\begin{itemize}
  \item The three terms for first-order lambda-calculus, namely,
    variables, abstractions and applications.
  \item Two basic terms $\ket 0$ and $\ket 1$ to represent qubits, and one test
    $\ite{}rs$ on them. We usually write $\ite trs$ for $(\ite{}rs)t$, see
    Example~\ref{ex:ite} for a clarification of why to choose this
    presentation.
  \item A product $\times$ to represent associative pairs (i.e.~lists), with
    its destructors $\head$ and $\tail$. We usually use the notations
    $\ket{b_1b_2\dots b_n}$ for
    $\ket{b_1}\times\ket{b_2}\times\dots\times\ket{b_n}$, $\ket b^n$ for
    $\ket{bb\cdots b}$ and $\prod_{i=1}^n t_i$ for $t_1\times\cdots\times t_n$.
  \item Constructors to write linear combinations of terms, namely $+$ (sum)
    and $.$ (scalar multiplication), and its destructor $\pi_j$ measuring the
    first $j$ qubits written as linear combinations of lists of qubits. Also, one
    null vector $\z$ for each type $SA$. We may write $-t$ for $-1{.}t$.
    The symbol $+$ is taken to be associative and commutative (that is, our
    terms are expressed modulo AC~\citep{ArrighiDowekLMCS17}), therefore, we may
    use the summation symbol $\sum$, with the convention that $\sum_{i=1}^1t=t$.
  \item Two casting functions $\Uparrow_r$ and $\Uparrow_\ell$ allowing
    to transform lists of superpositions into superpositions of lists (see
    Example~\ref{ex:cast}).
\end{itemize}

The rewrite system depends on types. Indeed, $\lambda x{:}S\Psi.t$ follows a
call-by-name strategy, while $\lambda x{:}\B.t$, which can duplicate its
argument, follows a call-by-base
strategy~\citep{AssafDiazcaroPerdrixTassonValironLMCS14}, that is, not only the
argument must be reduced first, but also it will distribute over linear
combinations. Therefore, we give first the type system, and then the rewrite
system.

The typing relation is given in Figure~\ref{fig:types}. {Recall that \LambdaS is
a first-order calculus, so only qubit types are allowed to the left of arrows,
and in the contexts.} We write $S^mA$ for $SS\cdots SA$, with $m\geq 1$.
Contexts, identified by the capital Greek letters $\Gamma$, $\Delta$, and
$\Theta$, are partial functions from $\vars$ to $\qtypes$. The contexts
assigning only types of the form $\B^n$ are identified with the super-index
$\B$, e.g.~$\Theta^\B$. Whenever more than one context appear in a typing rule,
their domains are considered pair-wise disjoint. Observe that all types are
linear (as in linear-logic) except on basis types $\B^n$, which can be weakened
and contracted (expressed by the common contexts $\Theta^\B$). {The particular
form of rule $S_E$, allows us to type a measurement even if its argument has
been typed with an arbitrary $k$ number of $S$. We choose this presentation to
avoid subtyping, which was present in the original presentation of
\LambdaS~\citep{DiazcaroDowekTPNC17}.}

\begin{figure}[h!]
  \centering
    \[
      \infer[^\tax] {\Theta^\B,x:\Psi\vdash x:\Psi} {}
      \qquad
      \infer[^{\tax_{\vec 0}}] {\Theta^\B\vdash \z:SA} {}
      \qquad
      \infer[^{\tax_{\ket 0}}] {\Theta^\B\vdash\ket 0:\B} {}
      \quad
      \infer[^{\tax_{\ket 1}}] {\Theta^\B\vdash\ket 1:\B} {}
    \]
    \[
      \infer[^{\alpha_I}] {\Gamma\vdash \alpha.t:S^mA} {\Gamma\vdash t:S^mA}
      \quad
      \infer[^{+_I}] {\Gamma,\Delta,\Theta^\B\vdash\pair tu:S^mA} {\Gamma,\Theta^\B\vdash t:S^mA & \Delta,\Theta^\B\vdash u:S^mA}
      \quad
      \infer[^{S_I}] {\Gamma\vdash t:SA} {\Gamma\vdash t:A}
      \quad
      \infer[^{S_E}] {\Gamma\vdash\pi_j t:\B^j\times S\B^{n-j}} {\Gamma\vdash
        t:S^k\B^n & {\scriptstyle k>0}}
    \]
    \[
      \infer[^\tif]{\Gamma\vdash\ite{}tr:\B\Rightarrow A}{\Gamma\vdash t:A &
      \Gamma\vdash r:A}
      \qquad
      \infer[^{\Rightarrow_I}] {\Gamma\vdash\lambda x{:}\Psi.t:\Psi\Rightarrow A} {\Gamma,x:\Psi\vdash t:A}
    \]
    \[
      \infer[^{\Rightarrow_E}] {\Delta,\Gamma,\Theta^\B\vdash tu:A}
      {
	\Delta,\Theta^\B\vdash u:\Psi
	&
	\Gamma,\Theta^\B\vdash t:\Psi\Rightarrow A
      }
      \qquad
      \infer[^{\Rightarrow_{ES}}] {\Delta,\Gamma,\Theta^\B\vdash tu:SA}
      {
	\Delta,\Theta^\B\vdash u:S\Psi
	&
	\Gamma,\Theta^\B\vdash t:S(\Psi\Rightarrow A)
      }
    \]
    \[
      \infer[^{\times_I}] {\Gamma,\Delta,\Theta^\B\vdash t\times u:\Psi\times\Phi} {\Gamma,\Theta^\B\vdash t:\Psi & \Delta,\Theta^\B\vdash u:\Phi}
      \qquad
      \infer[^{\times_{Er}}] {\Gamma\vdash \head~t:\B} {\Gamma\vdash t:\B^n &
      {\scriptstyle n>1}}
      \qquad
      \infer[^{\times_{E\ell}}] {\Gamma\vdash \tail~t:\B^{n-1}} {\Gamma\vdash t:\B^n &
      {\scriptstyle n>1}}
    \]
    \[
      \infer[^{\Uparrow_r}]{\Gamma\vdash \Uparrow_rt:S(\Psi\times \Phi)}{\Gamma\vdash t:S(S\Psi\times \Phi)}
      \qquad
      \infer[^{\Uparrow_\ell}]{\Gamma\vdash \Uparrow_\ell t:S(\Psi\times \Phi)}{\Gamma\vdash t:S(\Psi\times S\Phi)}
      \qquad
      \infer[^\parallel]{\Gamma\vdash\proba{p_1}t_1 \parallel\dots\parallel\proba{p_n}t_n :A}{\Gamma\vdash t_i:A & \sum_i\mathtt{p_i}=1}
    \]
    \caption{Typing relation}
    \label{fig:types}
\end{figure}

The rewrite relation is given in Figures~\ref{fig:TRSbeta} to
\ref{fig:TRScontext}. We write $t:A$ when there exists $\Gamma$ such that
$\Gamma\vdash t:A$, and $t\ntype A$ if not.

The two beta rules (Figure~\ref{fig:TRSbeta}) are applied according to the type
of the argument. If the abstraction expects an argument with a superposed type,
then the reduction follows a call-by-name strategy (rule $\rbetan$), while if
the abstraction expects a basis type, the reduction is call-by-base (rule
$\rbetab$): it $\beta$-reduces only when its argument is a basis term. However,
typing rules also allow typing an abstraction expecting an argument with basis
type, applied to a term with superposed type (see Example~\ref{ex:CM}). In this
case, the beta reduction cannot occur and, instead, the application must
distribute first, using the rules from Figure~\ref{fig:TRSld}: the linear distribution
rules.
\begin{figure}
  \begin{equation*}
    \begin{aligned}
      \textrm{If $b:\B^n$ and $b\in\tbasis$, } (\lambda x{:}{\B^n}.t)b &\lra (b/x)t & \rbetab\\
      \textrm{If $u:\Psi$, with $\Psi\neq\B^n$, } (\lambda x{:}{\Psi}.t)u &\lra (u/x)t & \rbetan
    \end{aligned}
  \end{equation*} 
  \caption{Beta rules}
  \label{fig:TRSbeta}
\end{figure}
\begin{example}\label{ex:CM}
  The term $\lambda x{:}\B.x\times x$ does not represent a cloning machine, but
  a CNOT with an ancillary qubit~$\ket 0$. Indeed,
  
  $\begin{aligned}[t]
    (\lambda x{:}\B.x\times x)\tfrac 1{\sqrt 2}{.}({\ket 0}+{\ket 1})
    &\xlra{\rlinscalr}
    \tfrac 1{\sqrt 2}.(\lambda x{:}\B.x\times x)({\ket 0}+{\ket 1})\\
    &\xlra{\rlinr}
    \tfrac 1{\sqrt 2}.({(\lambda x{:}\B.x\times x)\ket 0}+{(\lambda x{:}\B.x\times x)\ket 1})\\
    &\xlra{\rbetab^2} 
    \tfrac 1{\sqrt 2}.({\ket{00}}+{\ket{11}})
  \end{aligned}$

  The type derivation is the following:
  \[
    \infer[^{\Rightarrow_{ES}}]{\vdash(\lambda x{:}\B.x\times x)\frac 1{\sqrt 2}.(\ket 0+\ket 1):S\B^2}
    {
      \infer[^{S_I}]{\vdash\lambda x{:}\B.x\times x:S(\B\Rightarrow\B^2)}
      {
	\infer[^{\Rightarrow_I}]{\vdash\lambda x{:}\B.x\times x:\B\Rightarrow\B^2}
	{
	    \infer[^{\times_I}]{x:\B\vdash x\times x:\B^2}
	    {
	      \infer[^\tax]{x:\B\vdash x:\B}{}
	      &
	      \infer[^\tax]{x:\B\vdash x:\B}{}
	    }
	}
      }
      &
	\infer[^{\alpha_I}]{\vdash\frac 1{\sqrt 2}.({\ket 0}+{\ket 1}):S\B }
	{
	  \infer[^{+_I}]{\vdash{\ket 0}+{\ket 1}:S\B }
	  {
      \infer[^{S_I}]{\vdash\ket 0:S\B }{\infer[^{\tax_{\ket 0}}]{\vdash\ket 0:\B}{}}
	    &
	    \infer[^{S_I}]{\vdash\ket 1:S\B }{\infer[^{\tax_{\ket 1}}]{\vdash\ket 1:\B}{}}
	  }
	}
    }
  \]
\end{example}

{The rules from Figure~\ref{fig:TRSld} also say how superposed first-order functions reduce,
  which can be useful for example to describe an operator as the superposition
  of simpler operators, cf.~\citep{ArrighiDowekLMCS17} for more interesting examples.}

\begin{figure}
  \begin{equation*}
    \begin{aligned}
      \textrm{If }t:\B^n\Rightarrow A,\ t\pair uv &\lra \pair{tu}{tv} & \rlinr\\
      \textrm{If }t:\B^n\Rightarrow A,\ t(\alpha.u) &\lra\alpha.tu & \rlinscalr\\
      \textrm{If }
        t:\B^n\Rightarrow A,\ 
        t\z[\B^n] &\lra\z &\rlinzr\\
      \pair tuv &\lra\pair{tv}{uv} & \rlinl\\
      (\alpha.t)u &\lra\alpha.tu &\rlinscall\\
      \z[(\B^n\Rightarrow A)]t &\lra\z &\rlinzl
    \end{aligned}
  \end{equation*}
  \caption{Linear distribution rules}
  \label{fig:TRSld}
\end{figure}

Figure~\ref{fig:TRSif} gives the two rules for the conditional construction.
Together with the linear distribution rules (Figure~\ref{fig:TRSld}), these
rules implement the quantum-if~\citep*{AltenkirchGrattageLICS05}, as shown in the following example.
\begin{example}
  \label{ex:ite}
  The term $\ite{}rs$ is meant to test whether the condition is $\ket 1$ or
  $\ket 0$.  However, defining it as a function allows us to use the linear
  distribution rules from Figure~\ref{fig:TRSld}, implementing the quantum-if:
  \begin{align*}
    (\ite{}rs)(\alpha.\ket 1+\beta.\ket 0)
    &\xlra{\rlinr}
      (\ite{}rs)(\alpha.\ket 1)+(\ite{}rs)(\beta.\ket 0)
    \xlra{\rlinscalr^2}
      \alpha{.}(\ite{}rs)\ket 1+\beta{.}(\ite{}rs)\ket 0
      \\
    &=
      \alpha{.}(\ite{\ket 1}rs)+\beta{.}(\ite{\ket 0}rs)
    \xlra{\riftrue}
      \alpha{.}r+\beta{.}(\ite{\ket 0}rs)
    \xlra{\riffalse}
      \alpha{.}r+\beta{.}s
  \end{align*}

  This construction allow us to encode any quantum gate.
\end{example}

\begin{figure}
  \begin{equation*}
  \begin{aligned}
      \ite{\ket 1}tr &\lra t &\riftrue
  \end{aligned}\hspace{3cm}
  \begin{aligned}
      \ite{\ket 0}tr &\lra r &\riffalse
  \end{aligned}
  \end{equation*}
  \caption{Rules of the conditional construction}
  \label{fig:TRSif}
\end{figure}

Figure~\ref{fig:TRSlists} gives the rules for lists, \rhead\ and \rtail, which
can only act in basis qubits, otherwise, we would be able, for example, to
extract a qubit from an entangled pair of qubits.
\begin{figure}
  \begin{equation*}
  \begin{aligned}
      \textrm{If $h\neq u\times v$ and $h\in\tbasis$, }\head\ h\times t &\lra h & \rhead\\
      \textrm{If $h\neq u\times v$ and $h\in\tbasis$, }\tail\ h\times t &\lra t & \rtail
  \end{aligned}
  \end{equation*}
  \caption{Rules for lists}
  \label{fig:TRSlists}
\end{figure}

Figure~\ref{fig:TRSvs} deals with the vector space structure implementing a
directed version of the vector space axioms. The direction is chosen in order to
yield a canonical form~\citep{ArrighiDowekLMCS17}. The rules are
self-explanatory.
  There is a subtlety, however, on the rule \rzeros. A simpler rule, for example
  ``If $t:A$ then $0.t\lra \z$'', would lead to break confluence with the following critical pair:
$0.\z\xlra{\rzeros}\z[SA]$ and $0.\z\xlra{\rzero}\z$.
To solve the critical pair, \citet*{DiazcaroDowekRinaldiBIO19} added a new
definition ``$\min A$'', which leaves the type $A$ with a minimum amount of
$S$ in head position (one, if there is at least one, or zero, in other case).
This solution makes sense in such a presentation of \LambdaS, since the
interpretation of the type $SSA$ coincides with the interpretation of $SA$ (both
are the vector space generated by the span over $A$). However, in our
categorical interpretation these two types are not interpreted in the same way,
and so, our rule $\rzeros$ sends $0.\z$ to $\z$ directly. Similarly, all the
rules ending in $\z$ have been modified from its original presentation in the
same way, namely: $\rzeros$, $\rzero$, $\rlinzr$, and $\rlinzl$.

\begin{figure}
  \begin{equation*}
  \begin{aligned}
      \pair\z t &\lra t &\rneut\\
      1.t &\lra t &\runit\\
      \textrm{If }\left\{ \begin{matrix}
        t:A, \textrm{ with }A\in\btypes, \textrm{ or }\\
        t:SA, \textrm { and }t\ntype A
        \end{matrix}\right\},\  0.t &\lra\z &\rzeros\\
      \alpha.\z &\lra\z &\rzero\\
      \alpha.(\beta.t) &\lra (\alpha\beta).t &\rprod\\
      \alpha.\pair tu &\lra\pair{\alpha.t}{\alpha.u} &\rdists\\
      \pair{\alpha.t}{\beta.t} &\lra(\alpha+\beta).t &\rfact\\
      \pair{\alpha.t}t &\lra (\alpha+1).t &\rfacto\\
      \pair tt &\lra 2.t &\rfactt
  \end{aligned}
  \end{equation*}
  \caption{Rules implementing the vector space axioms}
  \label{fig:TRSvs}
\end{figure}
\begin{example}
  \begin{align*}
    2{.}\left(\frac 12{.}\ket 0+\ket 1\right)-2{.}\ket 1
    &\xlra{\rdists}
      2.\left(\frac 12{.}\ket 0\right)+2{.}\ket 1-2{.}\ket 1
    \\
    &\xlra{\rprod}
      1.\ket 0+2{.}\ket 1-2{.}\ket 1
    \\
    &\xlra{\runit}
      \ket 0+2.\ket 1-2.\ket 1
      \\
    &\xlra{\rfact}
      \ket 0+0{.}\ket 1
    \\
    &\xlra{\rzeros}
      \ket 0+\z[\B]
    \\
    &\xlra{\rneut}
      \ket 0
  \end{align*}
  Remind that the symbol $+$ is associative and commutative.
\end{example}

Figure~\ref{fig:TRScasts} are the rules to implement the castings\footnote{The
subtlety about $\z$ explained for Figure~\ref{fig:TRSvs} has led us to add some
extra rules to \LambdaS, with respect to its original presentation, in
Figure~\ref{fig:TRScasts}. Those are: $\rdistcazeror$, $\rdistcazerol$,
$\rcaneutzr$, and $\rcaneutzl$.}. The idea is that $\times$ does not distribute
with respect to $+$, unless a casting allows such a distribution. This way, the
types $S\B\times \B $ and $S(\B\times\B)$ are different. Indeed, $(\ket 0+\ket
1)\times\ket 0$ have the first type but not the second, while
$\ket{00}+\ket{10}$ have the second type but not the first. The first type gives
us the information that the state is separable, while the second type does not.
We can choose to take the first state as a pair of qubits forgetting the
separability information, by casting its type, in the same way as in certain
programming languages an integer can be cast to a float (and so, forgetting
the information that it was indeed an integer and not just any float).
\begin{example}
  \label{ex:cast}
  The term $(\frac 1{\sqrt 2}{.}(\ket 0+\ket 1))\times\ket 0$ is the encoding of
  the qubit $\frac 1{\sqrt 2}(\ket 0+\ket 1)\otimes\ket 0$. However, while the
  qubit $\frac 1{\sqrt 2}(\ket 0+\ket 1)\otimes\ket 0$ is equal to $\frac
  1{\sqrt 2}(\ket{00}+\ket{10})$, the term will not
  rewrite to the encoding of it, unless it is preceded by a casting $\Uparrow_r$:
  \begin{align*}
    \Uparrow_r\left(\frac 1{\sqrt 2}{.}(\ket 0+\ket 1)\right)\times\ket 0
    &\xlra{\rdistscalr}
      \frac 1{\sqrt 2}{.}\left(\Uparrow_r(\ket 0+\ket 1)\times\ket 0\right)
    \\
    &\xlra{\rdistsumr}
      \frac 1{\sqrt 2}{.}(\Uparrow_r\ket{00}+\Uparrow_r\ket{10})
      \\
    &\xlra{\rcaneutr^2}
      \frac 1{\sqrt 2}{.}(\ket{00}+\ket{10})
  \end{align*}
  Notice that $\left(\frac 1{\sqrt 2}{.}(\ket 0+\ket 1)\right)\times\ket 0$ has
type $S\B \times\B$, highlighting the fact that the second qubit is a basis
qubit, i.e.~duplicable, while $\frac 1{\sqrt 2}{.}(\ket{00}+\ket{10})$ has type
$S(\B\times\B)$, showing that the full term is a superposition where no
information can be extracted, and hence, non-duplicable.

\end{example}
\begin{figure}
  \begin{equation*}
  \begin{aligned}
      \Uparrow_r \pair rs\times u &\lra\pair{\Uparrow_r r\times u}{\Uparrow_r s\times u} &\rdistsumr\\
      \Uparrow_\ell u\times\pair rs &\lra\pair{\Uparrow_\ell u\times r}{\Uparrow_\ell u\times s} &\rdistsuml\\
      \Uparrow_r (\alpha.r)\times u &\lra \alpha.\Uparrow_r r\times u &\rdistscalr\\
      \Uparrow_\ell u\times(\alpha.r) &\lra \alpha.\Uparrow_r u\times r &\rdistscall\\
      \textrm{If $u$ has type $\Psi$, }\Uparrow_r \z[\Phi]\times u &\lra\z[(\Phi\times\Psi)] &\rdistzr\\
      \textrm{If $u$ has type $\Psi$, }\Uparrow_\ell u\times\z[\Phi] &\lra\z[(\Psi\times\Phi)] &\rdistzl\\
      \Uparrow\pair tu&\lra\pair{\Uparrow t}{\Uparrow u} &\rdistcasum\\
      \Uparrow(\alpha.t)&\lra\alpha.\Uparrow t &\rdistcascal\\
      \Uparrow_r\z[(S(S\Psi)\times\Phi)] &\lra\Uparrow_r\z[(S\Psi\times\Phi)] &\rdistcazeror\\
      \Uparrow_\ell\z[(\Psi\times S(S\Phi))] &\lra\Uparrow_\ell\z[(\Psi\times S\Phi)] &\rdistcazerol\\
      \Uparrow_r\z[(S\B^n \times\Phi)] &\lra\z[(\B^n\times\Phi)] &\rcaneutzr\\
      \Uparrow_\ell\z[(\Psi\times S\B^n)] &\lra\z[(\Psi\times\B^n)] &\rcaneutzl\\
      \textrm{If $u\in\tbasis$, }\Uparrow_r u\times v&\lra u\times v &\rcaneutr\\
      \textrm{If $v\in\tbasis$, }\Uparrow_\ell u\times v&\lra u\times v &\rcaneutl
  \end{aligned}
  \end{equation*}
  \caption{Rules for castings $\Uparrow_r$ and $\Uparrow_\ell$}
  \label{fig:TRScasts}
\end{figure}

Figure~\ref{fig:TRSproj} gives the rule $\rproj$ for the projective measurement
with respect to the basis $\{\ket 0,\ket 1\}$. It acts only on superpositions of
terms in normal-form, however, these terms do not necessarily represent a norm-1
vector, so the measurement must perform a division by the norm of the vector
prior to measure. In case the norm of the term is $0$, then an error is raised.
In the original version of \LambdaS, such an error is left as a term that does
not reduce. In this paper, however, we added a new rule $\rprojz$ for the
projective measurement over the null vector, in order to simplify the model
(otherwise, we would have been forced to add Moggi's exception monad
\citep{MoggiFCS88} to the model). Since the model we present in this paper is
already complex, we prefer to add a rule sending the error to a fixed value and
focus on the novel constructions.

In rule \rproj, $j\leq n$, and we use the following notations: 
\begin{align*}
  & \may t\textrm{ may be either $t$ or $\alpha.t$ (if it is not present, $\alpha=1$)} \\
  &\ket k=\ket{b_1\dots b_j}\textrm{ where }b_1\dots b_j\textrm{ is the binary representation of }k\\
  & \ket{\phi_k} = \sum_{i\in T_k}\left(\frac{\alpha_i}{\sqrt{\sum_{r\in T_k}|\alpha_r|^2}}\right)\prod_{h=j+1}^n \ket{b_{hi}}\\
  &\mathtt{p_k}=\sum_{i\in T_k}\left(\frac{|\alpha_i|^2}{\sum_{r=1}^m|\alpha_r|^2}\right)\\
  & T_k=\{i\leq m\mid\ket{b_{1i}\dots b_{ji}}=\ket k\}
\end{align*}
This way, $\ket k\times\ket{\phi_k}$ is the normalized $k$-th projection of
the term.
\begin{figure}
  \begin{equation*}
  \begin{aligned}
      \pi_j(\sum\limits_{i=1}^m\may[\alpha_i]\prod\limits_{h=1}^n\ket{b_{hi}})
      &\lra \bigparallel\limits_{k=0}^{2^j-1}\proba{p_k} (\ket k\times\ket{\phi_k}) & \rproj \\
      \pi_j\z[\B^n] &\lra \ket 0^{\times n} & \rprojz
  \end{aligned}
  \end{equation*}
  \caption{Rules for the projection}
  \label{fig:TRSproj}
\end{figure}
\begin{example}\label{ex:pi}
  Lets measure the first two qubits of a three qubits superposition. So, in rule
  $\rproj$ take $j=2$ and $n=3$.
  Say, the term to measure is $\ket{000}+2.\ket{110}+3.\ket{001}+\ket{111}$.
  So, we have $m=4$, and
  \[
    \begin{array}{c||c|c} 
      i & \alpha_i & \ket{b_{1i}b_{2i}b_{3i}} \\
      \hline
      1 & 1 & \ket{000}\\
      2 & 2 & \ket{110}\\
      3 & 3 & \ket{001}\\
      4 & 1 & \ket{111}
    \end{array}
  \qquad\qquad\qquad
  \begin{array}{c||c|c|c|c}
    k &\ket k & T_k & p_k & \ket{\phi_k}\\\hline
    0 &\ket{00} & \{1,3\} &\frac{1}{15}+\frac{9}{15}=\frac 23 & \frac 1{\sqrt{10}}.\ket 0+\frac{2.7}{\sqrt{10}}.\ket 1\\
    1 &\ket{01} & \emptyset & 0 &\mbox{--}\\
    2 &\ket{10} &\emptyset & 0 &\mbox{--}\\
    3 &\ket{11} & \{2,4\} & \frac 4{15}+\frac 1{15}=\frac 13 & \frac 2{\sqrt 5}.\ket 0+\frac 1{\sqrt 5}.\ket 1
  \end{array}
  \]
  All in all, the reduction is as follows:
  \begin{align*}
    &\pi_2(\ket{000}+2.\ket{110}+3.\ket{001}+\ket{111})\\
    &\xlra{\rproj}
    \proba{\tfrac 23}\left( \ket{00}\times
    (\tfrac 1{\sqrt{10}}.\ket 0+\tfrac 3{\sqrt{10}}.\ket 1)\right)
    \parallel \proba{\tfrac 13}\left(\ket{11}\times(\tfrac 2{\sqrt 5}.\ket 0+\tfrac 1{\sqrt 5}.\ket 1)\right)
  \end{align*}
  Notice that, since $\vdash\ket{000}+2.\ket{110}+3.\ket{001}+\ket{111}:S\B^3$,
  we have
  \[
  \vdash\pi_2(\ket{000}+2.\ket{110}+3.\ket{001}+\ket{111}):\B^2\times S\B
  \]
\end{example}

Finally, Figure~\ref{fig:TRScontext} gives the contextual rules implementing the
call-by-value and call-by-name weak strategies (weak in the sense that there is
no reduction under lambda).
\begin{figure}[t]
  \centering
  \[
    \begin{array}{r@{\ }lr}
      \multicolumn{3}{l}{\textrm{ If $t\lra u$, then}}\\
      \multicolumn{3}{c}{ 
	\begin{array}{c@{\qquad}c@{\qquad}c}
	  tv \lra uv & (\lambda x^{\B^n}.v)t\lra(\lambda x^{\B^n}.v)u & \pair tv\lra\pair uv \\
	  \alpha.t\lra\alpha.u& \pi_j t\lra\pi_j u& t\times v\lra u\times v \\
	  v\times t\lra v\times u& \Uparrow_r t\lra\Uparrow_r u& \Uparrow_\ell t\lra\Uparrow_\ell u \\
	  \head\ t\lra\head\ u& \tail\ t\lra\tail\ u& \ite trs\lra\ite urs \\
	  \multicolumn{3}{c}{
    \left(\proba{p_1}t_1\parallel\cdots\parallel\proba{p_k}t\parallel\cdots\parallel\proba{p_n}t_n\right)\quad\lra\quad\left(\proba{p_1}t_1\parallel\cdots\parallel\proba{p_k}u\parallel\cdots\parallel\proba{p_n}t_n\right)}
	\end{array}
      }
    \end{array}
  \]
  \caption{Contextual rules (notice that, in particular, there is no reduction
    under lambda).}
  \label{fig:TRScontext}
\end{figure}

\begin{example}
  A Hadamard gate can be implemented by $H=\lambda x:\B.\ite x{\ket -}{\ket +}$,
  where $\ket +=\frac 1{\sqrt 2}.\ket 0+\frac 1{\sqrt 2}.\ket 1$ and
  $\ket -=\frac 1{\sqrt 2}.\ket 0-\frac 1{\sqrt 2}.\ket 1$.

  Therefore, $H:\B\Rightarrow S\B $ and we have $H\ket 0\lra^*\ket +$ and $H\ket 1\lra^*\ket
-$.
\end{example}

Correctness has been established in previous works for slightly different
versions of \LambdaS, except for the case of confluence, which have only been
proved for Lineal \citep{ArrighiDowekLMCS17}. Lineal can be seen as an untyped fragment of \LambdaS without several
constructions (in particular, without $\pi_j$). The proof of confluence for
\LambdaS is delayed to future work, using the development of probabilistic
confluence by~\citet*{DiazcaroMartinezLSFA17}. The proof of Subject Reduction
and Strong Normalization are straightforward modifications from the proofs of
the different presentations of \LambdaS.

\begin{thm}[Confluence of Lineal, \citaCR]\label{thm:confluence}
  Lineal, an untyped fragment of~\LambdaS, is confluent.
  \qed
\end{thm}

\begin{thm}[Subject reduction on closed terms, \citaSR]\label{thm:SR}
  For any closed terms $t$ and $u$ and type $A$, if $t\lra \bigparallel_i
  \proba{p_i}u_i$ and $\vdash t:A$, then $\vdash\bigparallel_i\proba{p_i}u_i:A$.
  \qed
\end{thm}

\begin{thm}[Strong normalization, \citaSN]\label{thm:SN}
  If $\vdash t:A$ then $t$ is strongly normalizing{, that is, there
    is no infinite rewrite sequence starting from $t$}.
  \qed
\end{thm}

\begin{thm}[Progress]\label{thm:progress}
  If $\vdash t:A$ and $t$ does not reduce, then $t$ is 
  a value. 
\end{thm}
\begin{proof}
  By induction on $t$.
  \begin{itemize}
  \item If $t$ is a value, then we are done.
  \item Let $t=rs$, then $\vdash r:S(\Psi\Rightarrow C)$. So, by the induction
    hypothesis, $r$ is a value. Therefore, by its type, $r$ is either a lambda
    term, or a superposition of them,
    and so $rs$ reduces, which is absurd.
  \item Let $t=r+s$, then by the induction hypothesis both $r$ and $s$ are
    values, and so $r+s$ is a value.
  \item Let $t=\pi_jr$, then, by the induction hypothesis, $r$ is a value, and
    since $t$ is typed, $\vdash r:S\B^n$. Therefore, the only possible $r$ are
    superpositions of kets, and so, $t$ reduces, which is absurd.
  \item Let $t=\alpha.r$, then by the induction hypothesis $r$ is a value, and
    so $t$ is a value.
  \item Let $t=r\times s$, then by the induction hypothesis both $r$ and $s$ are
    values, and so $t$ is a value.
  \item Let $t=\head\ r$, then, by the induction hypothesis $r$ is a value, and
    since $t$ is typed, $\vdash r:\B^n$. Therefore, the only possible $r$ are
    products of kets, and so $t$ reduces, which is absurd.
  \item Let $t=\tail\ r$. Analogous to previous case.
  \item Let $t=\Uparrow_r r$, then, by the induction hypothesis, $r$ is a value.
    Since $t$ is typed, $\vdash r:S(S\Psi\times\Phi)$.
    Therefore, the cases for $r$ are:
    \begin{itemize}
    \item $r=x$. Absurd, since $r$ is closed.
    \item $r=\lambda x:\Theta.r'$. Absurd since $\vdash r:S(S\Psi\times\Phi)$.
    \item $r=\ket 0$. Absurd since $\vdash r:S(S\Psi\times\Phi)$.
    \item $r=\ket 1$. Absurd since $\vdash r:S(S\Psi\times\Phi)$.
    \item $r=v_1+v_2$, then $t$ reduces by rule $\rdistcasum$, which is absurd.
    \item $r=\z[(S\Psi\times\Phi)]$, then $t$ reduces by rule $\rdistcazeror$ or
      $\rcaneutzr$, which is absurd.
    \item $r=\alpha.v$, then $t$, reduces by rule $\rdistcascal$, which is absurd.
    \item $r=\ite{}{s_1}{s_2}$. Absurd since $\vdash r:S(S\Psi\times\Phi)$.
    \item $r=v_1\times\cdots\times v_n$, with $v_1$ not a pair, then the possible $v_1$ are:
      \begin{itemize}
      \item $v_1\in\tbasis$, then $t$ reduces by rule $\rcaneutr$, which is absurd.
      \item $v_1=v'_1+v'_2$, then $t$ reduces by rule $\rdistsumr$, which is absurd.
      \item $v_1=\z[(S\Psi\times\Phi)]$, then $t$ reduces by rule $\rdistzr$,
        which is absurd.
      \item $v_1=\alpha.v$, then $t$ reduces by rule $\rdistscalr$, which is absurd.
      \end{itemize}
    \end{itemize}
  \item Let $t=\Uparrow_\ell r$. Analogous to previous case.
    \qedhere
  \end{itemize}
\end{proof}

\section{Denotational semantics}\label{sec:DenSem}
Even though the semantic of this article is about particular categories i.e.~the
category of sets and the category of vector spaces, from the start our approach
uses theory and tools from category theory in an abstract way. The idea is that
the concrete situation
exposed in this article will pave the way to a more abstract formulation, and
that is why we develop the constructions as abstract and general as possible. A
more general treatment, using a monoidal adjunction between a Cartesian closed
category and a monoidal category with some extra conditions, remains a topic for
future work. {A first result in such direction has been published
recently~\citep{DiazcaroMalherbeACS20}, however in a simplified version of
\LambdaS without measurements}.

\subsection{Categorical constructions}
The concrete categorical model\footnote{{Although ``concrete categorical'' seems
  paradoxical, since a model can either be concrete, or categorical, we chose
to use this terms to stress the fact that we use a categorical presentation of
this concrete model.}} for \LambdaS will be given using the following constructions:
  \begin{itemize}
  \item A monoidal adjunction
    \begin{center}
          \begin{tikzcd}[column sep=3mm]
            (\mathbf{Set},\times,\One)\ar[rr,bend left,"\paren{S,m}"]
            & \bot & (\mathbf{Vec},\otimes,\I)\ar[ll,bend left,"\paren{U,n}"]
          \end{tikzcd}
    \end{center}
            where
            \begin{itemize}
            \item $\mathbf{Set}$ is the category of sets with $\One$ as a terminal
              object.
            \item $\mathbf{Vec}$ is the category of vector spaces over $\C$, in
              which $\I=\C$.

            \item $S$ is the functor such that for each set $A$, $SA$ is the vector
              space whose vectors are the formal finite linear combinations of the
              elements of $A$ with coefficients in $\C$, and given a function $f:A\to B$
              we define $Sf:SA\to SB$ by evaluating $f$ in $A$.
            \item $U$ is the forgetful functor such that for each vector space $V$,
              $UV$ is the underlying set of vectors in $V$ and for each linear map $f$, $Uf$ is just $f$ as function not taking into account its linear property.

            \item $m$ is a natural isomorphism defined by
              \[
                m_{AB}: SA\otimes SB\to S(A\times B)
                \qquad\qquad
                (\sum_{a\in A}\alpha_a a)\otimes (\sum_{b\in B}\beta_b b)  \mapsto \sum_{(a,b)\in A\times B}\alpha_a\beta_b(a,b)
              \]
            \item $n$ is a natural transformation defined by
              \[
                n_{AB}: UV\times UW \to U(V\otimes W)
                \qquad\qquad
                (v,w) \mapsto v\otimes w
              \]
            \end{itemize}
  \item $\mathbf{Vec}^\dagger$ is a subcategory of $\mathbf{Vec}$, where every
    morphism $f:V\to W$ have associated a
    morphism $f^\dagger:W\to V$, called the \emph{dagger} of $f$, such that
    for all $f:V\to W$ and $g:W\to U$ we have
    \[
      \Id^\dagger_V =\Id_V\qquad\qquad
      (g\circ f)^\dagger =f^\dagger\circ g^\dagger\qquad\qquad
      f^{\dagger\dagger} = f
    \]
    Notice that $\mathbf{Vec}^\dagger$ is a subcategory of $\mathbf{FinVec}$, the
    category of finite vector spaces over $\C$.
  \item $\mathbf{Set}_D$ is a Kleisli category over $\mathbf{Set}$ defined with the following monoidal monad, called the
    distribution monad~\citep{MoggiFCS88,GiryLNM82}, $(D,\etaD,\muD,\hat m_{AB},\hat m_1)$:
    \[
      D:\mathbf{Set}\to\mathbf{Set}
      \qquad\qquad
      DA=\left\{\sum_{i=1}^n p_i\chi_{a_i}\mid \sum_{i=1}^n p_i=1,	a_i\in A, n\in\mathbb N\right\}
    \]
    where $\chi_a$ is the characteristic function of $a$, and $\etaD$, $\muD$, $\hat m_{AB}$, and $\hat m_1$ are defined as follows:
    \[
       \begin{array}{r@{\qquad\qquad}l}
        \etaD:A \to DA
        &
        a\mapsto 1\chi_a
        \\
        \muD:DDA \to DA
        &
        \sum_{i=1}^n p_i\chi_{(\sum_{j=1}^{m_i}q_{ij}\chi_{a_{ij}})} \mapsto\sum_{i=1}^n\sum_{j=1}^{m_i} p_iq_{ij}\chi_{a_{ij}}
        \\
        \hat m_{AB}:DA\times DB\to D(A\times B)
        &
        \left(\sum_{i=1}^np_i\chi_{a_i} , \sum_{j=1}^mq_j\chi_{b_j}\right) \mapsto \sum_{i=1}^n\sum_{j=1}^m p_iq_j(\chi_{a_i},\chi_{b_j})
        \\
        \hat m_1:1 \to D1
        &
        \ast \mapsto 1\chi_\ast
      \end{array}
    \]
  \end{itemize}

\begin{rmks}\label{rmk:categoria}~
  \begin{itemize}
  \item There exists an object $\B$ and maps $i_1$, $i_2$ in $\mathbf{Set}$ such
    that for every $t:\One\longrightarrow A$ and $r:\One\longrightarrow A$,
    there exists a unique map $\home tr$ making following diagram
    commute:
    \begin{center}
      \begin{tikzpicture}
        \node at (0,0) {\begin{tikzcd}
            \One\ar[r,"i_1"]\ar[rd,"t"'] & \B\ar[d,near start,"\home tr"] &\One\ar[l,"i_2"']\ar[ld,"r"]\\
            &A &
          \end{tikzcd}};
	  \node at (7,0) {\parbox{.5\textwidth}{This object $\B$ is the Boolean set, and such a map will allow us to interpret the \emph{if} construction (Definition~\ref{def:if}).}};
      \end{tikzpicture}
    \end{center}
    
  \item For every $A\in|\mathbf{Set}|$, $\mathbf{Vec}(\I,SA)$ is an abelian
    group with the sum defined point-wise. Therefore, there exists a map
    $+:USA\times USA\rightarrow USA$ in $\mathbf{Set}$, given by $(a,b)\mapsto
    a+b$ {using the underlying} sum {from} $SA$.

  \item To have an adjunction means that each function $g:A\to UV$ extends to a
    unique linear transformation $f:SA\to V$, given explicitly by $f(\sum_i
    \alpha_i x_i)=\sum_i\alpha_i g(x_i)$, that is, formal linear combinations in
    $SA$ to actual linear combinations in $V$ (see \citep{MacLane98} for
    details).

  \item $\mathbf{Set}$ is a Cartesian closed category where $\eta^A$ is the unit
    and $\varepsilon^A$ is the counit of $-\times A\dashv[A,-]$, from which we
    can define the curryfication ($\mathsf{curry}$) and un-curryfication
    ($\mathsf{uncurry}$) of any map.

  \item The defined adjunction between $\mathbf{Set}$ and $\mathbf{Vec}$ gives
    rise to a monad $(T,\eta,\mu)$ in the category $\mathbf{Set}$, where $T=US$,
    $\eta:\Id\to T$ is the unit of the adjunction, and using the counit
    $\varepsilon$, we obtain $\mu=U\varepsilon S:TT\to T$, satisfying unity and
    associativity laws (see~\citep{MacLane98}).
  \end{itemize}
\end{rmks}

\subsection{Interpretation}

\begin{defin}
  Types are interpreted in the category $\mathbf{Set}_D$, as follows:
  \[
    \sem{\B} = \B\hspace{1cm}
    \sem{\Psi\Rightarrow A} = \sem{\Psi}\Rightarrow\sem{A}\hspace{1cm}
    \sem{SA} = US\sem{A}\hspace{1cm}
    \sem{\Psi\times \Phi} =\sem \Psi\times\sem \Phi
  \]
\end{defin}

\begin{rmk}
  To avoid cumbersome notation we omit the brackets $\sem{\cdot}$ when there is
  no ambiguity. For example, we write directly $USA$ for $\sem{SA}=US\sem A$
  and $A$ for $\sem{A}$.
\end{rmk}

Before giving the interpretation of typing derivation trees in the model, we
need to define certain maps that will serve to implement some of the
constructions in the language.

To implement the \emph{if} construction we define the following map.
\begin{defin}\label{def:if}
  Given $t,r \in\home \Gamma A$ there exists a map $\B\xlra{f_{t,r}}\home \Gamma
  A$  in $\mathbf{Set}$ defined by $f_{t,r}= \home{\hat{t}}{\hat{r}}$ where
  $\hat{t}:1\rightarrow\home \Gamma A$ and  $\hat{r}:1\rightarrow\home \Gamma A$
  are given by the constant maps $\star\mapsto t$ and $\star\mapsto s$ respectively. Concretely this means that $i_1(\star)\mapsto t$ and
  $i_2(\star)\mapsto r$.
\end{defin}

\begin{example}\label{ex:if}
  Consider $t=i_1$ and $r=i_2$, with
  $t,r\in\home\One\B$, where $\B=\{i_1(\star),i_2(\star)\}$.
  To make the example more clear, let us consider $i_1(\star)=\ket 0$ and
  $i_2(\star)=\ket 1$, hence $\B=\{\ket 0,\ket 1\}$.
  The map $\B\xlra{f_{t,r}}\home\One\B$ in $\mathbf{Set}$
  is defined by $f_{t,r}=\home{\hat i_1}{\hat i_2}$, where
  $\hat i_k:\One\rightarrow\home\One\B$, for $k=1,2$.
  Therefore, we have the following commuting diagram:
  
  \begin{tikzpicture}
    \node at (0,0) {\begin{tikzcd}
        \One\ar[r,"i_1"]\ar[rd,"\hat i_1"'] & \B\ar[d,near start,"f_{t,r}"] &\One\ar[l,"i_2"']\ar[ld,"\hat i_2"]\\
        &\home\One\B&
      \end{tikzcd}};
    \node at (7.9,0){\parbox{.51\textwidth}{
        Hence, we have:\\
        $\begin{aligned}[t]
          f_{t,r}\ket 0&=f_{t,r}(i_1(\star))=(f_{t,r}\circ i_1)\star=\hat i_1(\star)=i_1=t\\
          f_{t,r}\ket 1&=f_{t,r}(i_2(\star))=(f_{t,r}\circ i_2)\star=\hat i_2(\star)=i_2=r
        \end{aligned}$\\                 
        Therefore, $f_{t,r}$ is the map $\ket 0\mapsto t$ and $\ket 1\mapsto r$.}};
  \end{tikzpicture}
\end{example}

In order to implement the projection, we define a map $\pi_j$
(Definition~\ref{def:pi}), which is formed from the several maps that we describe below. 

A projection $\pi_{jk}$ acts in the following way: first it projects the first
$j$ components of its
argument, an $n$-dimensional vector,
to the basis vector $\ket k$ in the vector space of dimension $j$, then it renormalizes it, and
finally it factorizes the first $j$ components.
Then, the projection $\pi_j$ takes the probabilistic distribution between the
$2^j$ projectors $\pi_{jk}$, each of these probabilities, calculated from the
normalized vector to be projected.
\begin{example}\label{ex:explanation-pi}
  Let us analyse the Example~\ref{ex:pi}:
  \begin{align*}
    &\pi_2(\ket{000}+2.\ket{110}+3.\ket{001}+\ket{111})\\
    &\xlra{\rproj}
    \proba{\tfrac 23}\left( \ket{00}\times
    (\tfrac 1{\sqrt{10}}.\ket 0+\tfrac 3{\sqrt{10}}.\ket 1)\right)
    \parallel \proba{\tfrac 13}\left(\ket{11}\times(\tfrac 2{\sqrt 5}.\ket 0+\tfrac 1{\sqrt 5}.\ket 1)\right)
  \end{align*}
  We can divide this in four projectors (since $j=2$, we have $2^2$ projectors),
  which are taken in parallel (with the symbol $\parallel$). The four projectors
  are:
  $\pi_{2,00}$, $\pi_{2,01}$, $\pi_{2,10}$ and $\pi_{2,11}$. In this case, the
  probability for the projectors $\pi_{2,01}$ and $\pi_{2,10}$ are $\mathtt 0$, and
  hence these do not appear in the final term.

  The projector $\pi_{2,00}$ acts as described before: first it projects the
  first $2$ components of $\ket{000}+2.\ket{110}+3.\ket{001}+\ket{111}$ to the basis vector $\ket{00}$, obtaining
  $\ket{000}+3.\ket{001}$. Then it renormalizes it, by dividing it by its norm,
  obtaining $\frac{1}{\sqrt{10}}.\ket{000}+\frac 3{\sqrt{10}}.\ket{001}$.
  Finally, it factorizes the vector, obtaining $\ket{00}\times(\frac
  1{\sqrt{10}}.\ket 0+\frac 3{\sqrt{10}}.\ket 1)$.
  Similarly, the projector $\pi_{2,11}$ gives $\ket{11}\times(\tfrac 2{\sqrt 5}.\ket 0+\tfrac 1{\sqrt 5}.\ket 1)$.
  Finally, the probabilities to assemble the final term are calculated as
  $\mathtt{p_0}=\frac{|1|^2+|3|^2}{|1|^2+|2|^2+|3|^2+|1|^2}=\mathtt{\frac 23}$ and $\mathtt{p_1}=\frac{|2|^2+|1|^2}{|1|^2+|2|^2+|3|^2+|1|^2}=\mathtt{\frac 13}$.
\end{example}

Categorically, we can describe the operator $\pi_{jk}$
(Definition~\ref{def:projk}) by the composition of three arrows:
a projector arrow to the $\ket k$ basis vector (Definition~\ref{def:proj}),
a normalizing arrow $\s{Norm}$ (Definition~\ref{def:norm}),
and a factorizing arrow $\varphi_j$ (Definition~\ref{def:fact}).
Then, the projection $\pi_j$ (Definition~\ref{def:pi}) maps a vector to the
probabilistic distribution between the $2^j$ basis vectors $\ket k$, using a
distribution map (Definition~\ref{def:d}).

In the following definitions, if $\ket\psi$ is a vector of dimension $n$, we
write $\overline{\ket\psi}:\I\to S\B^n $ to the constant map $1\mapsto\ket\psi$.

\begin{defin}
  \label{def:proj}
  The projector arrow to the $\ket k$ basis vector  $\s{Proj}_k$ is defined as follows:
  \[
    \s{P}_k :(S\B)^{\otimes n}\mapsto (S\B)^{\otimes n}
    \qquad\qquad
    \s{P}_k= (\overline{\ket k}\circ\overline{\ket k}^\dagger)\otimes I
  \]
\end{defin}

\begin{defin}\label{def:norm}
  The normalizing arrow $\s{Norm}$ is defined as follows: 
  \[
    \s{Norm} : US\B^{n} \to US\B^{n} \hspace{2cm}
    \ket\psi\mapsto\left\{
      \begin{array}{ll}
        \frac{\ket\psi}{\sqrt{(\overline{\ket\psi}^\dagger\circ\overline{\ket\psi})(\star)}}
        &\textrm{ if }\ket\psi\neq\vec 0\\
        \ket 0 &\textrm{ otherwise}
      \end{array}
    \right.
  \]
\end{defin}
\begin{defin}\label{def:fact} 
  The factorizing arrow $\varphi_j$ is defined as any arrow making the following
  diagram commute:
  \begin{center}
    \begin{tikzcd}
      \B^j\times US\B^{n-j}\ar[r,"\eta\times\Id"]\ar[d,"\Id"] & US\B^j\times
      US\B^{n-j}\ar[r,"n"] & U(S\B^j\otimes S\B^{n-j})\ar[d,"Um"]\\
      \B^j\times US\B^{n-j} &&  US\B^n =US(\B^j\times \B^{n-j})
      \ar[ll,"\varphi_j"]
    \end{tikzcd}
  \end{center}
\end{defin}
\begin{example}
  For example, take $\varphi_j$ as the following map:
  \begin{align*}
    \varphi_j&:US\B^n \to\B^j\times US\B^{j-n}\\
    a&\mapsto
       \left\{
       \begin{array}{l@{\qquad}l}
         \prod\limits_{h=1}^j\ket{b_h}\times\sum\limits_{i=1}^n\alpha_i.\left( \prod\limits_{h=j+1}^n\ket{b_{ih}}\right)
         &\textrm{if }
           a=\sum\limits_{i=1}^n\alpha_i.\left( \prod\limits_{h=1}^j\ket{b_h}\times\prod\limits_{h=j+1}^n\ket{b_{ih}}\right)\\
         \ket 0^n & \textrm{otherwise}
       \end{array}
                    \right.
  \end{align*}
\end{example}
\begin{defin}\label{def:projk}
  For each $k=0,\dots,2^j-1$, the projection to the $\ket k$ basis vector,
  $\pi_{jk}$, is defined as any arrow making the following diagram commute:
  \begin{center}
    \begin{tikzcd}[column sep=2cm]
      US\B^n \ar[d,"\pi_{jk}"] \ar[r,"U\s{P}_k"]
      &US\B^n \ar[d,"\s{Norm}"]\\
      \B^j\times US\B^{n-j}
      & US\B^n \ar[l,"\varphi_j"]
    \end{tikzcd}
  \end{center}
  where we are implicitly using the isomorphism $US\B^n \cong U(S\B)^{\otimes
    n}$, obtained by composing $n-1$ times the mediating arrow $m$ and
  then applying the functor $U$.
\end{defin} 

 The following distribution map is needed to assemble the final distribution of
 projections in Definition~\ref{def:pi}.
\begin{defin}\label{def:d}
  Let $\{p_i\}_{i=1}^n$ be a set with $p_i\in[0,1]$ such that $\sum_{i=1}^np_i=1$.
  Then, we define ${d_{\{p_i\}_i}}$ as the arrow:
  \[
    d_{\{p_i\}_i}:A^n\to DA
    \qquad\qquad
    (a_1,\dots,a_n)\mapsto\sum_{i=1}^np_i\chi_{a_i}
  \]
\end{defin}
\begin{example}
  Consider $d_{\{\frac 12,\frac 13,\frac 16\}}:\B^3\to D\B^3$ defined by
  $d_{\{\frac 12,\frac 13,\frac 16\}}(b_1\times b_2\times b_3)=\frac
  12\chi_{b_1}+\frac 13\chi_{b_2}+\frac 16\chi_{b_3}$.
  Then, for example, $d_{\{\frac 12,\frac 13,\frac 16\}}\ket{101}=\frac 12\chi_{\ket
    1}+\frac 13\chi_{\ket 0}+\frac 16\chi_{\ket 1}$.
\end{example}

\begin{defin}\label{def:pi}
  The projective arrow is as follows,
  where
  $p_k=\overline{\s{Norm}(\ket\psi)}^\dagger\circ
  \s{P}_k\circ\overline{\s{Norm}(\ket\psi)}$.
  \[
  \pi_j : US\B^n \to D(\B^j\times US\B^{n-j})
  \qquad\qquad
  \ket\psi \mapsto\sum_{k=0}^{2^j-1} p_k\chi_{\pi_{jk}\ket\psi}
  \]
\end{defin}

\begin{example}
  Consider the set $\B^2$ and the vector space $S\B^2$. We can describe the
  projection $\pi_1$ as the map
  $\pi_1:US\B^2\to D(\B\times US\B )$ such that $\ket\psi \mapsto p_0\chi_{\pi_{10}\ket\psi}+p_1\chi_{\pi_{11}\ket\psi}$,
  where, if
  $\ket\psi=\alpha_1.\ket{00}+\alpha_2.\ket{01}+\alpha_3.\ket{10}+\alpha_4.\ket{11}$,
  then
  $p_0=\tfrac{|\alpha_1|^2+|\alpha_2|^2}{\sqrt{\sum_{i=1}^4|\alpha_i|^2}}$ and
  $p_1=\tfrac{|\alpha_3|^2+|\alpha_4|^2}{\sqrt{\sum_{i=1}^4|\alpha_i|^2}}$.

  The normalizing arrow is the arrow  $\s{Norm}:US\B^2\to US\B^2$ such that:
  \begin{align*}
    \alpha_1.\ket{00}&+\alpha_2.\ket{01}+\alpha_3.\ket{10}+\alpha_4.\ket{11}
    \\
    &\mapsto
    \tfrac{\alpha_1}{\sqrt{\sum_{i=1}^4|\alpha_i|^2}}.\ket{00}
    +
    \tfrac{\alpha_2}{\sqrt{\sum_{i=1}^4|\alpha_i|^2}}.\ket{01}
    +
    \tfrac{\alpha_3}{\sqrt{\sum_{i=1}^4|\alpha_i|^2}}.\ket{10}
    +
    \tfrac{\alpha_4}{\sqrt{\sum_{i=1}^4|\alpha_i|^2}}.\ket{11}
  \end{align*}

  The factorisation arrow is the arrow $\varphi_1:US\B^2\to\B\times US\B $
  such that:
  \[
    \alpha_1.\ket{00}+\alpha_2.\ket{01}+\alpha_3.\ket{10}+\alpha_4.\ket{11}\mapsto
    \left\{
      \begin{array}{ll}
        \ket 0\times(\alpha_1.\ket 0+\alpha_2.\ket 1)&\textrm{ if }\alpha_3=\alpha_4=0\\
        \ket 1\times(\alpha_3.\ket 0+\alpha_4.\ket 1)&\textrm{ if }\alpha_1=\alpha_2=0\\
        \ket{00} &\textrm{ otherwise}
      \end{array}
    \right.
  \]

  Finally, $\pi_{10}$ and $\pi_{11}$ are defined as
  maps in $US\B^2\to\B\times US\B $ 
  such that
  \(
    \pi_{10}=
    \varphi_1\circ\s{Norm}\circ  U\s{P}_0
    \) and 
    \(
    \pi_{11}=
    \varphi_1\circ\s{Norm}\circ U\s{P}_1
  \).
\end{example}

We write $(US)^mA$ for $US\dots USA$, where $m>0$. The arrow sum
on $(US)^mA$ with $A\neq USB$ will use the underlying sum in the vector space
$SA$.
To define such a sum, we need the following map.
\begin{defin}
  The map $g_k:\left((US)^{k+1}A\right)\times \left((US)^{k+1}A\right)\to (US)^k(USA\times USA)$
  is defined by:
  \[
    g_k=
    (US)^{k-1}Um\circ (US)^{k-1}n\circ 
    (US)^{k-2}Um\circ (US)^{k-2}n\circ 
    \dots\circ
    Um\circ n 
  \]
  \begin{center}
  \begin{tikzcd}[column sep=1in]
    \left((US)^{k+1}A\right)\times \left((US)^{k+1}A\right)\ar[d,"n"]\ar[r,"g_k"] & (US)^k(USA\times USA)\\
    U(S(US)^kA\otimes S(US)^kA)\ar[d,"Um"] & (US)^{k-1}U(SUSA\otimes SUSA)\ar[u,"(US)^{k-1}Um"]\\
    US((US)^kA\times (US)^kA)\ar[d,"USn"] & (US)^{k-1}((US)^2A\times (US)^2A)\ar[u,"(US)^{k-1}n"] \\
    USU(S(US)^{k-1}A\otimes S(US)^{k-1}A)\ar[r,"USUm"] & (US)^2((US)^{k-1}A\times (US)^{k-1}A)\ar[u,"\vdots",dotted]
  \end{tikzcd}
  \end{center}
\end{defin}
\begin{example}
  We can define the sum on $(US)^3A\times (US)^3A$ by using the sum on
  $SA$ as:\\
    $(US)^3A\times (US)^3A\xlra{g_2}(US)^2(USA\times USA)\xlra{(US)^2+}(US)^3A$\quad
  where
  $g_2=USUm\circ USn\circ Um\circ n$.
\end{example}

Using all the previous definitions, we can finally give the interpretation of a
type derivation tree in our model.
If $\Gamma\vdash t:A$ with a derivation $T$, we
write generically $\sem T$ as
$\Gamma\xlra t A$.
In the following definition, we write $S^mA$ for $S\dots SA$, where $m>0$ and
$A\neq SB$.

\begin{defin}
  If $T$ is a type derivation tree, we define inductively $\sem{T}$ as an arrow in the
  category $\mathbf{Set}_D$, as follows. To avoid cumbersome notation, we omit
  to write the monad $D$ in most cases (we only give it in the case of the
  measurement, which is the only interesting case).
  \begin{align*}
    &\sem{\vcenter{\infer[^\tax]{\Theta^\B,x:\Psi\vdash x:\Psi}{}}} = \Theta^\B\times\Psi\xlra{{!}\times\Id}\One\times\Psi\approx\Psi\quad\textrm{where }\Id\textrm{ is the identity in }\mathbf{Set}\\
    &\sem{\vcenter{\infer[^{\tax_{\vec 0}}]{\Theta^\B\vdash\z:SA}{}}} = \Theta^\B\xlra{!}\One\xlra{\hat{\vec 0}}USA\quad\textrm{where }\hat{\vec 0}\textrm{ is the constant function $\star\mapsto\vec 0$}\\
    &\sem{\vcenter{\infer[^{\tax_{\ket 0}}]{\Theta^\B\vdash\ket 0:\B}{}}}=\Theta^\B\xlra{!}\One\xlra{\hat{\ket 0}}\B\quad\textrm{where }\hat{\ket 0}\textrm{ is the constant function $\star\mapsto\ket 0$}\\
    &\sem{\vcenter{\infer[^{\tax_{\ket 1}}]{\Theta^\B\vdash\ket 1:\B}{}}}=\Theta^\B\xlra{!}\One\xlra{\hat{\ket 1}}\B\quad\textrm{where }\hat{\ket 1}\textrm{ is the constant function $\star\mapsto\ket 1$}\\
    &\sem{\vcenter{\infer[^{\alpha_I}]{\Gamma\vdash \alpha.t:S^mA}{\Gamma\vdash t:S^mA}}} =
      \begin{aligned}[t]
        &\Gamma\xlra{t}(US)^mA\xlra{(US)^{m-1}U\lambda} (US)^{m-1}U(SA\otimes\I)\\
        &\xlra{(US)^{m-1}U(\Id\otimes\alpha)} (US)^{m-1}U(SA\otimes\I) \xlra{(US)^{m-1}U\lambda^{-1}} (US)^mA
      \end{aligned}\\
    &\sem{\vcenter{\infer[^{+_I}]{\Gamma,\Delta,\Theta^\B\vdash t+r:S^mA}{\Gamma,\Theta^\B\vdash t:S^mA &  \Delta,\Theta^\B\vdash r:S^mA}}}
     =
     \begin{aligned}[t]
       &\Gamma\times\Delta\times\Theta^\B\xlra{\Id\times\delta}\Gamma\times\Delta\times\Theta^\B\times\Theta^\B\\
       &\xlra{\Id\times\sigma\times\Id}\Gamma\times\Theta^\B\times\Delta\times\Theta^\B\xlra{t\times r} (US)^mA\times (US)^mA\\
     &\xlra {g_{m-1}} (US)^{m-1}(USA\times USA)\xlra{(US)^{m-1}+} (US)^mA
     \end{aligned}\\
    &\sem{\vcenter{\infer[^{S_I}]{\Gamma\vdash t:SA}{\Gamma\vdash t:A}}}=\Gamma\xlra{t}A\xlra{\eta}USA\\
    &\sem{\vcenter{\infer[^{S_E}]{\Gamma\vdash\pi_j t:\B^j\times S\B^{n-j}}{\Gamma\vdash t:S^k\B^n}}} = \Gamma\xlra t (US)^k\left(\B^n\right)\xlra{\mu^{k-1}}US\B^n \xlra{\pi_j}D(\B^j\times S\left(\B^{n-j}\right))\\
    &\sem{\vcenter{\infer[^{\tif}]{\Gamma\vdash\ite{}{t}{r}:\B\Rightarrow A} {\Gamma\vdash t:A & \Gamma\vdash r:A}}}= \Gamma\xlra{\mathsf{curry}(\mathsf{uncurry}(f_{t,r})\,\circ\,\mathsf{swap})}[\B,A]\\
    &\sem{\vcenter{\infer[^{\Rightarrow_I}] {\Gamma\vdash\lambda x{:}\Psi.t:\Psi\Rightarrow A} {\Gamma,x:\Psi\vdash t:A}}} =\Gamma\xlra{\eta^\Psi}[\Psi,\Gamma\times\Psi]\xlra{[\Id,t]}[\Psi,A]\\
    &\sem{\vcenter{\infer[^{\Rightarrow_E}]{\Delta,\Gamma,\Theta^\B\vdash tu:A} {\Delta,\Theta^\B\vdash u:\Psi & \Gamma,\Theta^\B\vdash t:\Psi\Rightarrow A}}} =
    \begin{aligned}[t]
      &\Delta\times\Gamma\times\Theta^\B\xlra{\Id\times\delta}\Delta\times\Gamma\times\Theta^\B\times\Theta^\B\\
      &\xlra{\Id\times\sigma\times\Id} \Delta\times\Theta^\B\times\Gamma\times\Theta^\B\xlra{u\times t}\Psi\times[\Psi,A]
      \xlra{\varepsilon^\Psi}A
    \end{aligned}\\
    &\sem{\vcenter{\infer[^{\Rightarrow_{ES}}]{\Delta,\Gamma,\Theta^\B\vdash
      tu:SA} {\Delta,\Theta^\B\vdash u:S\Psi & \Gamma,\Theta^\B\vdash t:S(\Psi\Rightarrow A)}}} =
      \begin{aligned}[t]
        &\Delta\times\Gamma\times\Theta^\B\xlra{\Id\times\delta}\Delta\times\Gamma\times\Theta^\B\times\Theta^\B\\
        &\xlra{\Id\times\sigma\times\Id}\Delta\!\times\!\Theta^\B\!\times\!\Gamma\!\times\!\Theta^\B\xlra{u\times t}US\Psi\times US[\Psi,A]\\
        &\xlra{n} U(S\Psi\otimes S[\Psi,A])\xlra{Um} US(\Psi\times[\Psi,A])\\
        &\xlra{US\varepsilon^\Psi}USA
      \end{aligned}\\
    &\sem{\vcenter{\infer[^{\times_I}]{\Gamma,\Delta,\Theta^\B\vdash t\times u:\Psi\times \Phi} {\Gamma,\Theta^\B\vdash t:\Psi & \Delta,\Theta^\B\vdash u:\Phi}}} =
                                                                                                                           \begin{aligned}[t]
        &\Gamma\times\Delta\times\Theta^\B\xlra{\Id\times\delta}\Gamma\times\Delta\times\Theta^\B\times\Theta^\B\\
        &\xlra{\Id\times\sigma\times\Id}\Gamma\times\Theta^\B\times\Delta\times\Theta^\B\xlra{t\times u} \Psi\times \Phi
      \end{aligned}\\
    &    \sem{\vcenter{\infer[^{\times_{Er}}]{\Gamma\vdash\head\ t:\B}{\Gamma\vdash t:\B^n}}} =\Gamma\xlra t\B^n\xlra{\head}\B\quad{\parbox{7.7cm}{where $\head$ is the projector of the first component in \textbf{Set}}}\\
    &    \sem{\vcenter{\infer[^{\times_{E\ell}}]{\Gamma\vdash\tail\ t:\B^{n-1}}{\Gamma\vdash t:\B^n}}} =\Gamma\xlra t\B^n\xlra{\tail}\B^{n-1}\quad{\parbox{7.07cm}{where $\tail$ is the projector of the $n-1$ last components}}\\
    &    \sem{\vcenter{\infer[^{\Uparrow_r}]{\Gamma\vdash\Uparrow_r  t:S(\Psi\times \Phi)} {\Gamma\vdash t:S(S\Psi\times \Phi)}}} =
      \begin{aligned}[t]
        &\Gamma\xlra t US(US\Psi\times \Phi)\xlra{U(\Id\times\eta)} US(US\Psi\times US\Phi)\\
        &\xlra{USn} USU(S\Psi\otimes S\Phi)\xlra{USUm}USUS(\Psi\times \Phi) \xlra{\mu} US(\Psi\times \Phi)
      \end{aligned}\\
    &    \sem{\vcenter{\infer[^{\Uparrow_\ell}]{\Gamma\vdash\Uparrow_\ell  t:S(\Psi\times \Phi)} {\Gamma\vdash t:S(\Psi\times S\Phi)}}}  =
      \begin{aligned}[t]
        &\Gamma\xlra t US(\Psi\times US\Phi)\xlra{U(\eta\times\Id)} US(US\Psi\times US\Phi)\\
        &\xlra{USn} USU(S\Psi\otimes S\Phi)\xlra{USUm}USUS(\Psi\times \Phi) \xlra{\mu} US(\Psi\times \Phi)
      \end{aligned}\\
    &\sem{\vcenter{\infer[^\parallel]{\Gamma\vdash\proba{p_1}t_1\parallel\cdots\parallel\proba{p_n}t_n:A}{\Gamma\vdash t_i:A & \sum_i\mathtt{p_i}=1}}} = \Gamma\xlra{\delta}\Gamma^n\xlra{t_1\times\dots\times t_n}A^n\xlra{d_{\{{p_i}\}_i}}DA
  \end{align*}
\end{defin}

\begin{prop}
  [Independence of derivation]
  \label{prop:eqDer}
  If $\Gamma\vdash t:A$ can be derived with two different derivations $T$ and
  $T'$, then $\sem{T}=\sem{T'}$.
\end{prop}
\begin{proof}
  Without taking into account rules $\Rightarrow_E$, $\Rightarrow_{ES}$, and
  $S_I$, the typing system is syntax directed. Hence, we give a rewrite system
  on trees such that each time a rule $S_I$ can be applied before or after
  another rule, we choose a direction to rewrite the tree to one of these forms.
  Similarly, rules $\Rightarrow_E$ and $\Rightarrow_{ES}$, can be exchanged in
  few specific cases, so we also choose a direction for these.

  Then, we prove that every rule preserves the semantics of the tree. This
  rewrite system is clearly confluent and normalizing, hence for each tree $T$
  we can take the semantics of its normal form, and so every sequent will have
  one way to calculate its semantics: as the semantics of the normal tree.
  The full proof is given in the appendix.
\end{proof}

\begin{rmk}
  Proposition~\ref{prop:eqDer} allows us to write the semantics of a sequent,
  independently of its derivation tree. Hence, from now on, we will use
  $\sem{\Gamma\vdash t:A}$, without ambiguity.
\end{rmk}
\subsection{Soundness and Adequacy}
We first prove the soundness of the interpretation with respect to the reduction relation (Theorem~\ref{thm:soundness}), then we prove the computational adequacy (Theorem~\ref{thm:CA}). Finally, we prove
{adequacy} (Theorem~\ref{thm:FA}) as a consequence of both results.

\subsubsection{Soundness}
Soundness is proved only for closed terms, since the reduction is weak (cf.~Figure~\ref{fig:TRScontext}).
First, we need a substitution lemma.
\begin{lem}
  [Substitution]\label{lem:substitution}
  If $x:\Psi\vdash t:A$ and $\vdash r:\Psi$, the following diagram
  commutes:
  \begin{center}
    \begin{tikzcd}[row sep=9pt]
      \One\ar[rr,"(r/x)t"]\ar[rd,"r"] && A\\
      &\Psi\ar[ur,"t"]
    \end{tikzcd}
  \end{center}
  That is,
  $\sem{\vdash(r/x)t:A}=\sem{x:\Psi\vdash t:A}\circ\sem{\vdash r:\Psi}$.
\end{lem}
\begin{proof}
 We prove, more generally, that 
 if $\Gamma',x:\Psi,\Gamma\vdash t:A$ and $\vdash r:\Psi$, the following diagram
 commutes:
 \begin{center}
   \begin{tikzcd}[row sep=9pt]
     \Gamma'\times\Gamma\ar[r,"(r/x)t"]\ar[d,dashed,"\approx"] & A\\
     \Gamma'\times\One\times\Gamma\ar[r,"\Id\times r\times\Id"] &\Gamma'\times\Psi\times\Gamma\ar[u,"t"]
   \end{tikzcd}
 \end{center}
 That is,
 $\sem{\Gamma',\Gamma\vdash(r/x)t:A}=\sem{\Gamma',x:\Psi,\Gamma\vdash t:A}\circ(\Id\times\sem{\vdash r:\Psi}\times\Id)$.
 Then, by taking $\Gamma=\Gamma'=\emptyset$, we get the result stated by the lemma.
 
 We proceed by induction on the derivation of $\Gamma',x:\Psi,\Gamma\vdash t:A$.
  The full proof is given in the appendix.
\end{proof}

\begin{thm}
  [Soundness]\label{thm:soundness}
  If $\vdash t:A$, and $t\lra r$,
  then
  $\sem{\vdash t:A} = \sem{\vdash r:A}$.
\end{thm}
\begin{proof}
  By induction on the rewrite relation, using the first derivable type for each
  term.
  The full proof is given in the appendix.
\end{proof}

\subsubsection{Computational adequacy}
We adapt Tait's proof for strong normalization to prove the computational
adequacy of \LambdaS.
\begin{defin}
  Let $\mathfrak A,\mathfrak B$ be sets of closed terms. We define the following operators on them:

  \begin{itemize}
  \item\textit{Closure by antireduction:} $\overline{\mathfrak A}=\{t\mid
    t\lra^*r_, \textrm{ with }r\in \mathfrak A \textrm{ and }FV(t)=\emptyset\}$.
  \item\textit{Closure by distribution:} $\mathfrak
    A^\parallel=\{\bigparallel_i\proba{p_i}t_i\mid t_i\in\mathfrak A\textrm{ and }\sum_i\mathtt{p_i}=1\}$.
  \item\textit{Product:} $\mathfrak A\times \mathfrak B=\{t\times u\mid t\in \mathfrak
    A\textrm{ and }u\in \mathfrak B\}$.
  \item\textit{Arrow:} $\mathfrak A\Rightarrow \mathfrak B=\{t\mid\forall u\in\mathfrak
    A, tu\in \mathfrak B\}$.
  \item\textit{Span:} $S\mathfrak A=\{\sum_i\may[\alpha_i]r_i\mid r_i\in\mathfrak A\}$.
  \end{itemize}
  
  The set of computational closed terms of type $A$ (denoted $\com A$), is defined by: 
  \[
      \com{\B} =\overline{\{\ket 0,\ket 1\}}^\parallel\qquad
      \com{A\times B}=\overline{\com A\times\com B}^\parallel\qquad
      \com{\Psi\Rightarrow A} =\overline{\com\Psi\Rightarrow\com A}^\parallel\qquad
      \com{SA} =\overline{S\com A\cup\{\z\}}^\parallel
  \]
\end{defin}
  A substitution $\sigma$ is valid in a context $\Gamma$ (notation
  $\sigma\vDash\Gamma$) if for each $x:A\in\Gamma$, $\sigma x\in\com A$.

\begin{lem}\label{lem:Adequacy}
  If $\vdash t:A$ then $t\in\com A$.
\end{lem}
\begin{proof}
  We prove, more generally, that if $\Gamma\vdash t:A$ and $\sigma\vDash\Gamma$, then $\sigma t\in\com A$.
  We proceed by induction on the derivation of $\Gamma\vdash t:A$.
  The full proof is given in the appendix.
\end{proof}

\begin{defin}[Elimination context]
  An elimination context is a term of type $\B$ produced by the following grammar, where exactly one subterm has been replaced
  with a hole $[\cdot]$.
  \[
    C:=[\cdot]\mid Ct \mid tC\mid \pi_jC\mid\head\ C\mid\tail\ C\mid\Uparrow_r
    C\mid\Uparrow_\ell C
  \]
  We write $C[t]$ for the term of type $\B$ obtained from replacing the hole of
  $C$ by $t$.  
\end{defin}
\begin{defin}[Operational equivalence]
  We write $t\elimeq r$ if, for every elimination context $C$, there
  exists $s$ such that $C[t]\lra^* s$ and $C[r]\lra^* s$.

  We define the operational equivalence $\opeq$ inductively by
  \begin{itemize}
  \item If $t\elimeq r$, then $t\opeq r$.
  \item If $t\opeq r$ then $\alpha.t\opeq\alpha.r$.
  \item If $t_1\opeq r_1$ and $t_2\opeq r_2$, then $t_1+t_2\opeq r_2+r_2$.
  \item If $t_1\opeq r_1$ and $t_2\opeq r_2$, then $t_1\times t_2\opeq
    r_1\times r_2$.
  \end{itemize}
\end{defin}
{Remark that operational equivalence differ from the standard notion of
  observational equivalence since $t\opeq r$ does not imply $\lambda
  x:\Psi.t\opeq\lambda x:\Psi.r$, as a consequence of not having reductions
  under lambda.}

\begin{lem}\label{lem:contextApp}~
  If $C[t]\opeq C[r]$, then $t\opeq r$.
\end{lem}
\begin{proof}
  By the shape of $C$, the only possibility for $C[t]\opeq C[r]$ is
$C[t]\elimeq C[r]$. Then, by definition, there exists a term $s$ and a context
$D$ such that $D[C[t]]\lra^*s$ and $D[C[r]]\lra^*s$. Consider the context $E =
D[C]$, we have $E[t]=D[C[t]]\lra^*s$ and $E[r]=D[C[r]]\lra^*t'$. Therefore,
$t\elimeq r$, and so $t\opeq r$.
\end{proof}

\begin{thm}[Computational adequacy]\label{thm:CA}
  If $\sem{\vdash t:A}=\sem{\vdash v:A}$,
  then $t\opeq v$.
\end{thm}
 \begin{proof}
   We proceed by induction on $A$.
   \begin{itemize}
     \item $A=\B$.
  By Lemma~\ref{lem:Adequacy}, we have $t\in\com{A}$, thus,
  $t\lra^*\proba{q_1}\ket 0\parallel\proba{q_2}\ket 1$, and, by the same lemma, $v=\proba{p_1}\ket 0\parallel\proba{p_2}\ket 1$.
  Hence, by Theorem~\ref{thm:soundness}, we have $\sem{\vdash v:A}=\sem{\vdash
    t:A}=\sem{\vdash\proba{q_1}\ket 0\parallel\proba{q_2}\ket 1:A}$. So,
  $1^2\xlra{\ket 0\times\ket 1}\B^2\xlra{d_{\{{p_1},{p_2}\}}}D\B=
  1^2\xlra{\ket 0\times\ket 1}\B^2\xlra{d_{\{{q_1},{q_2}\}}}D\B$.
  Therefore, $\mathtt{p_i}=\mathtt{q_i}$, thus $t\lra^* v$.
  
\item $A=C_1\times C_2$.
     By Lemma~\ref{lem:Adequacy}, we have $t\in\com{A}$, thus,
     $t\lra^*\bigparallel_i\proba{q_i}(w_{i1}\times w_{i2})$, with $w_{ij}\in \com{C_j}$, and, by the same lemma,
     $v=\bigparallel_i\proba{p_i}(v_{i1}\times v_{i2})$, with $v_{ij}\in \com{C_j}$.
  Hence, by Theorem~\ref{thm:soundness}, we have $\sem{\vdash v:A}=\sem{\vdash
    t:A}=\sem{\vdash \bigparallel_i\proba{q_i}(w_{i1}\times w_{i2}):A}$. So,
  $(1^2)^n\xlra{ (w_{11}\times w_{12})\times\dots\times (w_{n1}\times
    w_{n2})}(C_1\times C_2)^n\xlra{d_{\{{q_i}\}_i}}D(C_1\times C_2)
  =(1^2)^n\xlra{(v_{11}\times v_{12})\times\dots\times (v_{m1}\times
    v_{m2})}(C_1\times C_2)^n\xlra{d_{\{{p_i}\}_i}}D(C_1\times C_2)$.
  Therefore,
  $\mathtt{p_i}=\mathtt{q_i}$, $m=n$, and $\sem{\vdash v_{ij}:C_j}=\sem{\vdash w_{ij}:C_j}$.
  Therefore, by the induction hypothesis, $w_{ij}\opeq v_{ij}$, and so, $t\opeq v$.
  
   \item $A=\Psi\Rightarrow C$.
     The only possibility for $v$, a value of type $\Psi\Rightarrow C$, is
     $v=\bigparallel_i\proba{p_i}\lambda x^\Psi.r_i$.

     Hence, let $f=\sem{\vdash t:A}=\sem{\vdash v:A}=
     1^n\xlra{(\eta^\Psi)^n}\home{\Psi}{1\times\Psi}^n\xlra{\home{\Id}{r_i}^n}\home\Psi
     C^n\xlra{d_{\{{p_i}\}_i}}D\home\Psi C$.
     
     By Lemma~\ref{lem:Adequacy}, we have $t\in\com{A}$. Hence, $t\lra^* t'$,
     such that for all $s\in\com\Psi$, $t's\in\com C$.

     Let $w\in\com\Psi$ be a value, and $g=\sem{\vdash w:\Psi}=1\xlra
     w\Psi\xlra{d_{\{1\}}}D\Psi$.
     
     Thus, $\sem{\vdash tw:C}=\sem{\vdash vw:C}=
     1^{n+1}\xlra{f\times g} D\home{\Psi}C\times
     D\Psi\xlra{m_D}D(\home{\Psi}C\times\Psi)\xlra{D\varepsilon} DC$.
     
     By Theorem~\ref{thm:SN}, and Theorem~\ref{thm:progress}, there exists $u$
     value, such that $vw\lra^*u$, and by Theorem~\ref{thm:SR}, $\vdash u:C$.
     So, by Theorem~\ref{thm:soundness}, $\sem{\vdash u:C}=\sem{\vdash vw:C}$.

     Therefore, by the induction hypothesis, $tw\opeq u$. Since $vw\lra^* u$,
     we have $u\opeq vw$. Hence, $tw\opeq vw$, and so, by Lemma~\ref{lem:contextApp},  $t\opeq v$.
   \item $A=SC$.
     By Lemma~\ref{lem:Adequacy}, we have $t\in\com{A}$, thus, $t\lra^*\bigparallel_i\proba{q_i}\sum_j\alpha_jw_{ij}$, with $w_{ij}\in \com C$, and, by the same lemma,
  $v=\bigparallel_i\proba{p_i}\sum_k\beta_kv_{ik}$, with $v_{ik}\in \com C$.
  Hence, by Theorem~\ref{thm:soundness}, we have $\sem{\vdash v:A}=\sem{\vdash
    t:A}=\sem{\vdash \bigparallel_i\proba{q_i}\sum_j\alpha_j w_{ij}:A}$. So,
  $d_{\{{q_i}\}_i}\circ US+\circ w_{11}\times\dots\times w_{nm} =
  d_{\{{p_i}\}_i}\circ US+ \circ v_{11}\times\dots\times v_{n'm'}$.
  Therefore,
  $\mathtt{p_i}=\mathtt{q_i}$, $m=m'$, $n=n'$ and $\sem{\vdash w_{ij}:C}=\sem{\vdash v_{ij}:C}$.
  Therefore, by the induction hypothesis, $w_{ij}\opeq v_{ij}$, and so, $t\opeq v$.
  \qedhere
   \end{itemize}
\end{proof}

\subsubsection{{Adequacy}}
{Adequacy} is a consequence of 
Theorems~\ref{thm:SR} (subject reduction), \ref{thm:SN} (strong normalization),
\ref{thm:progress} (progress), \ref{thm:soundness} (soundness), and \ref{thm:CA}
(computational adequacy).

\begin{thm}[{Adequacy}]\label{thm:FA}
  If $\sem{\vdash t:A}=\sem{\vdash r:A}$, then $t\opeq r$.
\end{thm}
\begin{proof}
  By Theorem~\ref{thm:SN}, $t$ is strongly normalizing, and by
  Theorem~\ref{thm:progress}, it normalizes to a value. Hence, there exists $v$ such
  that $t\lra^*v$, and, by Theorem~\ref{thm:SR}, we have $\vdash v:A$.
  By Theorem~\ref{thm:soundness}, $\sem{\vdash v:A}=\sem{\vdash t:A}=\sem{\vdash
    r:A}$.
  Then, by Theorem~\ref{thm:CA}, $v\opeq t$ and $v\opeq r$. Hence,
  $t\opeq r$.
\end{proof}

\section{Conclusion}\label{sec:conclusion}
We have revisited the concrete categorical semantics for \LambdaS presented in
our LSFA'18 paper \citep{DiazcaroMalherbeLSFA18} by slightly modifying the
operational semantics of the calculus, obtaining {an adequate model} (Theorem~\ref{thm:FA}).

Our semantics highlights the dynamics of the calculus: the algebraic rewriting
(linear distribution, vector space axioms, and typing casts rules) emphasize the
standard behaviour of vector spaces. The natural transformation $n$ takes
these arrows from the Cartesian category $\mathbf{Set}$ to the tensorial
category $\mathbf{Vec}$, where such a behaviour occurs naturally, and then are
taken back to the Cartesian realm with the natural transformation $m$. This way,
rules such as $\rlinr$: $t(u+v)\lra tu+tv$, are simply considered as $Um\circ n$
producing $(u+v,t)\mapsto (u,t)+(v,t)$ in two steps: $(u+v,t) \mapsto
(u+v)\otimes t=u\otimes t+v\otimes t \mapsto (u,t)+(v,t)$, using the fact that,
in $\mathbf{Vec}$, 
$(u+v)\otimes t=u\otimes t+v\otimes t$.

We have constructed a concrete mathematical semantic model of \LambdaS based on
a monoidal adjunction with some extra conditions. The construction
depends crucially on inherent properties of the categories of set and vector
spaces. In a future work we will study the semantics from a more abstract point
of view. Our approach will be based on recasting the concrete model at a more
abstract categorical level of monoidal categories with some axiomatic properties
that are now veiled in the concrete model. Some of these properties, such as to
consider an abstract dagger instead of an inner product, were introduced in the
concrete model from the very beginning, but others are described in
Remark~\ref{rmk:categoria} and
Definitions~\ref{def:if},~\ref{def:norm},~\ref{def:fact},~\ref{def:projk},~\ref{def:d},
and~\ref{def:pi}. Another question we hope to address in future work is the
exact categorical relationship between the notion of amplitude and probability
in the context of the abstract semantics. While some research has been done in
 this topic (e.g.,~\citep{SelingerQPL05,AbramskyCoeckeLICS04}) it differs from
 our point of view in some important aspects: for example to consider a notion of
 abstract normalization as primitive.

\bibliographystyle{abbrvnat}
\bibliography{references}

\appendix
\section{Detailed proofs}\label{ap:appendix}
\xrecap{Proposition}{Independence of derivation}{prop:eqDer}{
  If $\Gamma\vdash t:A$ can be derived with two different derivations $T$ and
  $T'$, then $\sem{T}=\sem{T'}$.
}
\begin{proof}
  Without taking into account rules $\Rightarrow_E$, $\Rightarrow_{ES}$, and
  $S_I$, the typing system is syntax directed. Hence, we give a rewrite system
  on trees such that each time a rule $S_I$ can be applied before or after
  another rule, we choose a direction to rewrite the tree to one of these forms.
  Similarly, rules $\Rightarrow_E$ and $\Rightarrow_{ES}$, can be exchanged in
  few specific cases, so we also choose a direction for these.

  Then, we prove that every rule preserves the semantics of the tree. This
  rewrite system is clearly confluent and normalizing, hence for each tree $T$
  we can take the semantics of its normal form, and so every sequent will have
  one way to calculate its semantics: as the semantics of the normal tree.

  In order to define the rewrite system, we first analyse the typing rules
  containing only one premise, and check whether these rules allow for a
  previous and posterior rule $S_I$. 
  If both are allowed, we choose a direction for the rewrite rule. Then we continue with rules with more than one
  premise and check under which conditions a commutation of rules is possible,
  choosing also a direction.

  \noindent {Rules with one premise:}
  \begin{itemize}
  \item Rule $\alpha_I$:
    \begin{align}
      \label{rule:alpha-SI}
      \vcenter{\infer[^{\alpha_I}]{\Gamma\vdash\alpha.t:SSA}{
      \infer[^{S_I}]{\Gamma\vdash t:SSA}{\Gamma\vdash t:SA}
      }}
      &\lra
        \vcenter{\infer[^{S_I}]{\Gamma\vdash\alpha.t:SSA}{
        \infer[^{\alpha_I}]{\Gamma\vdash\alpha.t:SA}{\Gamma\vdash t:SA}
        }}
    \end{align}
  \item Rules $S_E$, $\Rightarrow_I$, $\times_{E_r}$, $\times_{E_\ell}$,
    $\Uparrow_r$, and $\Uparrow_\ell$: These rules end with a specific
    types not admitting two $S$ in the head position (i.e.~$\B^j\times S\B^{n-j}$,
    $\Psi\Rightarrow A$, $\B$, $\B^{n-1}$, and $S(\Psi\times\Phi)$) hence removing an $S$ or adding an $S$ would not allow the rule to be applied, and hence, these rules followed or
    preceded by $S_I$ 
 cannot commute.
  \end{itemize}
  {Rules with more than one premise:}
  \begin{itemize}
  \item Rule $+_I$:
    \begin{align}\label{rule:sum-SI}
      \vcenter{
      \infer[^{+_I}]{\Gamma,\Delta,\Xi^\B\vdash(t+u):SSA}{
      \infer[^{S_I}]{\Gamma,\Xi^\B\vdash t:SSA}{\Gamma,\Xi^\B\vdash t:SA}
      &
        \infer[^{S_I}]{\Delta,\Xi^\B\vdash u:SSA}{\Delta,\Xi^\B\vdash r:SA}
        }
        }
      &  \lra
        \vcenter{
        \infer[^{S_I}]{\Gamma,\Delta,\Xi^\B\vdash(t+u):SSA}{
        \infer[^{+_I}]{\Gamma,\Delta,\Xi^\B\vdash (t+u):SA}{
        \Gamma,\Xi^\B\vdash t:SA & \Delta,\Xi^\B\vdash u:SA
                                     }
                                     }
                                     }
    \end{align}
  \item Rules $\Rightarrow_E$ and $\Rightarrow_{ES}$:
    \begin{align}\label{rule:arrow-SI}
      \vcenter{\infer[^{\Rightarrow_{ES}}]{{\Delta,\Gamma,\Xi^\B\vdash tu:SA}}{
      \infer[^{S_I}]{\Delta,\Xi^\B\vdash u:S\Psi}{\Delta,\Xi^\B\vdash u:\Psi}
      &
        \infer[^{S_I}]{\Gamma,\Xi^\B\vdash t:S(\Psi\Rightarrow A)}{\Gamma,\Xi^\B\vdash t:\Psi\Rightarrow A}
        }
        }
      &\lra
        \vcenter{
        \infer[^{S_I}]{\Delta,\Gamma,\Xi^\B\vdash tu:SA}{
        \infer[^{\Rightarrow_E}]{\Delta,\Gamma,\Xi^\B\vdash tu:A}{
        \Delta,\Xi^\B\vdash u:\Psi
      &
        \Gamma,\Xi^\B\vdash t:\Psi\Rightarrow A
        }
        }
        }
    \end{align}
  \item Rule $\parallel$:
    \begin{align}\label{rule:parallel-SI}
      \vcenter{\infer[^\parallel]{\Gamma\vdash\proba{p_1}t_1\parallel\dots\parallel\proba{p_n}t_n:SA}{
      \infer[^{S_I}]{\Gamma\vdash t_i:SA}{\Gamma\vdash t_i:A} & \sum_i\mathtt{p_i}=1
                                                                  }
                                                                  }
      &\lra
        \vcenter{
        \infer[^{S_I}]{\Gamma\vdash\proba{p_1}t_1\parallel\dots\parallel\proba{p_n}t_n:SA}{\infer[^{\parallel}]{\Gamma\vdash\proba{p_1}t_1\parallel\dots\parallel\proba{p_n}t_n:A}{\Gamma\vdash t_i:A & \sum_i\mathtt{p_i}=1}}
        }
    \end{align}
  \item Rules $\mathit{If}$ and $\times_I$: these rules end with specific 
    types not admitting two $S$ in the head position (i.e.~$\B\Rightarrow A$ and
    $\Psi\times\Phi$),
    hence removing an $S$ or adding an $S$ would not allow the rule to be
    applied, and hence, these rules followed or preceded by $S_I$ 
    cannot commute.
  \end{itemize}
  The confluence of this rewrite system is easily inferred from the fact that
  there are no critical pairs. The normalization follows from the fact that the
  trees are finite and all the rewrite rules push the $S_I$ 
  to the
  root of the trees.

  It only remains to check that each rule preserves the semantics.
  \begin{itemize}
  \item Rule \eqref{rule:alpha-SI}: The following diagram gives the semantics of
    both trees (we only treat, without loss of generality, the case where $A\neq S(A')$). This diagram commutes by the naturality of $\eta$.
    \begin{center}
      \begin{tikzcd}[column sep=1.1cm]
        \Gamma\ar[r,"t"] &
        USA\ar[d,"U\lambda"]\ar[r,"\eta"]& USUSA\ar[r,"USU\lambda"] &USU(SA\otimes\I)\ar[r,"USU(\Id\otimes\alpha)"] & USU(SA\otimes\I)\ar[d,"USU\lambda^{-1}"]\\
        &U(SA\otimes\I)\ar[r,"U(\Id\otimes\alpha)"]& U(SA\otimes\I)\ar[r,"U\lambda^{-1}"]& USA\ar[r,"\eta"]&USUSA
      \end{tikzcd}
    \end{center}
  \item Rule \eqref{rule:sum-SI}: The following diagram gives the semantics of
    both trees (we only treat, without lost of generality, the case where $A\neq SA'$).
    \begin{center}
      \begin{tikzcd}[row sep=10pt]
        \Gamma\times\Xi^\B\times\Delta\times\Xi^\B\ar[r,"t\times r"] & USA\times USA\ar[r,"\eta\times\eta"]\ar[d,"g_0=\Id"] & USUSA\times USUSA\ar[d,"g_1"] \\
        \Gamma\times\Delta\times\Xi^\B\times\Xi^\B \ar[u,"\Id\times\sigma\times\Id"] &USA\times USA\ar[d,"+"]&US(USA\times USA)\ar[d,"US+"]\\
        \Gamma\times\Delta\times\Xi^\B\ar[u,"\Id\times\delta"]
        &USA\ar[r,"\eta"]&USUSA
      \end{tikzcd}
    \end{center}
    This diagram commutes since the maps are as follows:
    
    $(t,r)\stackrel{\eta\times\eta}{\mapsto} (t,r)\stackrel{g_1}\mapsto (t,r)\stackrel{US+}\mapsto t+r$\quad
    and\quad
    $(t,r)\stackrel{\Id}\mapsto (t,r)\stackrel{+}\mapsto
    t+r\stackrel\eta\mapsto t+r$
  \item Rule \eqref{rule:arrow-SI}: The following diagram gives the semantics
    of both trees. The lower diagram with the dotted arrow commutes by the
    naturality of $\eta$, and the upper diagram commutes because $\eta$ is a
    monoidal natural transformation.
    \begin{center}
      \begin{tikzcd}[row sep=10pt]
        \Delta\times\Xi^\B\times\Gamma\times\Xi^\B\ar[d,"u\times t"] && \Delta\times\Gamma\times\Xi^\B\ar[ll,"(\Id\times\sigma\times\Id)\circ(\Id\times\delta)",swap]\\
        \Psi\times\home\Psi A\ar[d,"\varepsilon^\Psi"]\ar[rrd,dotted,"\eta"]\ar[r,"\eta^2"]&US\Psi\times US(\home\Psi A)\ar[r,"n"] &U(S\Psi\otimes S(\home\Psi A))\ar[d,"U(m)"]\\
        A\ar[r,"\eta"]& USA&US(\Psi\times\home\Psi A)\ar[l,"US(\varepsilon^\Psi)"]\\
      \end{tikzcd}
    \end{center}
  \item Rule \eqref{rule:parallel-SI}: The following diagram gives the
    semantics of both trees.
    \begin{center}
      \begin{tikzcd}[column sep=1.3cm,row sep=10pt]
        \Gamma\ar[r,"\delta"] &\Gamma^n\ar[r,"t_1\times\dots\times t_n"] &
        A^n\ar[r,"\eta^n"]\ar[d,"d_{\{{p_i}\}_i}"] & USA^n\ar[d,"d_{\{{p_i}\}_i}"]\\
        &&DA\ar[r,"\eta"]& USDA=DUSA
      \end{tikzcd}
    \end{center}
    The mappings are as follows:

    $(a_1,\dots,a_n)\stackrel{\eta^n}\mapsto (a_1,\dots,a_n)\stackrel{d_{\{{p_i}\}_i}}\mapsto \sum_i\mathtt{p_i}\chi_{a_i}$
    \ and\ 
    $(a_1,\dots,a_n)\stackrel{d_{\{{p_i}\}_i}}\mapsto\sum_i\mathtt{p_i}\chi_{a_i}\stackrel{\eta}\mapsto\sum_i\mathtt{p_i}\chi_{a_i}$
    \qedhere
  \end{itemize}
\end{proof}
\begin{lem}\label{lem:W}
  If $\Gamma\vdash t:A$, then $\Gamma,\Delta^\B\vdash t:A$. Moreover,
  $\sem{\Gamma,\Delta^\B\vdash t:A}=\sem{\Gamma\vdash t:A}\circ(\Id\times{!})$.
\end{lem}
\begin{proof}
  A derivation of $\Gamma\vdash t:A$ can be turned
  into a derivation $\Gamma,\Delta^\B\vdash t:A$ just by adding $\Delta^\B$
  in its axioms' contexts.
  Since $FV(t)\cap\Delta^\B=\emptyset$, we have
  $\sem{\Gamma,\Delta^\B\vdash t:A}=\sem{\Gamma\vdash t:A}\circ(\Id\times{!})$.
\end{proof}

\xrecap{Lemma}{Substitution}{lem:substitution}
   {If $x:\Psi\vdash t:A$ and $\vdash r:\Psi$, the following diagram
   commutes:}
   \begin{center}
   \begin{tikzcd}[row sep=9pt]
       \One\ar[rr,"(r/x)t"]\ar[rd,"r"] & & A\\
       &\Psi\ar[ur,"t"]&
     \end{tikzcd}
   \end{center}
   That is,
   $\sem{\vdash(r/x)t:A}=\sem{x:\Psi\vdash t:A}\circ\sem{\vdash r:\Psi}$.
\begin{proof}
 We prove, more generally, that 
 if $\Gamma',x:\Psi,\Gamma\vdash t:A$ and $\vdash r:\Psi$, the following diagram
 commutes:
 \begin{center}
   \begin{tikzcd}[row sep=9pt]
     \Gamma'\times\Gamma\ar[r,"(r/x)t"]\ar[d,dashed,"\approx"] & A\\
     \Gamma'\times\One\times\Gamma\ar[r,"\Id\times r\times\Id"] &\Gamma'\times\Psi\times\Gamma\ar[u,"t"]
   \end{tikzcd}
 \end{center}
 That is,
 $\sem{\Gamma',\Gamma\vdash(r/x)t:A}=\sem{\Gamma',x:\Psi,\Gamma\vdash t:A}\circ(\Id\times\sem{\vdash r:\Psi}\times\Id)$.
 Then, by taking $\Gamma=\Gamma'=\emptyset$, we get the result stated by the lemma.
 
 We proceed by induction on the derivation of $\Gamma',x:\Psi,\Gamma\vdash t:A$.
  In this proof, we write $d=(\Id\times\sigma\times\Id)\circ(\Id\times\delta)$.
  Also, we take the rules $\alpha_I$ and $+_I$ with $m=1$, the generalization is straightforward.
  \begin{itemize}
  \item $\vcenter{\infer{\Delta^\B,x:\Psi\vdash x:\Psi}{}}$\qquad
    By Lemma~\ref{lem:W}, $\sem{\Delta^\B\vdash r:\Psi}=\sem{\vdash
      r:\Psi}\circ{!}$. Hence,
    \begin{center}
      \begin{tikzcd}[row sep=10pt]
        \Delta^\B\ar[dd,dashed,"\approx"] \arrow[r,"{!}"] &\One\arrow[r,"r"] & \Psi \\
        & &\One\times\Psi\approx\Psi\ar[u,"\Id"]\\
        \Delta^\B\times\One\ar[ruu,dotted,"{!}"]\ar[rr,"\Id\times r"] && \Delta^\B\times\Psi\ar[u,"{!}\times\Id"]
      \end{tikzcd}
    \end{center}
    This diagram commutes by the naturality of the projection.

  \item
    $\vcenter{\infer{\Delta^\B,x:\B^n\vdash\z:SA}{}}$\qquad
    \begin{center}
      \begin{tikzcd}
	\Delta^\B\ar[d,dashed,"\approx"]\ar[r,"{!}"]&\One\ar[r,"\hat{\vec 0}"] &USA\\
	\Delta^\B\times\One\ar[r,"\Id\times r"]&\Delta^\B\times\B^n\ar[u,"{!}"]
      \end{tikzcd}
    \end{center}
    This diagram commutes by the naturality of the projection.

  \item The cases $\Delta^\B,x:\B^n\vdash\ket 0:\B$ and $\Delta^B,x:\B^n\vdash\ket 1:\B$ are analogous to the previous case.

  \item
    $\vcenter{\infer{\Gamma,x:\Psi,\Gamma\vdash\alpha.t:SA}{\Gamma',x:\Psi,\Gamma\vdash
        t:SA}}$
    
	\begin{center}
	  \begin{tikzcd}[column sep=1cm]
	    \Gamma'\times\Gamma\ar[d,dashed,"\approx"]\ar[r,"(r/x)t"] & USA\ar[r,"U(\lambda)"] &  U(SA\otimes\I) \ar[r,"U(\Id\otimes\alpha)"] & U(SA\otimes\I)\ar[r,"U(\lambda^{-1})"] & USA \\
	    \Gamma'\times\One\times\Gamma\ar[r,"\Id\times r\times\Id"] &\Gamma'\times\Psi\times\Gamma\ar[u,"t"] &&&
	  \end{tikzcd}
	\end{center}
    This diagram commutes by the induction hypothesis.

  \item $\vcenter{\infer{\Gamma',x:\Psi,\Gamma,\Delta,\Xi^\B\vdash t+u:SA}{\Gamma',x:\Psi,\Gamma,\Xi^\B\vdash t:SA & \Delta,\Xi^\B\vdash u:SA}}$
    \begin{center}
      \begin{tikzcd}[column sep=5mm]
        & USA &USA\times USA\ar[l,swap,"+"]  \\
        \Gamma'\times\Gamma\times\Delta\times\Xi^\B\ar[d,dashed,"\approx"]\ar[r,"d"]&\Gamma'\times\Gamma\times\Xi^\B\times\Delta\times\Xi^\B\ar[r,"(r/x)t\times u"] & USA\times USA\ar[u,"\Id"]\\
        \Gamma'\!\times\!\One\!\times\!\Gamma\!\times\!\Delta\!\times\Xi^\B \ar[r,"\Id\times r\times\Id"] &\Gamma'\!\times\!\Psi\!\times\!\Gamma\!\times\!\Delta\!\times\!\Xi^\B\ar[r,"d"]&\Gamma'\times\Psi\times\Gamma\times\Xi^\B\times\Delta\times\Xi^\B\ar[u,"t\times u"] 
      \end{tikzcd}
    \end{center}
    This diagram commutes by the induction hypothesis.
    
    If $x\in FV(u)$ or $x\in FV(u)\cap FV(t)$ the cases are analogous.

  \item $\vcenter{\infer{\Gamma',x:\Psi,\Gamma\vdash t:SA}{\Gamma',x:\Psi,\Gamma\vdash t:A}}$
    \begin{center}
      \begin{tikzcd}
        \Gamma'\times\Gamma\ar[r,"(r/x)t"]\ar[d,dashed,"\approx"] & A \ar[r,"\eta"] & USA\\
        \Gamma'\times\One\times\Gamma\ar[r,"\Id\times r\times\Id"] &\Gamma'\times\Psi\times\Gamma\ar[u,"t"]
      \end{tikzcd}
    \end{center}
    This diagram commutes by the induction hypothesis.

  \item $\vcenter{\infer{\Gamma',x:\Psi,\Gamma\vdash\pi_j t:\B^j\times S\B^{n-j}}{\Gamma',x:\Psi,\Gamma\vdash t:S\B^n }}$
    \begin{center}
      \begin{tikzcd}
        \Gamma'\times\Gamma\ar[r,"(r/x)t"]\ar[d,dashed,"\approx"] & US\B^n \ar[r,"\pi_j"] &
        D(\B^j\times S\B^{n-j})\\
        \Gamma'\times\One\times\Gamma\ar[r,"\Id\times r\times\Id"] &\Gamma'\times\Psi\times\Gamma\ar[u,"t"]
      \end{tikzcd}
    \end{center}
    This diagram commutes by the induction hypothesis.

  \item $\vcenter{\infer{\Gamma',x:\Psi,\Gamma\vdash\ite{}ts:\B\Rightarrow A}{\Gamma',x:\Psi,\Gamma\vdash t:A &
        \Gamma',x:\Psi,\Gamma\vdash s:A}}$

    \begin{center}
      \begin{tikzcd}
        \Gamma'\times\Gamma\ar[d,dashed,"\approx"]\ar[r,"(r/x)G"]   & \lbrack\B,A\rbrack\\
        \Gamma'\times\One\times\Gamma\ar[r,"\Id\times r\times\Id"] & \Gamma'\times\Psi\times\Gamma\ar[u,"G"]
      \end{tikzcd}
    \end{center}
    where
    $(r/x)G =
    \mathsf{curry}(\mathsf{uncurry}(f_{(r/x)t,(r/x)s})\circ\mathsf{swap})$ and
    $G =\mathsf{curry}(\mathsf{uncurry}(f_{t,s})\circ\mathsf{swap})$

    By the induction hypothesis, $(r\times\Id)\circ t=(r/x)t$ and
    $(r\times\Id)\circ s=(r/x)s$, hence, $(r\times\Id)\circ
    f_{t,s}=f_{(r/x)t,(r/x)s}$ and so $(r/x)G=(r\times\Id)\circ G$, which makes
    the diagram commute.

  \item $\vcenter{\infer{\Gamma',x:\Psi,\Gamma\vdash \lambda y{:}\Phi.t:\Phi\Rightarrow A}{\Gamma',x:\Psi,\Gamma,y:\Phi\vdash t:A}}$

    \begin{center}
      \begin{tikzcd}
        \Gamma'\times\Gamma\ar[d,dashed,"\approx"]\ar[r,"\eta^\Phi"] &
        \home\Phi{\Gamma'\times\Gamma\times\Phi}\ar[r,"\home{\Id}{(r/x)t}"]\ar[dr,dotted,near
        start,"\home\Id{\Id\times r\times\Id}",swap] &
        \home\Phi A\\ 
        \Gamma'\times\One\times\Gamma\ar[r,"\Id\times r\times\Id"] &\Gamma'\times\Psi\times\Gamma\ar[r,"\eta^\Phi"]
        &\home \Phi{\Gamma'\times\Psi\times\Gamma\times\Phi}\ar[u,"\home\Id t"]
      \end{tikzcd}
    \end{center}
    The dotted arrow divides the diagram in two. The upper part commutes by the
    IH and the functoriality of $\home\Phi{\--}$, while the lower part commutes by the naturality of $\eta^\Phi$.
  \item $\vcenter{\infer{\Delta,\Gamma',x:\Psi,\Gamma,\Xi^\B\vdash tu:A}
      {
        \Delta,\Xi^\B\vdash u:\Phi
        &
        \Gamma',x:\Psi,\Gamma,\Xi^\B\vdash t:\Phi\Rightarrow A
      }}$
    \begin{center}
      \begin{tikzcd}
        && A\\
        \Delta\times\Gamma'\times\Gamma\times\Xi^\B\ar[dd,dashed,near start,"\approx"]\ar[r,"d"]&\Delta\times\Xi^\B\times\Gamma'\times\Gamma\times\Xi^\B\ar[d,dotted,"\Id\times(r/x)t"]\ar[r,"u\times (r/x)t"] &\Phi\times \home\Phi
        A\ar[u,"\varepsilon^\Phi"] \\
        &\Delta\times\home\Phi A\ar[ru,dotted,"u\times\Id"]\\
        \Delta\times\Gamma'\times\One\times\Gamma\times\Xi^\B\ar[r,"\Id\times r\times\Id"] &\Delta\times\Gamma'\times\Psi\times\Gamma\times\Xi^\B\ar[r,"d"]&\Delta\times\Xi^\B\times\Gamma'\times\Psi\times\Gamma\times\Xi^\B\ar[uu,"u\times
        t"]\ar[ul,dotted,"\Id\times t"]
      \end{tikzcd}
    \end{center}
    This diagram commutes by the induction hypothesis and the functoriality of
    the product.

  \item $\vcenter{\infer{\Delta',x:\Psi,\Delta,\Gamma,\Xi^\B\vdash tu:A}
      {
        \Delta',x:\Psi,\Delta,\Xi^\B\vdash u:\Phi
        &
        \Gamma,\Xi^\B\vdash t:\Phi\Rightarrow A
      }}$

    Analogous to previous case. \\

  \item $\vcenter{\infer{\Delta,\Gamma',x:\Psi,\Gamma,\Xi^\B\vdash tu:SA}
      {
        \Delta,\Xi^\B\vdash u:S\Phi
        &
        \Gamma',x:\Psi,\Gamma,\Xi^\B\vdash t:S(\Phi\Rightarrow A)
      }}$
    \begin{center}
      \begin{tikzcd}
        \Delta\times\Gamma'\times\Gamma\times\Xi^\B\arrow[r, dashed, near start,"\approx"] \arrow[d, "d"] & \Delta\times\Gamma'\times\One\times\Gamma\times\Xi^\B\arrow[r, "\Id\times r\times\Id"] & \Delta\times\Gamma'\times\Psi\times\Gamma\times\Xi^\B \arrow[d, "d"] \\
        \Delta\times\Xi^\B\times\Gamma'\times\Gamma\times\Xi^\B \arrow[r,"u\times (r/x)t"] & US\Phi\times US\home\Phi A\arrow[ld, "n"] & \Delta\times\Xi^\B\times\Gamma'\times\Psi\times\Gamma\times\Xi^\B \arrow[l,"u\times t"] \\
        U(S\Phi\otimes S\home\Phi A)\arrow[r, "Um"] & US(\Phi\times\home\Phi A)\arrow[r, "US\varepsilon^\Phi"] & USA            
      \end{tikzcd}
    \end{center}

    This diagram commutes by the induction hypothesis and the functoriality of
    the product.
  \item $\vcenter{\infer{\Delta'x:\Psi,\Delta,\Gamma,\Xi^\B\vdash tu:SA}
      {
        \Delta',x:\Psi,\Delta,\Xi^\B\vdash u:S\Phi
        &
        \Gamma,\Xi^\B\vdash t:S(\Phi\Rightarrow A)
      }}$

    Analogous to previous case.\\

  \item $\vcenter{\infer{\Gamma',x:\Psi,\Gamma,\Delta,\Xi^\B\vdash t\times u:\Phi\times
        \Upsilon}{\Gamma',x:\Psi,\Gamma,\Xi^\B\vdash t:\Phi & \Delta,\Xi^\B\vdash u:\Upsilon}}$
    \begin{center}
      \begin{tikzcd}
        \Gamma'\times\Gamma\times\Delta\times\Xi^\B\ar[d,dashed,near start,"\approx"]\ar[r,"d"]&\Gamma'\times\Gamma\times\Xi^\B\times\Delta\times\Xi^\B\ar[r,"(r/x)t\times
        u"] & \Phi\times \Upsilon\\
        \Gamma'\times\One\times\Gamma\times\Delta\times\Xi^\B
        \ar[r,"\Id\times r\times\Id"]&
        \Gamma'\times\Psi\times\Gamma\times\Delta\times\Xi^\B
        \ar[r,"d"]
        &\Gamma'\times\Psi\times\Gamma\times\Xi^\B\times\Delta\times\Xi^\B\ar[u,"t\times u"]
      \end{tikzcd}
    \end{center}
    This diagram commutes by the induction hypothesis and coherence results.

  \item $\vcenter{\infer{\Gamma,\Delta',x:\Psi,\Delta,\Xi^\B\vdash t\times u:\Phi\times
        \Upsilon}{\Gamma,\Xi^\B\vdash t:\Phi & \Delta',x:\Psi,\Delta,\Xi^\B\vdash u:\Upsilon}}$

    Analogous to previous case.\\

  \item $\vcenter{\infer{\Gamma',x:\Psi,\Gamma\vdash\head\ t:\B}{\Gamma',x:\Psi,\Gamma\vdash t:\B^n}}$
    \begin{center}
      \begin{tikzcd}
        \Gamma'\times\Gamma\ar[d,dashed,near start,"\approx"]\ar[r,"(r/x)t"]
        &\B^n\ar[r,"\head"] &\B\\
        \Gamma'\times\One\times\Gamma\ar[r,"\Id\times r\times\Id"] &\Gamma'\times\Psi\times\Gamma\ar[u,"t"]
      \end{tikzcd}
    \end{center}
    This diagram commutes by the induction hypothesis.

  \item $\vcenter{\infer{\Gamma',x:\Psi,\Gamma\vdash\tail\ t:\B^{n-1}}{\Gamma',x:\Psi,\Gamma\vdash t:\B^n}}$
    \begin{center}
      \begin{tikzcd}
        \Gamma'\times\Gamma\ar[d,dashed,near start,"\approx"]\ar[r,"(r/x)t"]
        &\B^n\ar[r,"\tail"] &\B^{n-1}\\
        \Gamma'\times\One\times\Gamma\ar[r,"\Id\times r\times\Id"] &\Gamma'\times\Psi\times\Gamma\ar[u,"t"]
      \end{tikzcd}
    \end{center}
    This diagram commutes by the induction hypothesis.
  \item $\vcenter{\infer{\Gamma',x:\Psi,\Gamma\vdash\Uparrow_r t:S(\Phi\times
        \Upsilon)}{\Gamma',x:\Psi,\Gamma\vdash t:S(S\Phi\times \Upsilon)}}$
    \begin{center}
      \begin{tikzcd}
        \Gamma'\times\Gamma\ar[d,dashed,near start,"\approx"]\ar[r,"(r/x)t"] & US(US\Phi\times
        \Upsilon)\ar[r,"U(\Id\times n)"] & US(US\Phi\times US(\Upsilon))\ar[r,"U(n)"]
        &US(U(S\Phi\otimes S(\Upsilon))\ar[d,"USU(m)"]\\
        \Gamma'\times\One\times\Gamma\ar[r,"\Id\times r\times\Id"]
        &\Gamma'\times\Psi\times\Gamma\ar[u,"t"]
        & US(\Phi\times \Upsilon)
        & USUS(\Phi\times \Upsilon)\ar[l,"\mu"]
      \end{tikzcd}
    \end{center}
    This diagram commutes by the induction hypothesis.

  \item $\vcenter{\infer{\Gamma',x:\Psi,\Gamma\vdash\Uparrow_\ell t:S(\Phi\times \Upsilon)}{\Gamma',x:\Psi,\Gamma\vdash t:S(\Phi\times S(\Upsilon))}}$
    
    Analogous to previous case.\\

  \item 
    $\vcenter{\infer{\Gamma',x:\Psi,\Gamma\vdash\bigparallel_{i=1}^n\proba{p_i}t_i:A}{\Gamma',x:\Psi,\Gamma\vdash
        t_i:A & i=1,\dots,n & \sum_{i=1}^n\mathtt{p_i}=1}}$

    \begin{center}
      \begin{tikzcd}[column sep=large]
        \Gamma'\times\Gamma\ar[r,"\delta"]\ar[d,dashed,"\approx"] &
        (\Gamma'\times\Gamma)^n\ar[r,"\prod_{i=1}^n(r/x)t_i"] &
        A^n\ar[r,"d_{\{{p_i}\}_i}"] & D(A)\\
        \Gamma'\times\One\times\Gamma\ar[r,"\Id\times r\times\Id"]
        &\Gamma'\times\Psi\times\Gamma\ar[r,"\delta"]
        &(\Gamma'\times\Psi\times\Gamma)^n\ar[u,"\prod_{i=1}^nt_i"]
      \end{tikzcd}
    \end{center}
    This diagram commutes by the induction hypothesis.
    \qedhere
   \end{itemize}
 \end{proof}
\xrecap{Theorem}{Soundness}{thm:soundness}{
  If $\vdash t:A$, and $t\lra r$,
  then
  $\sem{\vdash t:A} = \sem{\vdash r:A}$.
}
\begin{proof}
  By induction on the rewrite relation, using the first derivable type for each
  term.
  We take the rules $\alpha_I$ and $+_I$ with $m=1$, the generalization is straightforward.
  \begin{itemize}
  \item \textbf{\rcomm} $(t+r)=(r+t)$. We have 
    \[
      \vcenter
      {
        \infer{\vdash (t+r):SA}
        {
          \vdash t:SA & \vdash r:SA 
        }
      }
      \qquad\textrm{and}\qquad
      \vcenter
      {
        \infer{\vdash (r+t):SA}
        {
          \vdash r:SA & \vdash t:SA
        }
      }
    \]
    Then
    \begin{center}
      \begin{tikzcd}
        \One\ar[r,dashed,"\approx"] & \One^2\ar[r,bend right,"t\times
        r",swap]\ar[r,bend left,"r\times t"] & USA^2\ar[r,"\Id"] & USA\times
        USA\ar[r,"+"] & USA
      \end{tikzcd}
    \end{center}
    This diagram commutes by the commutativity of sum in $SA$ as vector space.

  \item \textbf{\rassoc} $((t+r)+s)=(t+(r+s))$. We have
    \[
      \vcenter{
        \infer{\vdash((t+r)+s):SA}
        {
          \infer{\vdash(t+r):SA}{\vdash t:SA &\vdash r:SA}
          &
          \vdash s:SA
        }
      }
      \qquad\textrm{and}\qquad
      \vcenter{
        \infer{\vdash(t+(r+s)):SA}
        {
          \vdash t:SA
          &
          \infer{\vdash (r+s):SA}{\vdash r:SA & \vdash s:SA}
        }
      }
    \]
    Then
    \begin{center}
      \begin{tikzcd}
        \One\ar[r,dashed,"\approx"] & \One^3\ar[r,"t\times
        r\times s"] & USA^3\ar[d,"g_0\times\Id"]\ar[r,"\Id\times g_0"]& USA^3\ar[d,"\Id\times +"] & USA\\
        &&USA^3\ar[r,"+\times\Id"]& USA^2\ar[r,"\Id"] & USA^2\ar[u,"+"]
      \end{tikzcd}
    \end{center}
    This diagram commutes by the associativity of sum in $SA$ as vector space.
  \item \textbf{($\beta_b$)} If $b$ has type $\B^n$ and $b\in\tbasis$, then $(\lambda x{:}{\B^n}.t)b\lra (b/x)t$.
    We have
    \[
      \vcenter{\infer{\vdash(\lambda x{:}{\B^n}.t)b:A}
        {
          \infer{\vdash\lambda x{:}{\B^n}.t:\B^n\Rightarrow A}
          {x:{\B^n}\vdash t:A}
          &
          \vdash b:\B^n
        }}
      \qquad\textrm{and}\qquad
      \vcenter{\infer{\vdash (b/x)t:A}{}}
    \]
    Then
    \begin{center}
      \begin{tikzcd}
        \One^2\ar[r,"b\times\eta^{\B^n}"]&\B^n\times\home{\B^n}{\One\times\B^n}\approx\B^n\times\home{\B^n}{\B^n}\ar[r,"\Id\times\home\Id t"]&\B^n\times\home{\B^n}A\ar[d,"\varepsilon"]\\
        \One\ar[u,dashed,"\approx"]\ar[rr,"(b/x)t"]& & A
      \end{tikzcd}
    \end{center}
    This diagram commutes using Lemma~\ref{lem:substitution}.

  \item \textbf{($\beta_n$)} If $u$ has type $S\Psi$, then $(\lambda x{:}{S\Psi}.t)u\lra  (u/x)t$.
    We have
    \[
      \vcenter{
        \infer{\vdash(\lambda x{:}{S\Psi}.t)u:A}
        {
          \infer{\vdash\lambda x{:}{S\Psi}.t:S\Psi\Rightarrow A}
          {x:{S\Psi}\vdash t:A}
          &
          \vdash u:S\Psi
        }
      }
      \qquad\textrm{and}\qquad
      \vcenter{
        \infer{\vdash (b/x)t:A}{}
      }
    \]
    Then
    \begin{center}
      \begin{tikzcd}
        \One^2\ar[rr,"u\times\eta^{US\Psi}"]&&US\Psi\times\home{US\Psi}{\One\times
          US\Psi}\approx US\Psi\times\home{US\Psi}{US\Psi}\ar[r,"\Id\times\home\Id t"]&S\Psi\times\home{S\Psi}A\ar[d,"\varepsilon"]\\
        \One\ar[u,dashed,"\approx"]\ar[rrr,"(b/x)t"]& && A
      \end{tikzcd}
    \end{center}
    This diagram commutes using Lemma~\ref{lem:substitution}.

  \item \textbf{\rlinr} If $t$ has type $\B^n\Rightarrow A$, then
    $t(u+v)\lra tu+tv$.
    We have
    \[
      \infer{\vdash t(u+v):SA}
      {
        \infer{\vdash t:S(\B^n\Rightarrow A)}{\vdash t:\B^n\Rightarrow A}
        &
        \infer{\vdash u+v:S\B^n }
        {
          \vdash u:S\B^n  & \vdash v:S\B^n 
        }
      }
    \]
    and
    \[
      \infer{\vdash tu+tv:SA}
      {
        \infer{\vdash tu:SA}
        {
          \infer{\vdash t:S(\B^n\Rightarrow A)}{\vdash t:\B^n\Rightarrow A} & \vdash u:S\B^n 
        }
        &
        \infer{\vdash tv:SA}
        {
          \infer{\vdash t:S(\B^n\Rightarrow A)}{\vdash t:\B^n\Rightarrow A}  & \vdash v:S\B^n 
        }
      }
    \]
    \begin{center}
      \begin{tikzcd}
        US\B^n \times US([\B^n,A])\ar[r,"n"] 
        & U(S\B^n \otimes S([\B^n,A]))\ar[d,"U(m)"] & USA\\
        (US\B^n)^2\times
        US(\home{\B^n}A)\ar[u,"+\times\Id"]&US(\B^n\times[\B^n,A])\ar[ur,"US(\varepsilon^{\B^n})"]&
        USA^2\ar[u,"+"]\\
        (US\B^n)^2\times[\B^n,A]\ar[u,"g_0\times \eta"] & & USA^2\ar[u,"g_0"]\\
        \One^3\ar[u,"u\times v\times t"] & & US(\B^n\times\home{\B^n}A)^2\ar[u,"US(\varepsilon^{\B^n})^2"]\\
        \One\ar[d,dashed,"\approx"]\ar[u,dashed,"\approx"]& &
        U(S\B^n \otimes S(\home{\B^n}A))^2\ar[u,"U(m)^2"]\\
        \One^4\ar[r,"u\times t\times v\times t"]
        &(US\B^n \times\home{\B^n}A)^2\ar[r,"(\Id\times\eta)^2"] &
        (US\B^n \times US(\home{\B^n}A))^2\ar[u,"n^2"]
      \end{tikzcd}
    \end{center}
    The mappings are as follows:

    \noindent
    $\ast\mapsto(\ast,\ast,\ast)\mapsto(u,v,t)\mapsto(u,v,t)\mapsto (u+v,t)\mapsto
    (u+v)\otimes t = u\otimes t+v\otimes t\mapsto (u,t)+(v,t)\mapsto t(u)+t(v)$

    \noindent
    $
    \ast\mapsto(\ast,\ast,\ast,\ast)\mapsto(u,t,v,t)\mapsto(u,t,v,t)\mapsto
    (u\otimes t,v\otimes t)\mapsto (u,t,v,t)\mapsto (t(u),t(v))\mapsto
    (t(u),t(v))\mapsto t(u)+t(v)
    $  

  \item \textbf{\rlinscalr} If $t$ has type $\B^n\Rightarrow A$, then
    $t(\alpha.u)\lra \alpha.(tu)$. We have
    \[
      \vcenter{
        \infer{\vdash t(\alpha.u):SA}
        {
          \infer{\vdash t:S(\B^n\Rightarrow A)}{\vdash t:\B^n\Rightarrow A}
          &
          \infer{\vdash\alpha.u:S\B^n }{\vdash u:S\B^n }
        }
      }
      \qquad\textrm{and}\qquad
      \vcenter{
        \infer{\vdash\alpha.(tu):SA}
        {
          \infer{\vdash tu:SA}
          {
            \infer{\vdash t:S(\B^n\Rightarrow A)}{\vdash t:\B^n\Rightarrow A}
            &
            \vdash u:S\B^n 
          }
        }
      }
    \]
    Then
    \begin{center}
      \begin{tikzcd}[column sep=1.2cm]
        U(S\B^n \otimes\I)\times
        US\home{\B^n}A\ar[r,"U(\Id\otimes\alpha)\times\Id"]&U(S\B^n \otimes\I)\times
        US\home{\B^n}A\ar[r,"U\lambda^{-1}\times\Id"]&US\B^n \times US\home{\B^n}A\ar[d,"n"]\\
        US\B^n \times\home{\B^n}A\ar[u,"U\lambda\times\eta"]&&U(S\B^n \otimes
        S\home{\B^n}A)\ar[d,"Um"]\\
        \One^2\ar[u,"u\times t"]\ar[d,"u\times t"] & &
        US(\B^n\times\home{\B^n}A)\ar[d,"US\varepsilon^{\B^n}"]\\
        US\B^n \times\home{\B^n}A \ar[d,"\Id\times\eta"]
        &U(SA\otimes\I)\ar[r,"U\lambda^{-1}"]& USA\\
        US\B^n \times US\home{\B^n}A\ar[d,"n"] &U(SA\otimes\I)\ar[u,"U(\Id\otimes\alpha)"]&\\        
        U(S\B^n \otimes S\home{\B^n}A)\ar[r,"Um"]&
        US(\B^n\times\home{\B^n}A)\ar[r,"US\varepsilon^{\B^n}"]& USA\ar[ul,"U\lambda",swap]
      \end{tikzcd}
    \end{center}
    The mappings are as follows:
    
    \noindent
    $
    (\ast,\ast)\mapsto(u,t)\mapsto(u\otimes
    1,t)\mapsto(u\otimes\alpha,t)\mapsto(\alpha.u,t)\mapsto\alpha.u\otimes
    t=\alpha.(u\otimes t)\mapsto\alpha.(u,t)\mapsto\alpha.t(u)
    $

    \noindent
    $
    (\ast,\ast)\mapsto (u,t)\mapsto (u,t)\mapsto u\otimes t\mapsto (u,t)\mapsto
    t(u)\mapsto t(u)\otimes 1\mapsto t(u)\otimes\alpha\mapsto \alpha.t(u)
    $

   \item \textbf{\rlinzr} If $t$ has type $\B^n\Rightarrow A$, then
    $t\z[\B^n]\lra \z$. We have
    \[
      \vcenter{
        \infer{\vdash t\z[\B^n]:SA}
        {
          \infer{\vdash t:S(\B^n\Rightarrow A)}{\vdash t:\B^n\Rightarrow A}
          &
          \infer{\vdash \z[\B^n]:S\B^n }{}
        }
      }
      \qquad\textrm{and}\qquad
      \vcenter{\infer{\vdash\z:SA}{}}
    \]
    Then
    \begin{center}
      \begin{tikzcd}
        US\B^n \times\home{\B^n}A\ar[r,"\Id\times \eta"]& US\B^n \times US(\home{\B^n}A)\ar[r,"n"] & U(S\B^n \otimes S(\home{\B^n}A))\ar[d,"U(m)"]\\
        \One^2\ar[u,"\hat{\vec 0}\times t"] & &
        US(\B^n\times\home{\B^n}A)\ar[d,"US(\varepsilon^{\B^n})"]\\
        \One\ar[rr,"\lambda x.\mathbf{0}"]\ar[u,dashed,"\approx"] & &USA
      \end{tikzcd}
    \end{center}
    The mappings are as follows:

    \noindent
    $
    \ast\mapsto(\ast,\ast)\mapsto(\vec 0,t)\mapsto(\mathbf
    0,t)\mapsto\vec 0\otimes
    t=\vec 0\mapsto\vec 0\mapsto\vec 0
    $

    \noindent
    $
    \ast\mapsto\vec 0
    $

  \item \textbf{\rlinl} $(t+u)v\lra (tv+uv)$. We have
    \[
      \vcenter{\infer{\vdash(t+u)v:SA}
        {
          \infer{\vdash(t+u):S(\Psi\Rightarrow A)}
          {
            \vdash t:S(\Psi\Rightarrow A)
            &
            \vdash u:S(\Psi\Rightarrow A)
          }
          &
          \vdash v:S\Psi
        }}
    \]
    and
    \[
      \vcenter{\infer{\vdash(tv+uv):SA}
        {
          \infer{\vdash tv:SA}{\vdash t:S(\Psi\Rightarrow A) & \vdash v:S\Psi}
          &
          \infer{\vdash uv:SA}{\vdash u:S(\Psi\Rightarrow A) & \vdash v:S\Psi}
        }}
      \]
    Then
    \begin{center}
      \begin{tikzcd}
        US\Psi\times US(\home\Psi A)^2\ar[r,"\Id\times g_0"]& US\Psi\times
        US(\home\Psi A)^2\ar[r,"\Id\times +"]& US\Psi\times US(\home\Psi A)\ar[d,"n"]\\
        \One^3\ar[u,"v\times t\times u"]&&U(S\Psi\otimes S(\home\Psi A))\ar[d,"U(m)"]\\
        \One\ar[d,dashed,"\approx"]\ar[u,dashed,"\approx"] && US(\Psi\times\home\Psi
        A)\ar[d,"US(\varepsilon^\Psi)"]\\
        \One^4\ar[d,"v\times t\times v\times u"]&&USA\\
        (US\Psi\times US(\home\Psi A))^2\ar[d,"n^2"] && USA^2\ar[u,"+"]\\
        U(S\Psi\otimes S(\home\Psi A))^2\ar[rr,"U(m)^2"] &&
        US(\Psi\times\home\Psi A)^2\ar[u,"US(\varepsilon^\Psi)^2"]
      \end{tikzcd}
    \end{center}
    The mappings are as follows:

    \noindent
    $\ast\mapsto(\ast,\ast,\ast)\mapsto (v,t,u)\mapsto (v,t,u)\mapsto
    (v,t+u)\mapsto v\otimes (t+u)=v\otimes t+v\otimes u\mapsto
    (v,t)+(v,u)\mapsto t(v)+u(v)$

    \noindent
    $\ast\mapsto(\ast,\ast,\ast,\ast)\mapsto (v,t,v,u)\mapsto (v\otimes
    t,v\otimes u)\mapsto (v,t,v,u)\mapsto (t(v),u(v))\mapsto t(v)+u(v)$

  \item \textbf{\rlinscall} $(\alpha.t)u\lra \alpha.(tu)$. We have
    \[
      \vcenter{
        \infer{\vdash(\alpha.t)u:SA}
        {
          \infer{\vdash\alpha.t:S(\Psi\Rightarrow A)}
          {
            \vdash t:S(\Psi\Rightarrow A)
          }
          &
          \vdash u:S\Psi
        }
      }
      \qquad\textrm{and}\qquad
      \vcenter{
        \infer{\vdash\alpha.(tu):SA}
        {
          \infer{\vdash tu:SA}{\vdash t:S(\Psi\Rightarrow A) & \vdash u:S\Psi}
        }
      }
    \]
    Then
    \begin{center}
      \begin{tikzcd}[column sep=15mm]
        US\Psi\times U(S(\home\Psi A)\otimes\I)\ar[r,"\Id\times U(\lambda^{-1})"]
        &US\Psi\times US(\home\Psi A)\ar[r,"n"] &U(S\Psi\otimes S(\home\Psi A))\ar[d,"U(m)"]\\
        US\Psi\times U(S(\home\Psi A)\otimes\I)\ar[u,"\Id\times U(\Id\otimes\alpha)"] & & US(\Psi\times\home\Psi A)\ar[d,"S(\varepsilon^\Psi)"]\\
        US\Psi\times US(\home\Psi A)\ar[u,"\Id\times U(\lambda)"] &U(SA\otimes\I)\ar[r,"U(\lambda^{-1})"] & USA\\
        \One^2\ar[d,"u\times t"]\ar[u,"u\times t"] & U(SA\otimes\I)\ar[u,"\Id\times\alpha"] &USA\ar[l,"U(\lambda)"]\\
        US\Psi\times US(\home\Psi A)\ar[r,"n"] & U(S\Psi\otimes S(\home\Psi A))\ar[r,"U(m)"]& US(\Psi\times\home\Psi A)\ar[u,"US(\varepsilon^\Psi)"]
      \end{tikzcd}
    \end{center}
    The mappings are as follows:

    \noindent
    $(\ast,\ast)\mapsto (u,t)\mapsto(u,t\otimes 1)\mapsto
    (u,t\otimes\alpha)\mapsto (u,\alpha.t)\mapsto
    u\otimes(\alpha.t)=\alpha.(u\otimes t)\mapsto\alpha.(u,t)\mapsto\alpha.t(u)$

    \noindent
    $(\ast,\ast)\mapsto (u,t)\mapsto u\otimes t\mapsto (u,t)\mapsto t(u)\mapsto
    t(u)\otimes 1\mapsto t(u)\otimes\alpha\mapsto\alpha.t(u)$

  \item \textbf{\rlinzl}  $\z[(\B^n\Rightarrow A)]t\lra \z$. We have
    \[\vcenter{
        \infer{\vdash\z[(\B^n\Rightarrow A)]t:SA}
        {
          \infer{\vdash\z[(B^n\Rightarrow A)]:S(\B^n\Rightarrow A)}{}
          &
          \vdash t:S\B^n 
        }
      }
      \qquad\textrm{and}\qquad
      \vcenter{\infer{\vdash\z:SA}{}}
    \]
    Then
    \begin{center}
      \begin{tikzcd}
        US\B^n \times US(\home{\B^n}A)\ar[r,"n"] &U(S\B^n \otimes S(\home{\B^n}A))\ar[d,"U(m)"]\\
        \One^2\ar[u,"t\times\lambda x.\vec
        0"]&US(\B^n\times\home{\B^n}A)\ar[d,"US(\varepsilon^\Psi)"]\\
        \One\ar[u,dashed,"\approx"]\ar[r,"\hat{\vec 0}"] & USA
      \end{tikzcd}
    \end{center}
    The mappings are as follows:

    \noindent
    $\ast\mapsto(\ast,\ast)\mapsto (t,\vec 0)\mapsto t\otimes\mathbf
    0=\vec 0\mapsto\vec 0\mapsto\vec 0$

    \noindent
    $\ast\mapsto\vec 0$

  \item \textbf{\riftrue} $\ite{\ket 1}tr\lra t$. We have
    \[
      \vcenter{
        \infer{\vdash\ite{\ket 1}tr:A}
        {
          \infer{\vdash\ite{}tr:\B\Rightarrow A}{\vdash t:A & \vdash r:A}
          &
          \infer{\vdash\ket 1:\B}{}
        }
      }
      \qquad\textrm{and}\qquad
      \vcenter{\infer{\vdash t:A}{}}
    \]
    Then
    \begin{center}
      \begin{tikzcd}[column sep=4cm]
        \One^2\ar[r,"{\lambda x.\ket
          1\times\mathsf{curry}(\mathsf{uncurry}(f_{t,r})\circ\mathsf{swap})}"]
        &\B\times\home\B A\ar[d,"\varepsilon"]\\
        \One\ar[u,dashed,"\approx"]\ar[r,"t"] & A
      \end{tikzcd}
    \end{center}
    Notice that
    $\mathsf{curry}(\mathsf{uncurry}(f_{t,r})\circ\mathsf{swap})$ transforms the arrow
    $\B\xlra{f_{t,r}}\home\One A$ (which is the arrow $\ket 0\mapsto r$, $\ket
    1\mapsto t$) into an arrow $\One\xlra{}\home\B A$, and
    hence, $\hat{\ket 1}\times\mathsf{curry}(\mathsf{uncurry}(f_{t,r})\circ\mathsf{swap})\circ\varepsilon=t$.

  \item \textbf{\riffalse} Analogous to \riftrue.

  \item \textbf{\rhead} If $h\neq u\times v$, and $h\in\tbasis$, $\head\ h\times t\lra h$. We have
    \[
      \vcenter{
        \infer{\vdash\head\ h\times t:\B}{
          \infer{\vdash h\times t:\B^n}{
            \vdash  h:\B & \vdash t:\B^{n-1}
          }
        }
      }
      \qquad\textrm{and}\qquad
      \vdash h:\B
    \]
    Then
    \begin{center}
      \begin{tikzcd}
        \One^2\ar[r,"h\times t"] &\B^n\ar[d,"\head"]\\
        \One\ar[u,dashed,"\approx"]\ar[r,"h"] & \B
      \end{tikzcd}
    \end{center}
    This diagram commutes since $\head$ is just the projection $\pi_\B$.

  \item \textbf{\rtail} If $h\neq u\times v$, and $h\in\tbasis$, $\tail\ h\times t\lra t$. We have
    \[
      \vcenter{
        \infer{\vdash\tail\ h\times t:\B^{n-1}}{
          \infer{\vdash h\times t:\B^n}{
            \vdash  h:\B & \vdash t:\B^{n-1}
          }
        }
      }
      \qquad\textrm{ and }\qquad
      \vdash t:\B^{n-1}
    \]
    Then
    \begin{center}
      \begin{tikzcd}
        \One^2\ar[r,"h\times t"] &\B^n\ar[d,"\tail"]\\
        \One\ar[u,dashed,"\approx"]\ar[r,"t"] & \B^{n-1}
      \end{tikzcd}
    \end{center}
    This diagram commutes since $\tail$ is just the projection $\pi_{\B^{n-1}}$.

  \item \textbf{\rneut} $(\z+t)\lra  t$. We have
    \[
      \vcenter{
        \infer{\vdash\z+t:SA}
        {
          \infer{\vdash\z:SA}{}
          &
          \vdash t:SA
        }
      }
      \qquad\textrm{and}\qquad
      \vcenter{
        $\vdash t:SA$
      }
    \]
    Then
    \begin{center}
      \begin{tikzcd}
        \One^2\ar[r,"\hat{\vec 0}\times t"] & USA^2\ar[r,"g_0"] & USA^2\ar[d,"+"]\\
        \One\ar[u,dashed,"\approx"]\ar[rr,"t"] && USA
      \end{tikzcd}
    \end{center}
    The mappings are as follows:

    \noindent
    $\ast\mapsto(\ast,\ast)\mapsto (\vec 0,t)\mapsto(\vec 0,t)\mapsto t$

    \noindent
    $\ast\mapsto t$

  \item \textbf{\runit} $1.t\lra t$. We have
    \[
      \infer{\vdash 1.t:SA}
      {
        \vdash t:SA
      }
      \qquad\textrm{and}\qquad
      \vdash t:SA
    \]
    Then
    \begin{center}
      \begin{tikzcd}[column sep=15mm]
        USA\ar[r,"U(\lambda)"] &U(SA\otimes\I)\ar[r,"U(\Id\otimes 1)"] & U(SA\otimes\I)\ar[d,"U(\lambda^{-1})"]\\
        \One\ar[u,"t"]\ar[rr,"t"] & & USA
      \end{tikzcd}
    \end{center}
    The mappings are as follows:

    \noindent
    $\ast\mapsto t\mapsto t\otimes 1\mapsto t\otimes 1\mapsto 1.t=t$

    \noindent
    $\ast\mapsto t$

  \item \textbf{\rzeros} Cases:
    \begin{itemize}
    \item 
      If $t:A$ with $A\in\btypes$, $0.t\lra \z$. We have
    \[
      \vcenter{\infer{\vdash 0.t:SA}{\infer{\vdash t:SA}{\vdash t:A}}}
      \qquad\textrm{and}\qquad
      \vcenter{\infer{\vdash\z:SA}{}}
    \]
    Then
    \begin{center}
      \begin{tikzcd}[column sep=15mm]
        A\ar[r,"\eta"] & USA\ar[r,"U\lambda"] & U(SA\otimes\I)\ar[r,"U(\Id\otimes 0)"] & U(SA\otimes\I)\ar[d,"U\lambda^{-1}"]\\
        \One\ar[u,"t"]\ar[rrr,"\hat{\vec 0}"] & & & USA
      \end{tikzcd}
    \end{center}
    The mappings are as follows:

    \noindent
    $
    \ast\mapsto t\mapsto t\mapsto t\otimes 1\mapsto t\otimes 0=\vec 0\mapsto \vec 0
    $

    \noindent
    $\ast\mapsto\vec 0$
    \item
    If $t:SA$ and $t\ntype A$, $0.t\lra \z$. We have
    \[
      \vcenter{\infer{\vdash 0.t:SA}{\vdash t:SA}}
      \qquad\textrm{and}\qquad
      \vcenter{\infer{\vdash\z:SA}{}}
    \]
    Then
    \begin{center}
      \begin{tikzcd}[column sep=15mm]
        USA\ar[r,"U\lambda"] & U(SA\otimes\I)\ar[r,"U(\Id\otimes 0)"] & U(SA\otimes\I)\ar[d,"U\lambda^{-1}"]\\
        \One\ar[u,"t"]\ar[rr,"\hat{\vec 0}"] & & USA
      \end{tikzcd}
    \end{center}
    The mappings are as follows:

    \noindent
    $
    \ast\mapsto t\mapsto t\otimes 1\mapsto t\otimes 0=\vec 0\mapsto \vec 0
    $

    \noindent
    $\ast\mapsto\vec 0$

    \end{itemize}
  \item \textbf{\rzero} $\alpha.\z\lra \z$. We have
    \[
      \vcenter{
        \infer{\vdash\alpha.\z:SA}
        {
          \infer{\vdash\z:SA}{}
        }
      }
      \qquad\textrm{and}\qquad
      \vcenter{
        \infer{\vdash\z:SA}{}
      }
    \]
    Then
    \begin{center}
      \begin{tikzcd}
        USA\ar[r,"U(\lambda)"] &
        U(SA\otimes\I)\ar[r,"U(\Id\otimes\alpha)"] & U(SUSA\otimes\I)\ar[d,"U(\lambda^{-1})"]\\
        \One\ar[u,"\hat{\vec 0}"]\ar[rr,"\hat{\vec 0}"] && USA
      \end{tikzcd}
    \end{center}
    The mappings are as follows:
    
    \noindent
    $
    \ast\mapsto\vec 0\mapsto\vec 0\otimes 1\mapsto\vec 0\otimes\alpha=\vec 0\mapsto\vec 0
    $

    \noindent
    $\ast\mapsto\vec 0$

  \item \textbf{\rprod} $\alpha.(\beta.t)\lra (\alpha\beta).t$. We have
    \[
      \vcenter{
        \infer{\vdash\alpha.(\beta.t):SA}
        {
          \infer{\vdash\beta.t:SA}
          {\vdash t:SA} 
        }
      }
      \qquad\textrm{and}\qquad
      \vcenter{
        \infer{\vdash(\alpha\beta).t:SA}
        {\vdash t:SA}
      }
    \]
    Then
    \begin{center}
      \begin{tikzcd}[column sep=15mm]
        U(SA\otimes\I)\ar[rr,"U(\Id\otimes\beta)"] &&
        U(SA\otimes\I)\ar[r,"U(\lambda^{-1})"] & USA\ar[d,"U(\lambda)"]\\
        USA\ar[u,"U(\lambda)"] &&& U(SA\otimes\I)\ar[d,"U(\Id\otimes\alpha)"]\\
        \One\ar[d,"t"]\ar[u,"t"] &&& U(SA\otimes\I)\ar[d,"U(\lambda^{-1})"]\\
        USA\ar[r,"U(\lambda)"] &
        U(SA\otimes\I)\ar[r,"U(\Id\otimes(\alpha.\beta))"]& U(SA\otimes\Id)\ar[r,"U(\lambda^{-1})"]& USA
      \end{tikzcd}
    \end{center}
    The mappings are as follows:

    \noindent
    $\ast\mapsto t\mapsto t\otimes 1\mapsto t\otimes\beta\mapsto\beta.t\mapsto\beta.t\otimes 1\mapsto\beta.t\otimes\alpha\mapsto\alpha.(\beta.t)=(\alpha.\beta).t$

    \noindent
    $\ast\mapsto t\mapsto\mapsto t\otimes 1\mapsto t\otimes
    (\alpha.\beta)\mapsto (\alpha.\beta).t$
    
  \item \textbf{\rdists} $\alpha.(t+u)\lra \alpha.t+\alpha.u$. We
    have
    \[
      \vcenter{
        \infer{\vdash\alpha.(t+u):SA}
        {
          \infer{\vdash t+u:SA}
          {\vdash t:SA & \vdash u:SA}
        }
      }
      \qquad\textrm{and}\qquad
      \vcenter{
        \infer{\vdash\alpha.t+\alpha.u:SA}
        {
          \infer{\vdash\alpha.t:SA}{\vdash t:SA}
          &
          \infer{\vdash\alpha.u:SA}{\vdash u:SA}
        }
      }
    \]
    Then
    \begin{center}
      \begin{tikzcd}[column sep=13mm]
        USA^2\ar[r,"+"] & USA\ar[r,"U(\lambda)"]& U(SA\otimes\I)\ar[d,"U(\Id\otimes\alpha)"]\\
        USA^2\ar[r,"U(\lambda)^2"]\ar[u,"g_0"] &
        U(SA\otimes\I)^2\ar[d,"U(\Id\otimes\alpha)"]& U(SA\otimes\Id)\ar[d,"U(\lambda^{-1})"]\\
        \One^2\ar[u,"t\times u"] &
        U(SA\otimes\I)^2\ar[dr,"U(\lambda^{-1})^2"]& USA\\
        \One\ar[u,dashed,"\approx"] && USA^2\ar[u,"+"]\\
      \end{tikzcd}
    \end{center}
    The mappings are as follows:

    \noindent
    $(t,u)\mapsto(t,u)\mapsto t+u\mapsto (t+u)\otimes 1\mapsto
    (t+u)\otimes\alpha\mapsto\alpha.t+\alpha.u$

    \noindent
    $(t,u)\mapsto (t\otimes 1,u\otimes
    1)\mapsto(t\otimes\alpha,u\otimes\alpha)\mapsto (\alpha.t,\alpha.u)\mapsto\alpha.t+\alpha.u$

  \item \textbf{\rfact} $(\alpha.t+\beta.t)\lra (\alpha+\beta).t$. We
    have
    \[
      \vcenter{
        \infer{\vdash(\alpha.t+\beta.t):SA}
        {
          \infer{\vdash\alpha.t:SA}{\vdash t:SA}
          &
          \infer{\vdash\beta.t:SA}{\vdash t:SA}
        }
      }
      \qquad\textrm{and}\qquad
      \vcenter{
        \infer{\vdash(\alpha+\beta).t:SA}{\vdash t:SA}
      }
    \]
    Then
    \begin{center}
      \begin{tikzcd}[column sep=20mm]
        U(SA\otimes\I)^2\ar[rr,"U(\Id\otimes\alpha)\times U(\Id\otimes\beta)"] && U(SA\otimes\I)^2\ar[d,"U(\lambda^{-1})^2"]\\
        USA^2\ar[u,"U(\lambda)^2"] && USA^2\ar[d,"g_0"]\\
        \One^2\ar[u,"t^2"] && USA^2\ar[d,"+"]\\
        \One\ar[d,"t"]\ar[u,dashed,"\approx"] && USA\\
        USA\ar[r,"U(\lambda)"] &
        U(SA\otimes\I)\ar[r,"U(\Id\otimes(\alpha+\beta))"] & U(SA\otimes\I)\ar[u,"U(\lambda^{-1})"]
      \end{tikzcd}
    \end{center}
    The mappings are as follows:

    \noindent
    $\ast\mapsto(\ast,\ast)\mapsto(t,t)\mapsto(t\otimes 1,t\otimes
    1)\mapsto(t\otimes\alpha,t\otimes\beta)\mapsto(\alpha.t,\beta.t)\mapsto(\alpha.t,\beta.t)\mapsto(\alpha+\beta).t$

    \noindent
    $\ast\mapsto t\mapsto t\otimes 1\mapsto t\otimes(\alpha+\beta)\mapsto(\alpha+\beta).t$

  \item \textbf{\rfacto} $(\alpha.t+t)\lra (\alpha+1).t$. We
    have
    \[
      \vcenter{
        \infer{\vdash(\alpha.t+t):SA}
        {
          \infer{\vdash\alpha.t:SA}{\vdash t:SA}
          &
          \infer{\vdash t:SA}{\vdash t:SA}
        }
      }
      \qquad\textrm{and}\qquad
      \vcenter{
        \infer{\vdash(\alpha+1).t:SA}{\vdash t:SA}
      }
    \]
    Then
    \begin{center}
      \begin{tikzcd}[column sep=18mm]
        U(SA\otimes\I)\times USA\ar[rr,"U(\Id\otimes\alpha)\times\Id"] &&
        U(SA\otimes\I)\times USA\ar[d,"U(\lambda^{-1})\times\Id"]\\
        USA^2\ar[u,"U(\lambda)\times\Id"] && USA^2\ar[d,"g_0"]\\
        \One^2\ar[u,"t^2"] && USA^2\ar[d,"+"]\\
        \One\ar[d,"t"]\ar[u,dashed,"\approx"] && USA\\
        USA\ar[r,"U(\lambda)"] & U(SA\otimes\I)\ar[r,"U(\Id\otimes(\alpha+1))"] & U(SA\otimes\I)\ar[u,"U(\lambda^{-1})"]
      \end{tikzcd}
    \end{center}
    The mappings are as follows:

    \noindent$\ast\mapsto(\ast,\ast)\mapsto(t,t)\mapsto(t\otimes
    1,t)\mapsto(t\otimes\alpha,t)\mapsto(\alpha.t,t)\mapsto(\alpha.t,t)\mapsto(\alpha+1).t$

    \noindent$\ast\mapsto t\mapsto t\otimes 1\mapsto t\otimes(\alpha+1)\mapsto(\alpha+1).t$

  \item \textbf{\rfactt} $(t+t)\lra 2.t$. We
    have
    \[
      \vcenter{
        \infer{\vdash(t+t):SA}
        {
          \vdash t:SA
          &
          \vdash t:SA
        }
      }
      \qquad\textrm{and}\qquad
      \vcenter{
        \infer{\vdash 2.t:SA}{\vdash t:SA}
      }
    \]
    Then
    \begin{center}
      \begin{tikzcd}[column sep=15mm]
        \One^2\ar[r,"t\times t"] & USA^2\ar[r,"g_0"] & USA^2\ar[d,"+"]\\
        \One\ar[u,dashed,"\approx"]\ar[d,"t"] && USA\\
        USA\ar[r,"U(\lambda)"] & U(SA\otimes\I)\ar[r,"U(\Id\otimes 2)"] & U(SA\otimes\I)\ar[u,"U(\lambda^{-1})"] 
      \end{tikzcd}
    \end{center}
    The mappings are as follows:

    \noindent$\ast\mapsto(\ast,\ast)\mapsto(t,t)\mapsto(t,t)\mapsto 2.t$

    \noindent$\ast\mapsto t\mapsto t\otimes 1\mapsto t\otimes 2\mapsto 2.t$

  \item \textbf{\rdistsumr} $\Uparrow_r ((r+s)\times
    u)\lra \Uparrow_r (r\times u)+\Uparrow_r (s\times u)$.

    We have
    \[
      \vcenter{
        \infer{\vdash\Uparrow_r ((r+s)\times u):S(\Psi\times \Phi)}
        {
          \infer{\vdash(r+s)\times u:S(S\Psi\times \Phi)}
          {
            \infer{\vdash (r+s)\times u:S\Psi\times \Phi}
            {
              \infer{\vdash r+s:S\Psi}{\vdash r:S\Psi & \vdash s:S\Psi}
              &
              \vdash u:\Phi
            }
          }
        }
      }
      \textrm{\ \ \ \ and\ \ \ \ }
      \vcenter{
        \infer{\vdash\Uparrow_r (r\times u)+\Uparrow_r (s\times u):S(\Psi\times \Phi)}
        {
          \infer{\vdash\Uparrow_r (r\times u):S(\Psi\times \Phi)}
          {
            \infer{\vdash r\times u:S(S\Psi\times \Phi)}
            {
              \infer{\vdash r\times u:S\Psi\times \Phi}
              {
                \vdash r:S\Psi
                &
                \vdash u:\Phi
              }
            }
          }
          &
          \infer{\vdash\Uparrow_r (s\times u):S(\Psi\times \Phi)}
          {
            \infer{\vdash s\times u:S(S\Psi\times \Phi)}
            {
              \infer{\vdash s\times u:S\Psi\times \Phi}
              {
                \vdash s:S\Psi
                &
                \vdash u:\Phi
              }
            }
          }
        }
      }
    \]
    Then
    \begin{center}
      \begin{tikzcd}[row sep=10pt]
        US\Psi\times \Phi\ar[r,"\eta"] & US(US\Psi\times \Phi)\ar[r,"U(\Id\times\eta)"] &  US(US\Psi)\times US\Phi\ar[d,"US(n)"]\\
        US\Psi^2\times \Phi\ar[u,"+\times\Id"] && US(U(S\Psi\otimes S\Phi))\ar[d,"USU(m)"]\\
        US\Psi^2\times \Phi\ar[u,"g_0\times\Id"]&& USUS(\Psi\times \Phi)\ar[d,"\mu"]\\
        \One^3\approx\One\approx\One^4\ar[u,"r\times s\times u"]\ar[d,"r\times u\times s\times u"]&& US(\Psi\times \Phi)\\
        (US\Psi\times \Phi)^2\ar[d,"\eta^2"]&& US(\Psi\times \Phi)^2\ar[u,"+"]\\
        US(US\Psi\times \Phi)^2\ar[d,"U(\Id\times\eta)^2"]&& US(\Psi\times \Phi)^2\ar[u,"g_0"]\\
        US(US\Psi\times US\Phi)^2\ar[r,"US(n)^2"] & US(U(S\Psi\otimes S\Phi))^2\ar[r,"USU(m)^2"] & USUS(\Psi\times \Phi)^2 \ar[u,"\mu^2"]
      \end{tikzcd}
    \end{center}

    The mappings are as follows:
    \begin{align*}
      &\ast\mapsto(\ast,\ast,\ast)\mapsto(r,s,u)\mapsto(r,s,u)\mapsto (r+s,u)\mapsto (r+s,u)\mapsto (r+s,u)\\
      &\qquad\mapsto (r+s)\otimes u=(r\otimes u)+(s\otimes u)\mapsto (r,u)+(s,u)\mapsto (r,u)+(s,u)\\
      &\ast\mapsto(\ast,\ast,\ast,\ast)\mapsto (r,u,s,u)\mapsto(r,u,s,u)\mapsto(r,u,s,u)\\
      &\qquad\mapsto (r\otimes u,s\otimes u)\mapsto (r,u,s,u)\mapsto (r,u,s,u)\mapsto(r,u,s,u)\mapsto (r,u)+(s,u)
    \end{align*}
  \item \textbf{\rdistsuml} $\Uparrow_\ell u\times
    (r+s)\lra \Uparrow_\ell (u\times r)+\Uparrow_\ell (u\times s)$.
    Analogous to case $(\mathsf{dist}_r^+)$

  \item \textbf{\rdistscalr} $\Uparrow_r (\alpha.r)\times u\lra \alpha.\Uparrow_r r\times u$.
    We have
    \[
      \vcenter{
        \infer{\vdash\Uparrow_r (\alpha.r)\times u:S(\Psi\times \Phi)}
        {
          \infer{\vdash(\alpha.r)\times u:S(S\Psi\times \Phi)}
          {
            \infer{\vdash (\alpha.r)\times u:S\Psi\times \Phi}
            {
              \infer{\vdash \alpha.r:S\Psi}{\vdash r:S\Psi}
              &
              \vdash u:\Phi
            }
          }
        }
      }
      \qquad\textrm{and}\qquad
      \vcenter{
        \infer{\vdash\alpha.\Uparrow_r r\times u:S(\Psi\times \Phi)}
        {
          \infer{\vdash\Uparrow_r (r\times u):S(\Psi\times \Phi)}
          {
            \infer{\vdash r\times u:S(S\Psi\times \Phi)}
            {
              \infer{\vdash r\times u:S\Psi\times \Phi}
              {
                \vdash r:S\Psi
                &
                \vdash u:\Phi
              }
            }
          }
        }
      }
    \]
    Then
    \begin{center}
      \begin{tikzcd}[column sep=20mm]
        US\Psi\times \Phi\ar[r,"\eta"] &US(US\Psi\times \Phi)\ar[d,"U(\Id\times\eta)"] \\
        U(S\Psi\otimes\I)\times\Phi\ar[u,"U(\lambda^{-1})\times\Id"] & US(US\Psi\times US\Phi)\ar[d,"US(n)"] \\
        U(S\Psi\otimes\I)\times\Phi\ar[u,"U(\Id\otimes\alpha)\times\Id"] & US(U(S\Psi\otimes S\Phi))\ar[d,"USU(m)"]\\
        US\Psi\times\Phi\ar[u,"U(\lambda)\times\Id"] & USUS(\Psi\times \Phi)\ar[d,"\mu"]\\
        \One^2\ar[d,"r\times u"]\ar[u,"r\times u"] & US(\Psi\times \Phi)\\
        US\Psi\times\Phi\ar[d,"\eta"]&U(S(\Psi\times\Phi)\otimes\I)\ar[u,"U(\lambda^{-1})"]\\
        US(US\Psi\times\Phi)\ar[d,"U(\Id\times\eta)"]&U(S(\Psi\times\Phi)\otimes\I)\ar[u,"U(\Id\otimes\alpha)"]\\
        US(US\Psi\times US\Phi)\ar[d,"US(n)"]&US(\Psi\times\Phi)\ar[u,"U(\lambda)"] \\
        US(U(S\Psi\otimes S\Phi))\ar[r,"USU(m)"]& USUS(\Psi\times\Phi)\ar[u,"\mu"]
      \end{tikzcd}
    \end{center}
    The mappings are as follows:

    \noindent
    $(\ast,\ast)\mapsto(r,u)\mapsto(r\otimes
    1,u)\mapsto(r\otimes\alpha,u)\mapsto(\alpha.r,u)\mapsto(\alpha.r,u)\mapsto(\alpha.r,u)\mapsto\alpha.r\otimes
    u\mapsto\alpha.(r,u)\mapsto\alpha.(r,u)$

    \noindent
    $(\ast,\ast)\mapsto(r,u)\mapsto(r,u)\mapsto(r,u)\mapsto r\otimes u\mapsto
    (r,u)\mapsto(r,u)\mapsto(r,u)\otimes 1\mapsto(r,u)\otimes\alpha\mapsto\alpha.(r,u)$

  \item \textbf{\rdistscall}  $\Uparrow_\ell
    u\times(\alpha.r)\lra \alpha.\Uparrow_\ell u\times r$. Analogous to
    case $\rdistscalr$.

  \item \textbf{\rdistzr} If $u$ has type $\Phi$, $\Uparrow_r\z[\Psi]\times
    u\lra\z[(\Psi\times \Phi)]$. We have
    \[
      \vcenter{
        \infer{\vdash\Uparrow_r\z[\Psi]\times u:S(\Psi\times \Phi)}
        {
          \infer{\vdash\z[\Psi]\times u:S(S\Psi\times \Phi)}
          {
            \infer{\vdash\z[\Psi]\times u:S\Psi\times \Phi}
            {
              \infer{\vdash\z[\Psi]:S\Psi}{}
              &
              \vdash u:\Phi
            }
          }
        }
      }
      \qquad\textrm{and}\qquad
      \vcenter{\infer{\vdash\z[(\Psi\times \Phi)]:S(\Psi\times \Phi)}{}}
    \]
    
    Then
    \begin{center}
      \begin{tikzcd}[column sep=10mm]
        \One\approx\One^2\ar[r,"\hat{\vec 0}\times u"]\ar[d,"\hat{\vec 0}"] & US\Psi\times\Phi\ar[r,"\eta"] & US(US\Psi\times\Phi)\ar[r,"US(\Id\times\eta)"] & US(US\Psi\times US\Phi)\ar[d,"USn"]\\
        US(\Psi\times\Phi) &&USUS(\Psi\times\Phi)\ar[ll,"\mu"] & USU(S\Psi\otimes S\Phi)\ar[l,"USUm"]
      \end{tikzcd}
    \end{center}
    The mappings are as follows:

    \noindent$\ast\mapsto(\ast,\ast)\mapsto(\vec 0,u)\mapsto(\vec
    0,u)\mapsto(\vec 0,u)\mapsto\vec 0\otimes u=\vec 0\mapsto\vec 0\mapsto\vec
    0$

    \noindent$\ast\mapsto\vec 0$

  \item \textbf{\rdistzl} If $u$ has type $\Psi$, $\Uparrow_\ell
    u\times\z[\Phi]\lra\z[(\Psi\times\Phi)]$. Analogous to case $\rdistzr$.

  \item \textbf{\rdistcasum} 
    $\Uparrow (t+u)\lra (\Uparrow t+\Uparrow u)$.
    We only give the details for $\Uparrow_r$, the case $\Uparrow_\ell$ is
    analogous.
    \[
      \vcenter{
        \infer{\vdash\Uparrow_r (t+u):S(\Psi\times \Phi)}
        {
          \infer{\vdash t+u:S(S\Psi\times \Phi)}
          {
            \vdash t:S(S\Psi\times \Phi)
            &
            \vdash u:S(S\Psi\times \Phi)
          }
        }
      }
      \qquad\textrm{and}\qquad
      \vcenter{
        \infer{\vdash\Uparrow_r  t+\Uparrow_r u:S(\Psi\times \Phi)}
        {
          \infer{\vdash\Uparrow_r  t:S(\Psi\times \Phi)}
          {\vdash t:S(S\Psi\times \Phi)}
          &
          \infer{\vdash\Uparrow_r  u:S(\Psi\times \Phi)}
          {\vdash u:S(S\Psi\times \Phi)}
        }
      }
    \]
    Then
    \begin{center}
      \begin{tikzcd}[column sep=1.8cm]
        (US(US\Psi\times \Phi))^2\ar[r,"+"] &US(US\Psi\times \Phi)\ar[r,"U(\Id\times\eta)"]& US(US\Psi\times US\Phi)\ar[d,"USn"]\\
        (US(US\Psi\times \Phi))^2\ar[u,"g_0"]\ar[r,"(U(\Id\times\eta))^2"]&(US(US\Psi\times US\Phi))^2\ar[d,"(USn)^2"] & USU(S\Psi\otimes S\Phi)\ar[d,"USUm"]\\
        \One\times\One\ar[u,"t\times u"] &(USU(S\Psi\otimes S\Phi))^2\ar[d,"(USUm)^2"]& USUS(\Psi\times \Phi)\ar[d,"\mu"]\\
        & (USUS(\Psi\times \Phi))^2\ar[d,"\mu^2"]& US(\Psi\times \Phi)\\
        &(US(\Psi\times \Phi))^2\ar[r,"g_0"]& US(\Psi\times \Phi)^2\ar[u,"+"]
      \end{tikzcd}
    \end{center}
    The mappings are as follows. For $i=1,\dots,m$, let $a_i=\sum_{k_i}\gamma_{ik_i}.a_{ik_i}$,
    $t=\sum_{i=1}^n\beta_i(a_i,b_i)$ and $u=\sum_{i=n+1}^m\beta_i(a_i,b_i)$. To
    avoid a more
    cumbersome notation, we only consider the case where
    $\Psi$ and $\Phi$ do not have an $S$ in head position, and we omit the steps
    not modifying the argument.
    \begin{align*}
      (t,u)&\mapsto t+u\mapsto\sum_{i=1}^m\beta_i. a_i\otimes b_i= \sum_{i=1}^m\beta_i.\sum_{k_i}\gamma_{ik_i}.(a_{ik_i}\otimes b_i)\\
           &\mapsto\sum_{i=1}^m\beta_i.\sum_{k_i}\gamma_{ik_i}.(a_{ik_i}, b_i)\mapsto\sum_{i=1}^m\sum_{k_i}\beta_i.\gamma_{ik_i}.(a_{ik_i}, b_i)
             \\
      \\
    (t,u)
    &\mapsto\left(\sum_{i=1}^n\beta_i.a_i\otimes b_i,\sum_{i=n+1}^m\beta_i.a_i\otimes b_i\right)\\
    &\mapsto\left(\sum_{i=1}^n\beta_i.\sum_{k_i}\gamma_{ik_i}.(a_{ik_i},b_i),\sum_{i=n+1}^m\beta_i.\sum_{k_i}\gamma_{ik_i}.(a_{ik_i},b_i)\right)\\
    &\mapsto\left(\sum_{i=1}^n\sum_{k_i}\beta_i.\gamma_{ik_i}.(a_{ik_i},b_i),\sum_{i=n+1}^m\sum_{k_i}\beta_i.\gamma_{ik_i}.(a_{ik_i},b_i)\right)\\
    &\mapsto\sum_{i=1}^m\sum_{k_i}\beta_i.\gamma_{ik_i}.(a_{ik_i},b_i)
    \end{align*}

  \item \textbf{\rdistcascal} $\Uparrow (\alpha.t)\lra \alpha.\Uparrow t$.
    We only give the details for $\Uparrow_r$, the case $\Uparrow_\ell$ is
    similar.
    \[
      \vcenter{
        \infer{\vdash\Uparrow_r (\alpha.t):S(\Psi\times \Phi)}
        {
          \infer{\vdash\alpha.t:S(S\Psi\times \Phi)}
          {
            \vdash t:S(S\Psi\times \Phi)
          }
        }
      }
      \qquad\textrm{and}\qquad
      \vcenter{
        \infer{\vdash\alpha.\Uparrow_r t:S(\Psi\times \Phi)}
        {
          \infer{\vdash\Uparrow_r  t:S(\Psi\times \Phi)}
          {\vdash t:S(S\Psi\times \Phi)}
        }
      }
    \]
    
    Then
    \begin{center}
      \begin{tikzcd}[column sep=2cm]
        U(S(US\Psi\times \Phi)\otimes\I)\ar[r,"U\lambda^{-1}"] & US(US\Psi\times \Phi)\ar[r,"U(\Id\times\eta)"] &US(US\Psi\times US\Phi)\ar[d,"USn"]\\
        U(S(US\Psi\times\Phi)\otimes\I)\ar[u,"U(\Id\otimes\alpha)"] && USU(S\Psi\otimes S\Phi)\ar[d,"USUm"]\\
        US(US\Psi\times\Phi)\ar[u,"U\lambda"]\ar[r,"U(\Id\times\eta)"] & US(US\Psi\times US\Phi)\ar[d,"USn"] &USUS(\Psi\times \Phi)\ar[d,"\mu"]\\
        \One\ar[u,"t"] &USU(S\Psi\otimes S\Phi)\ar[d,"USUm"]& US(\Psi\times \Phi) \\
        US(\Psi\times\Phi)\ar[rd,"U\lambda"]&USUS(\Psi\times\Phi)\ar[l,"\mu"]& U(S(\Psi\times\Phi)\otimes\I)\ar[u,"U\lambda^{-1}"]\\
        &U(S(\Psi\times\Phi)\otimes\I)\ar[ru,"US(\Id\otimes\alpha)"]&
      \end{tikzcd}
    \end{center}
    The mappings are as follows. For $i=1,\dots,m$, let
    $a_i=\sum_{k_i}\gamma_{ik_i}.a_{ik_i}$ and $t=\sum_i\beta_i.(a_i,b_i)$. To
    avoid a more
    cumbersome notation, we only consider the case where
    $\Psi$ and $\Phi$ do not have an $S$ in head position, and we omit the steps
    not modifying the argument.
    \begin{align*}
    t
    &\mapsto t\otimes 1\mapsto t\otimes\alpha=(\sum_i\beta_i.(a_i,b_i))\otimes\alpha=\sum_i\alpha\beta_i.(a_i,b_i)\mapsto\sum_i\alpha\beta_i.(a_i\otimes b_i)\\
      &=\sum_i\alpha\beta_i.\sum_{k_i}\gamma_{ik_i}.(a_{ik_i}\otimes b_i)\mapsto\sum_i\alpha\beta_i.\sum_{k_i}\gamma_{ik_i}.(a_{ik_i},b_i)\mapsto\sum_i\sum_{k_i}\alpha\beta_i\gamma_{ik_i}.(a_{ik_i},b_i)\\
      \\
    t
    &\mapsto\sum_i\beta_i.(a_i\otimes b_i)=\sum_i\beta_i.\sum_{k_i}\gamma_{ik_i}.(a_{ik_i}\otimes b_i)\mapsto\sum_i\beta_i.\sum_{k_i}\gamma_{ik_i}.(a_{ik_i},b_i)\\
    &\mapsto\sum_i\sum_{k_i}\beta_i\gamma_{ik_i}.(a_{ik_i},b_i)\mapsto(\sum_i\sum_{k_i}\beta_i\gamma_{ik_i}.(a_{ik_i},b_i))\otimes 1\mapsto(\sum_i\sum_{k_i}\beta_i\gamma_{ik_i}.(a_{ik_i},b_i))\otimes \alpha \\
    &=\alpha.\sum_i\sum_{k_i}\beta_i\gamma_{ik_i}.(a_{ik_i},b_i)\mapsto\sum_i\sum_{k_i}\alpha\beta_i\gamma_{ik_i}.(a_{ik_i},b_i)
    \end{align*}

  \item \textbf{\rdistcazeror}
    $\Uparrow_r\z[(S(S\Psi)\times\Phi)]\lra\Uparrow_r\z[(S\Psi\times\Phi)]$. We have
    \[
      \vcenter{
        \infer{\vdash\Uparrow_r\z[(SS\Psi\times\Phi)]:S(S\Psi\times\Phi)}
        {
          \infer{\vdash\z[(SS\Psi\times\Phi)]:S(SS\Psi\times\Phi)}{}
        }
      }
      \qquad\textrm{and}\qquad
      \vcenter{
        \infer{\vdash\Uparrow_r\z[(S\Psi\times\Phi)]:S(S\Psi\times\Phi)}
        {
          \infer{\vdash\z[(S\Psi\times\Phi)]:S(SS\Psi\times\Phi)}
          {
            \infer{\vdash\z[(S\Psi\times\Phi)]:S(S\Psi\times\Phi)}{}
          }
        }
      }
    \]

    Then
    \begin{center}
      \begin{tikzcd}[column sep=15mm]
        \One\ar[rr,"\hat{\vec 0}"]\ar[rd,"\hat{\vec 0}"] & & US(US(US\Psi)\times\Phi)\ar[r,"US(\Id\times\eta)"] & US(US(US\Psi)\times US\Phi)\ar[dd,"US(n)"]\\
        &US(US\Psi\times\Phi)\ar[ru,"\eta"] &\\
        &&USUS(US\Psi\times\Phi)\ar[lu,"\mu"] & US(U(S(US\Psi)\otimes S\Phi))\ar[l,"USU(m)"]
      \end{tikzcd}
    \end{center}
    Both mappings start with $\ast\mapsto\vec 0$, and then continue mapping, by
    linearity, to $\vec 0$.

  \item \textbf{\rdistcazerol}
    $\Uparrow_\ell\z[(\Phi\times SS\Psi)]\lra\Uparrow_\ell\z[(\Phi\times S\Psi)]$.
    Analogous to case $\rdistcazeror$.

  \item \textbf{\rcaneutzr} $\Uparrow_r\z[(S\B^n \times\Phi)]\lra\z[(\B^n\times\Phi)]$.
    We have
    \[
      \vcenter{
        \infer{\vdash\Uparrow_r\z[(S\B^n \times\Phi)]:S(\B^n\times\Phi)}
        {
          \infer{\vdash\z[(S\B^n \times\Phi)]:S(S\B^n \times\Phi)}{}
        }
      }
      \qquad\textrm{and}\qquad
      \vcenter{
        \infer{\vdash\z[(\B^n\times\Phi)]:S(\B^n\times\Phi)}{}
      }
    \]

    Then
    \begin{center}
      \begin{tikzcd}[column sep=15mm]
        \One\ar[r,"\lambda x.\mathbf
        0"]\ar[d,"\hat{\vec 0}"]&US(US\B^n \times\Phi)\ar[r,"U(\Id\times\eta)"]&US(US\B^n \times
        US\Phi)\ar[d,"US(n)"]\\
        US(\B^n\times\Phi)&USUS(\B^n\times\Phi)\ar[l,"\mu"] & US(U(S\B^n \otimes S\Phi))\ar[l,"USU(m)"]
      \end{tikzcd}
    \end{center}
    Both mappings start with $\ast\mapsto\vec 0$, and then continue mapping, by
    linearity, to $\vec 0$.

  \item \textbf{\rcaneutzl} $\Uparrow_\ell\z[(\Phi\times S\B^n)]\lra\z[(\Phi\times\B^n)]$.
    Analogous to case \rcaneutzr.
  \item \textbf{\rcaneutr} If $u\in\tbasis$, $\Uparrow_r u\times
    v\lra  u\times v$. We have
    \[
      \vcenter{
        \infer{\vdash\Uparrow_r u\times v:S(\Psi\times \Phi)}
        {
          \infer{\vdash u\times v:S(S\Psi\times \Phi)}
          {
            \infer{\vdash u\times v:S\Psi\times \Phi}
            {
              \infer{\vdash u:S\Psi}{\vdash u:\Psi}
              &
              \vdash v:\Phi
            }
          }
        }
      }
      \qquad\textrm{and}\qquad
      \vcenter{
        \infer{\vdash u\times v:S(\Psi\times \Phi)}
        {
          \infer{\vdash u\times v:\Psi\times \Phi}
          {
            \vdash u:\Psi & \vdash v:\Phi
          }
        }
      }
    \]
    Then
    \begin{center}
      \begin{tikzcd}
        \One\approx\One^2\ar[r,"u\times v"] & \Psi\times
        \Phi\ar[dl,"\eta"]\ar[r,"\eta\times\Id"] & US\Psi\times
        \Phi\ar[r,"\eta"] & US(US\Psi\times \Phi)\ar[d,"U(\Id\times\eta)"]\\
        US(\Psi\times \Phi) & USUS(\Psi\times \Phi)\ar[l,"\mu"] & US(U(S\Psi\otimes S\Phi))\ar[l,"USU(m)"]  & US(US\Psi\times US\Phi)\ar[l,"US(n)"]
      \end{tikzcd}
    \end{center}
    Both mappings are the identity, so we do not give the mappings. Notice that even if $v$ is a linear
    combination, the $\eta$ on $\Phi$ will freeze its linearity by considering
    it as a basis vector in a new vector space $US\Phi$ having $\Phi$ as base.

  \item \textbf{\rcaneutl} If $v\in\tbasis$, $\Uparrow_\ell u\times
    v\lra  u\times v$. Analogous to case $\rcaneutr$.

  \item \textbf{\rproj}
    $\pi_j\ket\psi
    \lra
    \bigparallel\limits_{k=0}^{2^j-1} \proba{p_k} (\ket k\times\ket{\phi_k})$,

    where
    \begin{align*}
      &\ket\psi = \sum\limits_{i=1}^n\may[\alpha_i]\prod\limits_{h=1}^m\ket{b_{hi}}\\
      &\ket k=\ket{b_1}\times\dots\times\ket{b_j}\textrm{ where }b_1\dots b_j\textrm{ is the binary representation of }k\\
      &\ket{\phi_k} = \sum\limits_{i\in T_k}\beta_{ik}.\prod\limits_{h=j+1}^m \ket{b_{hi}}\\
      &\beta_{ik}=\left(\frac{\alpha_i}{\sqrt{\sum\limits_{r\in T_k}|\alpha_r|^2}}\right)\quad\textrm{and}\quad\mathtt{p_k}=\sum\limits_{i\in T_k}\left(\frac{|\alpha_i|^2}{{\sum\limits_{r=1}^n|\alpha_r|^2}}\right)\\
      &\textrm{with }T_k=\{i\leq n\mid\ket{b_{1i}}\times\dots\times\ket{b_{ji}}=\ket k\}.
    \end{align*}

    We have
    \[
      \infer{
        \vdash\pi_j\ket\psi:\B^j\times S\B^{m-j}}
      {
        \infer{
          \vdash\ket\psi:S\B^m
        }
        {
          \infer{\vdash\lbrack\alpha_i.\rbrack\prod\limits_{h=1}^m \ket{b_{h1}}:S\B^m}
          {
            \infer{\vdash\prod\limits_{h=1}^m \ket{b_{h1}}:S\B^m}
            {
              \infer{\vdash\prod\limits_{h=1}^m \ket{b_{h1}}:\B^m}
              {
                \infer{\vdash \ket{b_{11}}:\B}{}
                &
                \dots 
                &
                \infer{\vdash \ket{b_{m1}}:\B}{}
              }
            }
          }
          &
          \dots
          &
          \infer{\vdash\lbrack\alpha_n.\rbrack\prod\limits_{h=1}^m \ket{b_{hn}}:S\B^m}
          {
            \infer{\vdash\prod\limits_{h=1}^m \ket{b_{hn}}:S\B^m}
            {
              \infer{\vdash\prod\limits_{h=1}^m \ket{b_{hn}}:\B^m}
              {
                \infer{\vdash \ket{b_{1n}}:\B}{}
                &
                \dots 
                &
                \infer{\vdash \ket{b_{mn}}:\B}{}
              }
            }
          }
        }
      }
    \]
    and
    \[
      \infer{\vdash\bigparallel\limits_{k=1}^{2^j-1} \proba{p_k}(\ket k\times\ket{\phi_k}):\B^j\times S\B^{m-j}}
      {
        \infer{\vdash \ket{k}\times\ket{\phi_k}:\B^j\times S\B^{m-j}}
        {
          \vdash \ket{k}:\B^j
          &
          \infer{\vdash\ket{\phi_k}:S\B^{m-j}}
          {
            \infer{\vdash\beta_{i_11}.\prod\limits_{h=j+1}^m \ket{b_{hi_1}}:S\B^{m-j}}
            {
              \infer{\vdash\prod\limits_{h=j+1}^m \ket{b_{hi_1}}:S\B^{m-j}}
              {
                \infer{\vdash\prod\limits_{h=j+1}^m \ket{b_{hi_1}}:\B^{m-j}}
                {
                  \infer{\vdash \ket{b_{j+1,i_1}}:\B}{}
                  &\dots&
                  \infer{\vdash \ket{b_{m,i_1}}:\B}{}
                }
              }
            }
            &\dots&
            \infer{\vdash\beta_{i_{{2^j}-1} 2^j-1}.\prod\limits_{h=j+1}^m \ket{b_{hi_{{2^{j}}-1}}}:S\B^{m-j}}
            {
              \infer{\vdash\prod\limits_{h=j+1}^m \ket{b_{hi_{{2^j}-1}}}:S\B^{m-j}}
              {
                \infer{\vdash\prod\limits_{h=j+1}^m \ket{b_{hi_{{2^j}-1}}}:\B^{m-j}}
                {
                  \infer{\vdash \ket{b_{j+1,i_{2^j-1}}}:\B}{}
                  &\dots&
                  \infer{\vdash \ket{b_{m,i_{2^j-1}}}:\B}{}
                }
              }
            }
          }
        }
      }
    \]

    The following diagram, where
    \[
      \Psi = \sem{\vdash\ket\psi:S\B^m}
      \quad\textrm{and}\quad
      P_k = \sem{\vdash\ket k\times\ket{\phi_k}:B^j\times S\B^{m-j}}
    \]
    commutes.
    \begin{center}
      \begin{tikzcd}[column sep=2.5cm]
        \One\approx\One^{2^j}\ar[d,"\Psi"]\ar[r,"P_0\times\dots\times P_{2^j-1}"] & (\B^j\times US\B^{m-j})^{2^j}\ar[d,"d_{{\{{p_k}\}}_k}"]\\
        US\B^m\ar[r,"\pi_j"] & D(\B^j\times US\B^{m-j})
      \end{tikzcd}
    \end{center}
    Indeed
    \[
      \pi_j\circ \Psi
      =\sum_{k=0}^{2^j-1}p_k\chi_{\pi_{jk}\ket\psi}=
      d_{\{p_k\}_k}\circ\left( \prod_{k=0}^{2^j-1} P_k\right)
    \]
  \item \textbf{\rprojz} $\pi_j\z[\B^n]\lra\ket 0^{\times n}$. We have
    \[
      \vcenter{\infer{\vdash\pi_j\z[\B^n]:\B^j\times S\B^{n-j}}{\vdash\z[\B^n]:S\B^n}}
      \qquad\textrm{and}\qquad
      \vcenter{\infer{\vdash\ket 0^n:\B^j\times S\B^{n-j}}
      {
	\infer{\vdash\ket 0^j:\B^j}{
	    \infer{\vdash\ket 0:\B}{}
	    &\cdots &
	    \infer{\vdash\ket 0:\B}{}
	}
	&
	\infer{\vdash\ket 0^{n-j}:S\B^{n-j}}{
	  \infer{\vdash\ket 0^{n-j}:\B^{n-j}}{
	    \infer{\vdash\ket 0:\B}{}
	    &\cdots &
	    \infer{\vdash\ket 0:\B}{}
	  }
	}
    }}
    \]
    Then
    \begin{center}
      \begin{tikzcd}
\One^n\approx\One \arrow[r, "\hat{\vec 0}"] \arrow[d, "\hat{\ket 0}^n"] & US\B^n \arrow[d, "\pi_j"] \\
\B^n \arrow[r, "\Id\times\eta^{n-j}"]                                   & \B^j\times US\B^{n-j}
\end{tikzcd}
    \end{center}
    The mappings are as follows:

    \noindent
    $\ast\mapsto\vec 0\mapsto\ket 0^n$\\
    $\ast\approx\ast^n\mapsto\ket 0^n\mapsto\ket 0^n$

  \item \textbf{Contextual rules} Trivial by composition law.
    \qedhere
  \end{itemize}
\end{proof}

\recap{Lemma}{lem:Adequacy}{
  If $\vdash t:A$ then $t\in\com A$.}
\begin{proof}
  We prove, more generally, that if $\Gamma\vdash t:A$ and $\sigma\vDash\Gamma$,
  then $\sigma t\in\com A$. We proceed by structural induction on the derivation
  of $\Gamma\vdash t:A$. In order to avoid cumbersome notation, we do not take the
  closure by parallelism into account, except when needed. The extension of this
  proof to such a closure is straightforward.
  \begin{itemize}
   \item Let $\Gamma^\B,x:\Psi\vdash x:\Psi$ as a consequence of rule $\tax$. Since $\sigma\vDash\Gamma^\B,x:\Psi$, we have $\sigma x\in\com\Psi$.
   \item $\tax_{\vec 0}$, $\tax_{\ket 0}$ and $\tax_{\ket 1}$ are trivial since,
     by definition $\z[A]\in\com{SA}$, $\ket 0\in\com{\B}$ and $\ket 1\in\com{B}$.
  \item Let $\Gamma\vdash \alpha.t:SA$ as a consequence of $\Gamma\vdash t:SA$
    and rule $\alpha_I$. By the IH, $\sigma t\in\com{SA}$,
    hence, by definition $\alpha.\sigma t=\sigma\alpha.t\in\com{SA}$.
    
  \item Let ${\Gamma,\Delta,\Xi^\B\vdash\pair tu:SA}$ as a consequence of
    $\Gamma,\Xi^\B\vdash t:SA$, $\Delta,\Xi^\B\vdash u:SA$, and rule $+_I$. By
    the IH, $\sigma_1\sigma t,\sigma_2\sigma u\in\com{SA}$,
    hence, by definition $\sigma_1\sigma t+\sigma_2\sigma
    t=\sigma_1\sigma_2\sigma(t+u)\in\com{SA}$.

  \item Let $\Gamma\vdash\pi_j t:\B^j\times S\B^{n-j}$ as a consequence of
    $\Gamma\vdash t:S\B^n$ and rule $S_E$. By the IH, $\sigma
    t\in\com{S\B^n }=\overline{S\com{\B^n}\cup\{\z[\B^n]\}}$. Then, $\sigma
    t\in\overline{S\{\ket 0,\ket 1\}^n}$, so $\sigma
    t\lra^*\sum_i\alpha_i\ket{b_{i1}}\times\dots\times\ket{b_{in}}$, with
    $b_{ij}=0$ or $b_{ij}=1$. Therefore, $\pi_j\sigma
    t\lra^*\pi_j\sum_i\alpha_i\ket{b_{i1}}\times\dots\times\ket{b_{in}}\rightarrow
    \ket{b'_1}\times\dots\ket{b'_j}\times\sum_i\beta_i\ket{b'_{i,{j+1}}}\times\dots\times\ket{b'_{in}}\in\{\ket
    0,\ket 1\}^j\times\overline{\{\ket 0,\ket 1\}^{n-j}}\subseteq\com{\B^j\times
      S\B^{n-j}}$.
    
  \item Let $\Gamma\vdash\ite{}tr:\B\Rightarrow A$ as a consequence of
    $\Gamma\vdash t:A$, $\Gamma\vdash r:A$ and rule $\tif$. By the induction
hypothesis, $\sigma t\in\com A$ and $\sigma r\in\com A$. Hence, for any
$s\in\com\B$, $\ite s{\sigma t}{\sigma r}$ reduces either to $\sigma t$ or to
$\sigma r$, hence it is in $\com A$, therefore, $\ite{}{\sigma t}{\sigma
  r}\in\com{\B\Rightarrow A}$.

  \item Let $\Gamma\vdash\lambda x{:}\Psi.t:\Psi\Rightarrow A$ as a consequence
    of $\Gamma,x:\Psi\vdash t:A$ and rule $\Rightarrow_I$. Let $r\in\com\Psi$.
    Then, $\sigma(\lambda x{:}\Psi.t)r=(\lambda x{:}\Psi.\sigma t)r\rightarrow
    (r/x)\sigma t$. Since $(r/x)\sigma t\vDash\Gamma,x{:}\Psi$, we have, by the
    IH, that $(r/x)\sigma t\in\com A$. Therefore, $\lambda
    x{:}\Psi.t\in\com{\Psi\Rightarrow A}$.

  \item Let $\Delta,\Gamma,\Xi^\B\vdash tu:A$ as a consequence of
    $\Delta,\Xi^\B\vdash u:\Psi$, $\Gamma,\Xi^\B\vdash t:\Psi\Rightarrow A$ and
    rule $\Rightarrow_E$. By the IH, $\sigma_1\sigma
    u\in\com\Psi$ and $\sigma_2\sigma t\in\com{\Psi\Rightarrow A}$. Then, by
    definition, $\sigma_1\sigma t\sigma_2\sigma
    r=\sigma_1\sigma_2\sigma(tr)\in\com A$.
    
  \item Let $\Delta,\Gamma,\Xi^\B\vdash tu:SA$ as a consequence of
    $\Delta,\Xi^\B\vdash u:S\Psi$, $\Gamma,\Xi^\B\vdash t:S(\Psi\Rightarrow A)$
    and rule $\Rightarrow_{ES}$. By the IH $\sigma_1\sigma
    t\in\com{S(\Psi\Rightarrow A)}=\overline{S\com{\Psi\Rightarrow
        A}\cup\{\z[(\Psi\Rightarrow A)]\}}$ and $\sigma_2\sigma
    u\in\com{S\Psi}=\overline{S\com\Psi\cup\{\z[\Psi]\}}$. Cases: 
    \begin{itemize}
      \item[$\ast$] $\sigma_1\sigma t\lra^*\z[(\Psi\Rightarrow A)]$ and
        $\sigma_2\sigma u\rightarrow\z[\Psi]$. Then
        $\sigma_1\sigma_2\sigma (tu)=\sigma_1\sigma t\sigma_2\sigma u\lra^*\z[(\Psi\Rightarrow A)]\z[\Psi]\rightarrow\z[A]\in\com{SA}$.
      \item[$\ast$] $\sigma_1\sigma t\lra^*\z[(\Psi\Rightarrow A)]$ and
        $\sigma_2\sigma u\rightarrow\sum_j\beta_ju_j$, with $u_j\in\com\Psi$. Then
        $\sigma_1\sigma_2\sigma (tu)=\sigma_1\sigma t\sigma_2\sigma u\lra^*\z[(\Psi\Rightarrow A)]\sum_j\beta_ju_j\rightarrow\z[A]\in\com{SA}$.
      \item[$\ast$] $\sigma_1\sigma t\lra^*\sum_i\alpha_it_i$ with $t_i\in\com{\Psi\Rightarrow A}$ and
        $\sigma_2\sigma u\rightarrow\z[\Psi]$. Then
        $\sigma_1\sigma_2\sigma (tu)=\sigma_1\sigma t\sigma_2\sigma u\lra^*\sum_i\alpha_i(t_i\z[\Psi])\lra^*\z[A]\in\com{SA}$.
        
      \item[$\ast$] $\sigma_1\sigma t\lra^*\sum_i\alpha_it_i$ with $t_i\in\com{\Psi\Rightarrow A}$ and
        $\sigma_2\sigma u\rightarrow\sum_j\beta_ju_j$, with $u_j\in\com\Psi$. Then
        $\sigma_1\sigma_2\sigma (tu)=\sigma_1\sigma t\sigma_2\sigma u\lra^*\sum_{ij}\alpha_i\beta_jt_iu_j$
        with $t_iu_j\in\com A$, therefore, $\sigma_1\sigma_2\sigma (tu)\in\com{SA}$.
    \end{itemize}
    
  \item Let $\Gamma\vdash t:SA$ as a consequence of $\Gamma\vdash t:A$ and rule
    $S_I$. By the IH, $\sigma t\in\com A\subseteq S\com
    A\subseteq\com{SA}$.

  \item Let $\Gamma,\Delta,\Xi^\B\vdash t\times u:\Psi\times\Phi$ as a
    consequence of $\Gamma,\Xi^\B\vdash t:\Psi$, $\Delta,\Xi^\B\vdash u:\Phi$
    and rule $\times_I$. By the IH, $\sigma_1\sigma
    t\in\com\Psi$ and $\sigma_2\sigma u\in\com\Phi$, hence, $\sigma_1\sigma
    t\times\sigma_2\sigma u=\sigma_1\sigma_2\sigma (t\times
    u)\in\com\Psi\times\com\Phi\subseteq\com{\Psi\times\Phi}$.

  \item Let $\Gamma\vdash \head~t:\B$ as a consequence of $\Gamma\vdash t:\B^n$
    and rule $\times_{Er}$. By the IH, $\sigma
    t\in\com{\B^n}=\overline{\com\B\times\com{\B^{n-1}}}=\{u\mid
    u\lra^*u_1\times u_2\textrm{ with }u_1\in\com\B\textrm{
      and} u_2\in\com{\B^{n-1}}\}$. Hence, $\sigma(\head\ t)=\head\ \sigma
    t\lra^*\head(u_1\times u_2)\rightarrow u_1\in\com\B$.

  \item Let $\Gamma\vdash \tail~t:\B^{n-1}$ as a consequence of $\Gamma\vdash
    t:\B^n$ and rule $\times_{E\ell}$. By the IH, $\sigma
    t\in\com{\B^n}=\overline{\com\B\times\com{\B^{n-1}}}=\{u\mid
    u\lra^*u_1\times u_2\textrm{ with }u_1\in\com\B\textrm{
      and} u_2\in\com{\B^{n-1}}\}$. Hence, $\sigma(\tail\ t)=\tail\ \sigma
    t\lra^*\tail(u_1\times u_2)\rightarrow u_2\in\com{\B^{n-1}}$.

 \item Let $\Gamma\vdash \Uparrow_rt:S(\Psi\times \Phi)$, as a consequence of $\Gamma\vdash t:S(S\Psi\times \Phi)$ and rule $\Uparrow_r$.
      By the IH, we have that
      $\sigma t\in\com{S(S\Psi\times\Phi)}$. Therefore, $\sigma
      t\in
      \overline{S(\overline{\overline{(S\com\Psi\cup\{\z[\Psi]\})}\times\com\Phi})\!\cup\!\{\z[(S\Psi\times\Phi)]\}}$.
      Cases:
      \begin{itemize}
      \item[$\ast$] $\sigma t\lra^*\z[(S\Psi\times\Phi)]$, then
        $\sigma \Uparrow_r t=\Uparrow_r\sigma
        t\lra^*\Uparrow_r\z[(S\Psi\times\Phi)]\lra\z[(\Psi\times\Phi)]\in\com{S(\Psi\times\Phi)}$.
      \item[$\ast$] Otherwise, $\sigma
        t\in\overline{S(\overline{\overline{(S\com\Psi\cup\{\z[\Psi]\})}\times\com\Phi})}$,
        so $\sigma t\lra^*\sum_i\alpha_i (r_i\times u_i)$ with
        $u_i\in\com\Phi$ and
        $r_i\in\overline{S\com\Psi\cup\{\z[\Psi]\}}$.
        Cases:
        If
        $r_i\lra^*\z[\Psi]$, then $\Uparrow_r r_i\times
          u_i\lra^* \Uparrow_r\z[\Psi]\times u_i\lra
          \z[(\Psi\times\Phi)]\in\com{S(\Psi\times\Phi)}$.
          If $r_i\lra^*\sum_{j=1}^{n_i}\beta_{ij}r_{ij}$, with
          $r_{ij}\in\com\Psi$. Hence, $\Uparrow_r r_i\times
          u_i\lra^*\sum_{j=1}^{n_i}\beta_{ij}\Uparrow_r(r_{ij}\times
          u_i)\in\com{S(\Psi\times\Phi)}$.
        Hence, if all the $r_i$ reduce to $\z[\Psi]$,
        $\sigma\Uparrow_rt=\Uparrow_r\sigma t\lra^*\z[(\Psi\times\Phi)]$.
        Otherwise, let $I$ be the set of index of $r_i$ not reducing to
        $\z[\Psi]$, therefore,
        $\sigma\Uparrow_rt=\Uparrow_r\sigma
        t\lra^*\sum_i\alpha_i(r_i\times
        u_i)\rightarrow\sum_i\alpha_i\Uparrow_r(r_i\times
        u_i)\lra^*\sum_{i\in
          I}\alpha_i\sum_{j=1}^{n_i}\beta_{ij}\Uparrow_r(r_{ij}\times
        u_i)\lra^*\sum_{i\in
          I}\sum_{j=1}^{n_i}\alpha_i\beta_{ij}\Uparrow_r(r_{ij}\times u_i)\in\com{S(\Psi\times\Phi)}$.
      \end{itemize}
    
  \item Let $\Gamma\vdash \Uparrow_\ell t:S(\Psi\times \Phi)$ as a consequence
    of $\Gamma\vdash t:S(\Psi\times S\Phi)$ and rule $\Uparrow_\ell$. Analogous
    to previous case.

  \item Let $\Gamma\vdash\proba{p_1}t_1 \parallel\dots\parallel \proba{p_n}t_n
    :A$ as a consequence of $\Gamma\vdash t_i:A$ and rule $\parallel$. By the
    IH each $\sigma t_i\in\com A$, hence, by definition
    $\sigma(\proba{p_1}t_1\parallel\dots\parallel\proba{p_n}t_n)=\proba{p_1}\sigma
    t_1\parallel\dots\parallel\proba{p_n}\sigma t_n\in\com A$.
    \qedhere
  \end{itemize}
 \end{proof}

\end{document}